\documentclass[11pt, letterpaper]{article}

\usepackage{indentfirst}
\usepackage{amsmath}
\usepackage{amsthm}
\usepackage{color}
\usepackage{eufrak}
\usepackage{verbatim}
\usepackage[makeroom]{cancel}
\usepackage{rotating}
\usepackage{float}
\usepackage{fullpage}
\usepackage[margin=1in]{geometry}
\usepackage{amssymb}
\usepackage{amsthm}
\usepackage{mathrsfs}
\usepackage{epsfig}
\usepackage{graphicx}
\usepackage{indentfirst}
\usepackage{framed}
\usepackage{color}
\usepackage{hyperref}
\usepackage{graphicx}
\usepackage{hyperref}
\usepackage{float}
\usepackage{bm}
\usepackage{cleveref}
\usepackage{enumitem}
\usepackage{mathtools}
\usepackage{mleftright}
\usepackage{thmtools} 
\usepackage{relsize}\usepackage{scalerel}
\usepackage{tikz}
\usepackage{multicol}
\usetikzlibrary{positioning,shapes,shadows,arrows}
\usepackage[labelfont=bf]{caption}
\usepackage[boxruled,linesnumbered,lined,commentsnumbered]{algorithm2e}
\usepackage[noend]{algpseudocode}
\usepackage{tablefootnote}

\usepackage[style=alphabetic,natbib=true,maxnames=99,maxalphanames=99]{biblatex}

\definecolor{corlinks}{RGB}{0,0,150}
\definecolor{cormenu}{RGB}{0,0,150}
\definecolor{corurl}{RGB}{0,0,150}

\hypersetup{
	colorlinks=true,
	urlcolor=corlinks,
	linkcolor=corlinks,
	menucolor=cormenu,
	citecolor=corlinks,
	pdfborder= 0 0 0
}

\newcommand{\numberthis}{\addtocounter{equation}{1}\tag{\theequation}}

\newtheorem{theorem}{Theorem}[section]
\newtheorem{lemma}[theorem]{Lemma}
\newtheorem{fact}[theorem]{Fact}
\newtheorem{claim}[theorem]{Claim}

\newtheorem{proposition}[theorem]{Proposition}

\theoremstyle{definition}

\newtheorem{definition}[theorem]{Definition}

\newcommand{\D}{\mathcal{D}}

\let\phi\varphi

\newcommand{\binset}{\{0,1\}}
\newcommand{\bin}{\binset}
\newcommand{\poly}{\mathsf{poly}}

\newcommand{\supp}{\mathsf{Support}}
\newcommand{\SD}{\Delta_{\mathsf{TV}}} 
\newcommand{\KL}{\mathrm{KL}}
\DeclareMathOperator*\Exp{{\bf E}}
\DeclareMathOperator*\Prob{{\bf Pr}}

\newcommand{\bool}{\mleft\{0,1\mright\}}

\newcommand{\N}{\mathbb{N}}
\newcommand{\Nat}{\N}

\newcommand{\rK}{\mathsf{rK}}
\newcommand{\pK}{\mathsf{pK}}
\newcommand{\K}{\mathsf{K}}

\renewcommand{\P}{\mathsf{P}}
\newcommand{\NP}{\mathsf{NP}}

\DeclarePairedDelimiterX\set[1]\lbrace\rbrace{#1}

\newcommand{\cd}{ \mathsf{cd} }

\newcommand{\Shannon}{\mathrm{H}}

\newcommand{\ccc}{\circ\cdots\circ}
\newcommand{\Enc}{\mathrm{Enc}}
\newcommand{\Dec}{\mathrm{Dec}}
\newcommand{\Next}{\mathsf{Next}}
\newcommand{\Ext}{\mathsf{Ext}}
\newcommand{\lrang}[1]{\langle{#1}\rangle}

\newcommand{\num}{\mathfrak{m}}

\newcommand{\PComp}{\textsc{PComp}}
\newcommand{\PSamp}{\textsc{PSamp}}
\title{Optimal Coding for Randomized Kolmogorov Complexity \\ and Its Applications\vspace{0.2cm}}

\author{Shuichi Hirahara\thanks{National Institute of Informatics, Japan. \texttt{E-mail:~s\_hirahara@nii.ac.jp}} \and 
	Zhenjian Lu\thanks{University of Warwick, UK. \texttt{E-mail:~zhenjian.lu@warwick.ac.uk}} \and
Mikito Nanashima\thanks{Tokyo Institute of Technology, Japan. \texttt{E-mail:~nanashima@c.titech.ac.jp}}\vspace{0.1cm}}
\date{}

\usepackage{lineno}

\begin{document}
\maketitle

\begin{abstract}
The coding theorem for Kolmogorov complexity states that any string sampled from a computable distribution has a description length close to its information content. A coding theorem for resource-bounded Kolmogorov complexity is the key to obtaining fundamental results in average-case complexity, yet whether any samplable distribution admits a coding theorem for randomized time-bounded Kolmogorov complexity ($\mathsf{rK}^\mathsf{poly}$) is open and a common bottleneck in the recent literature of meta-complexity. Previous works bypassed this issue by considering {probabilistic Kolmogorov complexity} ($\mathsf{pK}^\mathsf{poly}$), in which public random bits are assumed to be available.

In this paper, we present an efficient coding theorem for randomized Kolmogorov complexity under the non-existence of one-way functions, thereby removing the common bottleneck.
This enables us to prove $\mathsf{rK}^\mathsf{poly}$ counterparts of virtually all the average-case results that were proved only for $\mathsf{pK}^\mathsf{poly}$, and enables the resolution of the following concrete open problems.

\begin{enumerate}
    \item The existence of a one-way function is characterized by the failure of average-case symmetry of information for randomized time-bounded Kolmogorov complexity, as well as a conditional coding theorem for randomized time-bounded Kolmogorov complexity.  This resolves the open problem of Hirahara, Ilango, Lu, Nanashima, and Oliveira (STOC'23).
     \item
     Hirahara, Kabanets, Lu, and Oliveira (CCC'24) showed that randomized time-bounded Kolmogorov complexity admits search-to-decision reductions in the errorless average-case setting over any samplable distribution, and left open whether a similar result holds in the error-prone setting.  We resolve this question affirmatively, and as a consequence, characterize the existence of a one-way function by the average-case hardness of computing $\mathsf{rK}^\mathsf{poly}$ with respect to an arbitrary samplable distribution, which is an $\mathsf{rK}^\mathsf{poly}$ analogue of the $\mathsf{pK}^\mathsf{poly}$ characterization of Liu and Pass (CRYPTO'23).
\end{enumerate}

The key technical lemma is that any distribution whose next bits are efficiently predictable admits an efficient encoding and decoding scheme, which could be of independent interest to data compression.
\end{abstract}

\thispagestyle{empty} 
\newpage
\tableofcontents
\thispagestyle{empty} 

\newpage

\setcounter{page}{1}

\section{Introduction}

Shannon's source coding theorem is a centerpiece of information theory.
It shows that if $m$ independent samples are drawn from a distribution $\mathcal{D}$, then the $m$ samples can be encoded into a string of expected length $m \cdot 
 ( \Shannon(\mathcal{D}) + o(1) )$,
where $\Shannon(\mathcal{D})$ denotes the Shannon entropy of $\mathcal{D}$.
A computationally efficient variant of Shannon's source coding theorem was given by \citet{ImpagliazzoZ89_focs_conf}, who showed that $m$ independent samples drawn from any polynomial-time samplable distribution $\mathcal{D}$ can be (inefficiently) compressed into a string of expected length $m\cdot  ( \Shannon(\mathcal{D}) + o(1))$ that can be decoded in polynomial time.
Thus, the amortized encoding length of one string in the many samples from $\mathcal{D}$ approaches to $\Shannon(\mathcal{D})$, which is information-theoretically optimal.

Less understood is a ``one-shot'' setting, in which \emph{one} string $x$ is drawn from a distribution $\mathcal{D}$, and the question is whether $x$ has a short description.
Information theoretically,
for any distribution $\mathcal{D}$ over $\binset^*$,
there exists an encoding scheme that
compresses
any string $x$ in the support of a distribution $\mathcal{D}$ 
into a string of length $\log \frac{1}{\mathcal{D}(x)} + O(1)$,
where $\mathcal{D}(x)$ denotes the probability that $x$ is sampled from $\mathcal{D}$.
In terms of Kolmogorov complexity, 
this result is often referred to as the \emph{coding theorem for Kolmogorov complexity} (coined in \cite{LuO21_icalp_conf}; it is also called a source compression theorem \cite{Lee06_PhD_Thesis}).
It states that any string $x$ in the support of a computable distribution $\mathcal{D}$
satisfies that 
\[
\K(x) \le \log \frac{1}{\mathcal{D}(x)} + O(1),
\]
where $\K(x)$ denotes the Kolmogorov complexity of $x$, i.e., the length of a shortest program that prints $x$,
and the constant $O(1)$ depends only on the distribution $\mathcal{D}$.
Note that Kolmogorov complexity does not impose any time bound on the time it takes to print $x$.
This limits the applicability of the coding theorem in the literature of computational complexity theory.
More relevant to complexity theory is a coding theorem for resource-bounded Kolmogorov complexity measures, such as $\K^t(x)$, i.e., the length of a shortest program that prints $x$ in time $t$.

The coding theorem is one of the most fundamental properties of Kolmogorov complexity,%
\footnote{According to \citet{Lee06_PhD_Thesis}, a coding theorem is one of the four ``pillars'' of Kolmogorov complexity.}
and is the key to establishing fundamental theorems of average-case complexity theory.
For example, \citet{Levin86_siamcomp_journals} initiated the theory of average-case $\NP$-completeness by presenting a natural distributional problem which is complete for $\NP$ with respect to the class $\PComp$ of polynomial-time computable distributions.
Here, a distribution is said to be \emph{polynomial-time computable} if the cumulative distribution function is computable in polynomial time.
\citeauthor{Levin86_siamcomp_journals} showed this completeness result by showing that
$\PComp$ admits an \emph{efficient coding theorem} for $\K^t$, that is, that
any string in the support of a polynomial-time computable distribution $\mathcal{D}$ can be compressed into a polynomial-time program of length $\log \frac{1}{\mathcal{D}(x)} + O(1)$ in polynomial time.
We refer the reader to the survey of \citet{BogdanovT06_fttcs_journals} for the background on average-case complexity theory.

Can we obtain a coding theorem for resource-bounded Kolmogorov complexity with respect to a wider class of distributions?
The most standard class of distributions considered in the literature of average-case complexity theory is the class $\PSamp$ of (polynomial-time) samplable distributions.
A distribution $\mathcal{D} = \{\mathcal{D}_n\}_{n \in \Nat}$ is said to be (\emph{polynomial-time}) 
 \emph{samplable} if there exists a randomized polynomial-time algorithm that, on input $1^n$, outputs a string that is distributed according to $\mathcal{D}_n$.
Under a plausible derandomization assumption, 
\citet{AntunesF09_coco_conf} proved a coding theorem for $\K^\poly$ with respect to any samplable distribution.
The assumption can be removed if we consider a randomized variant of $\K^\poly$ in which public random bits are given to short programs.
The \emph{$t$-time-bounded probabilistic Kolmogorov complexity} of a string $x$, denoted by $\pK^t$ \cite{GoldbergKLO22_coco_conf}, is defined to be the minimum integer $k$ such that the probability that $\K^t(x \mid r) \le k$ over a uniformly random $r \in \binset^t$ is at least $\frac{2}{3}$, where $\K^t(x \mid r)$ denotes the conditional Kolmogorov complexity of $x$ given $r$, i.e., the length of a shortest program that prints $x$ given $r$ as input in time $t$.
\citet{LuOZ22_icalp_conf} showed that any samplable distribution admits a coding theorem for probabilistic Kolmogorov complexity.
Note that the notion of $\pK^t$ deviates from the standard notion of Kolmogorov complexity in that depending on the public random bits $r$, the shortest program that prints $x$ on input $r$ may be different.
In fact, $\pK^\poly$ is essentially equivalent to (the logarithm of the reciprocal of) the time-bounded universal probability \cite{HiraharaN23_focs_conf}, which is technically useful but somewhat artificial from the perspective of data compression.%
\footnote{In terms of data compression, the difference between $\rK^\poly$ and $\pK^\poly$ can be explained as follows.  In $\pK^\poly$, we assume that an inefficient encoding algorithm and an efficient decoding algorithm share random bits, which may not be the case in practice.  In $\rK^\poly$, an efficient decoding algorithm is allowed to be randomized, but the random bits are private and not shared with an encoding algorithm.}

A more natural randomized variant of time-bounded Kolmogorov complexity is \emph{randomized Kolmogorov complexity}.
The $t$-time-bounded randomized Kolmogorov complexity of a string $x$, denoted by $\rK^t(x)$, is defined to be the length of a shortest randomized program that prints $x$ in time $t$ with probability $\frac{2}{3}$ over the internal randomness of the randomized program.
This is arguably more natural than $\pK^\poly$ in that the program is fixed irrespective of the random bits used by the program.
It is evident that $\pK^t(x) \le \rK^t(x) \le \K^t(x)$, and thus the compression power of $\rK^t$ is in between $\pK^t$ and $\K^t$.
Partial progress towards obtaining coding theorems for randomized Kolmogorov complexity was made by \citet{LuO21_icalp_conf,LuOZ22_icalp_conf}, who proved a (information-theoretically sub-optimal) coding theorem for an exponential-time variant of $\rK^\poly$ (the randomized variant of Levin's $\mathsf{Kt}$-complexity \cite{Oliveira19_icalp_conf}).
No optimal coding theorem for $\rK^\poly$ is known for any class of distributions larger than $\PComp$.
This leads us to the following question: \emph{For which distributions \emph{(}and when\emph{)} does a coding theorem for $\rK^\poly$ holds}?

Answering this question is indispensable for a closely related area of research --- \emph{data compression}.  The main question investigated in the literature of data compression \cite{GoldbergS91_siamcomp_journals,TrevisanVZ05_cc_journals,Wee04_coco_conf,BarakSW03_random_conf,HsiaoLR07_eurocrypt_conf,HaitnerMS23_innovations_conf} is which class of distributions admits \emph{efficient} coding theorems rather than (existential) coding theorems.
The difference between the two types of the coding theorems is that
in the latter, we do not care about the efficiency of an encoding algorithm.
In an \emph{efficient} coding theorem for a distribution $\mathcal{D}$, we require that there exists a \emph{polynomial-time} algorithm that takes a string $x$ drawn the distribution $\mathcal{D}$ and outputs a compressed string of length close to its information content $\log \frac{1}{\mathcal{D}(x)}$.
\citet{GoldbergS91_siamcomp_journals,TrevisanVZ05_cc_journals} identified several classes of distributions that admit efficient coding theorems, such as distributions samplable with logspace machines \cite{TrevisanVZ05_cc_journals}, high entropy sources \cite{GoldbergS91_siamcomp_journals,TrevisanVZ05_cc_journals}, and samplable witness sets for $\NP$ \cite{TrevisanVZ05_cc_journals}.
However, no efficient coding theorem for any class of distributions that strictly contains $\PComp$ is known, just because even existential coding theorems for $\rK^\poly$ are unknown.

\subsection{Interplay between One-Way Functions and Kolmogorov Complexity}
Faced with the lack of a coding theorem for $\rK^\poly$, previous works in the recent literature of meta-complexity bypassed this issue by considering probabilistic Kolmogorov complexity $\pK^\poly$ or resource-unbounded Kolmogorov complexity $\K$. 
There has been a flurry of new characterizations of the existence of one-way functions based on Kolmogorov complexity \cite{LiuP20_focs_conf,RenS21_coco_conf,LiuP21_crypto_conf,IlangoRS22_stoc_conf,AllenderCMTV21_fsttcs_conf,LiuP22a_coco_conf,LiuP23_tcc_conf,LiuP23_crypto_conf,HiraharaILNO23_stoc_conf,Hirahara23_stoc_conf,ImpagliazzoL90_focs_conf,HiraharaN23_focs_conf}, starting from the influential work of \citet{LiuP20_focs_conf}.
A one-way function is one of the most fundamental cryptographic primitives because its existence is equivalent to the existence of a variety of cryptographic primitives, such as a private-key encryption scheme \cite{GoldwasserM84_jcss_journals}, a pseudorandom generator \cite{HastadILL99_siamcomp_journals}, a digital signature \cite{Rompel90_stoc_conf}, and a commitment scheme \cite{Naor91_joc_journals}.
The new ``meta-computational'' characterizations of one-way functions provide us with the hope that the improved understanding of one-way functions might lead us to the resolution of long-standing open problems, such as the elimination of Pessiland \cite{Impagliazzo95_coco_conf} (i.e., does the average-case hardness of $\NP$ imply the existence of a one-way function?).
Among the characterizations,
we highlight the characterizations that are based on an arbitrary samplable distribution. 
\begin{enumerate}
    \item \citet{HiraharaILNO23_stoc_conf} showed that average-case asymmetry of information for probabilistic Kolmogorov complexity $\pK^\poly$ characterizes the existence of a one-way function.  That is, a one-way function can be constructed if and only if for \emph{some} samplable distribution $\mathcal{D}$,
    the symmetry of information for $\pK^\poly$, i.e.,
    \[
    \pK^\poly(x \mid y) + \pK^\poly(y) \approx \pK^\poly(x, y) \approx  \pK^\poly(y \mid x) + \pK^\poly(x) 
    \]
    does not hold with a non-negligible probability over $(x, y)$ drawn from $\mathcal{D}$. 
    \item \citet{ImpagliazzoL90_focs_conf} and \citet{HiraharaN23_focs_conf} showed that a one-way function exists if and only if approximating time-bounded universal probability is hard with respect to some samplable distribution.
    \item \citet{IlangoRS22_stoc_conf} characterized the existence of a one-way function by the average-case hardness of Kolmogorov complexity $\K$ with respect to an arbitrary samplable distribution.
    \item \citet{LiuP23_crypto_conf} characterized the existence of a one-way function by the average-case hardness of computing probabilistic Kolmogorov complexity $\pK^\poly$ with respect to an arbitrary samplable distribution.
\end{enumerate}

These results provide fascinating approaches to construct one-way functions in that it suffices to construct \emph{some} samplable distribution that witnesses asymmetry of information or the computational intractability of computing Kolmogorov complexity measures, which appears to be intuitively easier.
However, the results do not extend to randomized Kolmogorov complexity $\rK^\poly$, precisely because of the lack of a coding theorem for $\rK^\poly$.
Indeed, 
all the proofs of the results above rely on a coding theorem for corresponding Kolmogorov complexity measures;
for example, the result of \cite{IlangoRS22_stoc_conf} relies on the coding theorem for resource-unbounded Kolmogorov complexity $\K$.

\subsection{Our Results}
In this paper, we identify a class of distributions that contains $\PComp$ and admits an efficient coding theorem for $\rK^\poly$.
Roughly speaking, we present an efficient and information-theoretically optimal coding theorem for a distribution $\mathcal{D}$ if there exists a ``{next-bits predictor}'' for $\mathcal{D}$ in the sense that any next bit of a given arbitrary prefix is predictable with high accuracy in randomized polynomial time.
This enables us to show that $\rK^\poly$, $\pK^\poly$ and $\K$ are all approximately equal to each other on average if one-way functions do not exist.
In particular, we demonstrate that virtually all the average-case results that were previously shown to hold only for $\pK^\poly$ can be translated into $\rK^\poly$ counterparts.
This enables us to resolve the main open problems left in previous works \cite{HiraharaILNO23_stoc_conf,HaitnerMS23_innovations_conf,HiraharaKLO24_coco_conf}.
We describe details below.

\subsubsection{Symmetry of Information for \texorpdfstring{$\rK^\poly$}{rK poly} versus One-Way Functions}

Symmetry of information for Kolmogorov complexity \cite{ZvonkinL1970} is one of the most fundamental properties of Kolmogorov complexity and is yet another one of the four ``pillars'' of Kolmogorov complexity \cite{Lee06_PhD_Thesis}.
It states that for all strings $x$ and $y$ of length $n$,
\[
    \K(x \mid y) + \K(y) \approx \K(x, y) \approx  \K(y \mid x) + \K(x),
\]
where the approximate equality holds up to an additive $O(\log n)$ term.
The original proof of symmetry of information due to Kolmogorov and Levin \cite{ZvonkinL1970} relies on an exhaustive search, and thus does not extend to the case of resource-bounded Kolmogorov complexity.
As early as the 1960s, Kolmogorov suggested that it is an interesting avenue of research to investigate symmetry of information for time-bounded Kolmogorov complexity \cite{LeeR05_tcs_journals}.
After a long line of research \cite{ZvonkinL1970,LongpreM93_ipl_journals,LongpreW95_iandc_journals,LeeR05_tcs_journals,Hirahara21_stoc_conf,Hirahara22_coco_conf,GoldbergK22_eccc_journals,GoldbergKLO22_coco_conf}, \citet{HiraharaILNO23_stoc_conf} presented two characterizations of an average-case variant of symmetry of information:
\begin{enumerate}
    \item The average-case asymmetry of information for $\pK^\poly$ characterizes the existence of a one-way function.
    \item The average-case asymmetry of information for $\rK^{\mathsf{quasipoly}}$ characterizes the existence of a one-way function secure against \emph{quasi-polynomial-time} algorithms.
\end{enumerate}
It was left as a main open problem (highlighted by Osamu Watanabe in \cite{HiraharaILNO23_stoc_conf}) whether the failure of the average-case symmetry of information for $\rK^\poly$ characterizes the existence of a \emph{standard} one-way function (i.e., secure against polynomial-time algorithms).

We resolve this open problem affirmatively and obtain the following new characterization of the non-existence of one-way functions through the validity of symmetry of information for $\rK^{\poly}$, as well as an average-case conditional coding theorem.

\begin{restatable}{theorem}{OWFIO}\label{t:OWF-io}
	The following are equivalent.
	\begin{enumerate}
		\item\label{i:OWF-io-noOWF} One-way functions do not exist.

		\item\label{i:OWF-io-SoI} \emph{\textbf{(Infinitely-Often Average-Case Symmetry of Information for $\rK^t$)}} For every polynomial-time samplable distribution family $\{\mathcal{D}_n\}_{n \in\mathbb{N}}$ supported over $\{0,1\}^n\times \{0,1\}^n$, and every polynomial $q$, there exists a polynomial $p$ such that for infinitely many $n\in\mathbb{N}$, the following holds for all $t\geq p(n)$,
		\[
		\Prob_{(x,y)\sim \mathcal{D}_n} \mleft[\mathsf{rK}^{t}(x\mid y)\leq \mathsf{rK}^{t}(x,y)-\mathsf{rK}^{t}(y)+ \log t \mright] \geq 1-\frac{1}{q(n)}.
		\]

  	\item\label{i:OWF-io-Coding} \emph{\textbf{(Infinitely-Often Average-Case Conditional Coding for $\rK^t$)}}
		For every polynomial-time samplable distribution family $\{\mathcal{D}_n\}_{n \in\mathbb{N}}$ supported over $\{0,1\}^n\times \{0,1\}^n$, and every polynomial $q$, there exists a polynomial $p$ such that for infinitely many $n\in\mathbb{N}$,
		\[
		\Prob_{(x,y)\sim \mathcal{D}_n} \mleft[\mathsf{rK}^{p(n)}(x\mid y)\leq \log \frac{1}{\D_n(x\mid y)}+ \log p(n) \mright] \geq 1-\frac{1}{q(n)},
		\]		
       where $\mathcal{D}_n(x \mid y)$ denotes the probability that $(x, y)$ is sampled from $\mathcal{D}_n$ conditioned that the second item being sampled is $y$.

      	\item\label{i:OWF-io-efficient} \emph{\textbf{(Infinitely-Often Average-Case Efficient Conditional Coding for $\rK^t$)}}
		For every polynomial-time samplable distribution family $\{\mathcal{D}_n\}_{n \in\mathbb{N}}$ supported over $\{0,1\}^n\times \{0,1\}^n$, and every polynomial $q$, there exists a polynomial $p$ such that for infinitely many $n\in\mathbb{N}$,
		\[
		\Prob_{(x,y)\sim \mathcal{D}_n} \mleft[\mathsf{rK}^{p(n)}(x\mid y)\leq \log \frac{1}{\D_n(x\mid y)}+ \log p(n) \mright] \geq 1-\frac{1}{q(n)}.
		\]		
  Moreover, it admits an efficient encoder in the following sense: there exists an efficient algorithm $\Enc$ that outputs, for given $(x,y)\sim \mathcal{D}_n$, a description of a $p(n)$-time program $\Pi$ of length at most $-\log\D_n(x\mid y)+\log p(n)$ with probability at least $1-1/q(n)$ over the choice of $(x,y)\sim \mathcal{D}_n$ and randomness for $\Enc$, such that
  $\Pi$ outputs $x$ for given $y$ and randomness $r\sim\bin^{p(n)}$ with probability at least $2/3$ over the choice of $r$.
	\end{enumerate}
\end{restatable}

Our proof is fundamentally different from the previous proof of \cite{HiraharaILNO23_stoc_conf} for $\rK^\mathsf{quasipoly}$.
The proof of \cite{HiraharaILNO23_stoc_conf} relies on the reconstructive extractors of \citet{Trevisan01_jacm_journals,RazRV02_jcss_journals}, whose advice complexity of the reconstruction procedure is information-theoretically sub-optimal by an additive $O(\log^3 n)$ term, and this term is what forced \cite{HiraharaILNO23_stoc_conf} to consider quasi-polynomial-time one-way functions.
Roughly speaking, the additive error term corresponds to the seed length of an extractor.
Even without the reconstructive property, the state-of-the-art extractor construction due to \citet{GuruswamiUV09_jacm_journals} has seed length $O(\log^2 n)$.
Thus, in order to obtain \Cref{t:OWF-io} using the approach of \cite{HiraharaILNO23_stoc_conf}, we would need to improve the seed length of the state-of-the-art extractor construction to $O(\log n)$.
We sidestep this issue by taking a new approach based on the proof techniques of \citet{GoldbergS91_siamcomp_journals,TrevisanVZ05_cc_journals}.

\citet{HaitnerMS23_innovations_conf} obtained
an efficient coding theorem for $\pK^\poly$  with respect to any samplable distribution under the non-existence of a one-way function; that is, any samplable distribution $\mathcal{D}$ admits a polynomial-time encoding and decoding scheme of expected length $\Shannon(\mathcal{D}) + O(1)$ \emph{when shared random bits are available.}
\Cref{i:OWF-io-efficient} of \Cref{t:OWF-io} provides the same conclusion (up to an additive logarithmic term) \emph{without any shared random bit}, which resolves the natural open problem left in \cite{HaitnerMS23_innovations_conf}.
We also make progress towards the uniform version of the main result of \cite{HaitnerMS23_innovations_conf} by showing that
any distribution incompressible to $k$ bits without shared random bits
has $(1-\epsilon) k - O(\log n)$ bits of uniform-next-bit pseudoentropy for any constant $\epsilon > 0$; we defer the details to \Cref{sec:uniform-HMS}. 

\Cref{t:OWF-io} elucidates that the non-existence of a one-way function is both necessary and sufficient for the conditional version of an average-case coding theorem for $\rK^\poly$.
Moreover, both efficient and existential coding theorems for $\rK^\poly$ are equivalent to each other (\Cref{i:OWF-io-Coding,i:OWF-io-efficient}).
We also mention that $\rK^\poly$, $\pK^\poly$, and $\K$ are all approximately equal to each other on average under the non-existence of a one-way function, which enables us to translate any average-case result about $\pK^\poly$ into an $\rK^\poly$ counterpart (unless a one-way function exists); see \Cref{optimal coding-io K-version}.

We also have an analogous result for \emph{infinitely-often} one-way functions.
In this case, using the notion of computational depth \cite{AntunesFMV06_tcs_journals}, we obtain a characterization for a \emph{worst-case} variant of symmetry of information for $\rK^\poly$, which comes tantalizingly closer to the worst-case symmetry of information for $\K^\poly$ investigated by \citet{LongpreM93_ipl_journals,LongpreW95_iandc_journals}.
A $t$-time-bounded computational depth $\cd^t(x)$ is defined as $\pK^t(x) - \K(x)$.

\begin{restatable}{theorem}{OWFAE}\label{t:OWF-ae}
	The following are equivalent.
	\begin{enumerate}
		\item \label{i:OWF-ae-noOWF}Infinitely-often one-way functions do not exist.
		
		\item \label{i:OWF-ae-Coding} \emph{\textbf{(Almost-Everywhere Average-Case Conditional Coding for $\rK^t$)}}
		For every polynomial-time samplable distribution family $\{\mathcal{D}_n\}_{n \in\mathbb{N}}$ supported over $\{0,1\}^n\times \{0,1\}^n$, there exists a polynomial $p$ such that for all $n,k\in\mathbb{N}$,
		\[
		\Prob_{(x,y)\sim \mathcal{D}_n} \mleft[\mathsf{rK}^{p(n,k)}(x\mid y)\leq \log \frac{1}{\D_n(x\mid y)}+ \log p(n,k) \mright] \geq 1-\frac{1}{k}.
		\]		
		
		\item \label{i:OWF-ae-Coding-cd} \emph{\textbf{(Almost-Everywhere Worst-Case Conditional Coding for $\rK^t$ with Computational Depth)}}
		There exists a constant $c>0$ such that the following holds. For every computable distribution family $\{\mathcal{D}_n\}_{n \in\mathbb{N}}$ supported over $\{0,1\}^n\times \{0,1\}^n$, all $n,t\in\mathbb{N}$ such that $t\geq n$ and all $(x,y)\in\supp(\D_n)$,
		\[
		\mathsf{rK}^{({2^{\alpha}\cdot t})^c}(x\mid y)\leq \log \frac{1}{\D_n(x \mid y)} + c\cdot(\log t + \alpha),
		\]
		where $\alpha\vcentcolon=\mathsf{cd}^{t}(x,y)$.
		
		\item \label{i:OWF-ae-SoI} \emph{\textbf{(Almost-Everywhere Average-Case Symmetry of Information for $\rK^t$)}} For every polynomial-time samplable distribution family $\{\mathcal{D}_n\}_{n \in\mathbb{N}}$ supported over $\{0,1\}^n\times \{0,1\}^n$, there exists a polynomial $p$ such that for all $n,k\in\mathbb{N}$ and $t\geq p(n,k)$,
		\[
		\Prob_{(x,y)\sim \mathcal{D}_n} \mleft[\mathsf{rK}^{t}(x\mid y)\leq \mathsf{rK}^{t}(x,y)-\mathsf{rK}^{t}(y)+ \log t \mright] \geq 1-\frac{1}{k}.
		\]
		
		\item\label{i:OWF-ae-SoI-cd} \emph{\textbf{(Almost-Everywhere Worst-Case Symmetry of Information for $\rK^t$ with Computational Depth)}} 		There exists a constant $c>0$ such that the following holds. For all $n,t\in\mathbb{N}$ such that $t\geq n$ and all $x,y\in\bool^n$,
		\[
		\mathsf{rK}^{({2^{\alpha}\cdot t})^c}(x\mid y)\leq \rK^t(x,y)-\rK^t(y) + c\cdot(\log t + \alpha),
		\]
		where $\alpha\vcentcolon=\mathsf{cd}^{t}(x,y)$.
	\end{enumerate}
\end{restatable}

\subsubsection{Error-Prone Average-Case Search-to-Decision Reductions for \texorpdfstring{$\rK^\poly$}{rK poly}}

Next, we investigate the open question left by \citet{HiraharaKLO24_coco_conf}
and an $\rK^\poly$ counterpart of the $\pK^\poly$ characterization of \citet{LiuP23_crypto_conf}.
Whether a search-to-decision reduction exists for the problem of computing time-bounded Kolmogorov complexity is a long-standing open problem that dates back to as early as the 1960s \cite{Trakhtenbrot84_annals_journals},
and recently there has been progress on this question \cite{CarmosinoIKK16_coco_conf,Hirahara18_focs_conf,Ilango20_coco_conf,LiuP20_focs_conf,MazorP24_coco_conf,HiraharaKLO24_coco_conf}.
\citet{HiraharaKLO24_coco_conf} presented a search-to-decision reduction for computing $\rK^\poly$ in the errorless average-case setting:
if there exists an efficient errorless average-case algorithm for computing $\rK^\poly$ on average, then there exists an efficient errorless average-case algorithm that finds a shortest randomized program of length $\rK^\poly(x)$ on a random input $x$ drawn from an arbitrary samplable distribution.
Designing such a reduction is well motivated by the fact that such a reduction is \emph{necessary} for excluding (an errorless variant of) Pessiland from Impagliazzo's five worlds \cite{Impagliazzo95_coco_conf}; see \cite{HiraharaKLO24_coco_conf} for more details on the background.

The proof of \cite{HiraharaKLO24_coco_conf} is based on a highly non-trivial combination of a reconstructive disperser of \cite{Hirahara20_focs_conf} and a non-reconstructive disperser of \cite{Ta-ShmaUZ07_combinatorica_journals}, and does not extend to the \emph{error-prone} average-case setting.
Designing a similar reduction in the error-prone average-case setting was left as one of the main open questions in \cite{HiraharaKLO24_coco_conf}.
The difference between error-prone and errorless average-case complexities \cite{HiraharaS22_innovations_conf,HiraharaN22_coco_conf} is that 
in the latter, an average-case algorithm is not allowed to make any error and instead allowed to indicate its failure of an algorithm, which is equivalent to the notion of average-polynomial-time \cite{Impagliazzo95_coco_conf,BogdanovT06_fttcs_journals}.

Using our new coding theorem, we present a search-to-decision reduction in the error-prone average-case setting, thereby answering the open problem of \cite{HiraharaKLO24_coco_conf}.
To state the result formally, we need a couple of definitions.
For $\lambda\in [0,1)$, let $\lambda$-$\mathsf{MINrKT}$ be the following promise problem $({\sf YES}, {\sf NO})$:
\begin{align*}
    {\sf YES} &\vcentcolon= \mleft\{(x,1^s,1^t,1^\ell) \mid \rK^t_{\lambda}(x)\leq s\mright\},\\
    {\sf NO} &\vcentcolon= \mleft\{(x,1^s, 1^t,1^\ell) \mid \rK^t_{\lambda-1/\ell}(x)> s\mright\}.
\end{align*}
We say that an algorithm $A$ \emph{decides} $(\mathsf{YES}, \mathsf{NO})$ on input $x$ if $x \in \mathsf{YES}$ implies $A(x) = 1$ and $x \in \mathsf{NO}$ implies $A(x) = 0$.
The promise problem has a natural search version.
For $x\in\bool^{n}$, $t\in\mathbb{N}$ and $0<\varepsilon,\lambda < 1$, we say that a randomized program $M$ is an \emph{$\varepsilon$-$\rK_{\lambda}^t$-witness of $x$} if
\begin{itemize}
    \item $|M|\leq \rK_{\lambda}^t(x)$, and
    \item $M$ outputs $x$ within $t$ steps with probability at least $\lambda-\varepsilon$ over the internal randomness of $M$.
\end{itemize}
For $\lambda\in [0,1)$, let $\lambda$-$\mathsf{Search}$-$\mathsf{MINrKT}$ be the following search problem: Given $(x,1^t,1^\ell)$, where $x\in\bool^*$, $t,\ell\in\mathbb{N}$, find an $(1/\ell)$-$\rK^t_{\lambda}$-witness of $x$.

\begin{restatable}{theorem}{OWFrK}\label{t:owf-rKt-formal}
	The following are equivalent.
	\begin{enumerate}
		\item\label{i:owf-rKt-1} Infinitely-often one-way functions do not exist.
		
		\item\label{i:owf-rKt-2}  \emph{\textbf{($\mathsf{Search}$-$\mathsf{MINrKT}$ is easy on average over polynomial-time samplable distributions)}} 	 For every $\lambda\in[0,1]$, every polynomial-time samplable distribution family $\{\D_n\}_{n\in\mathbb{N}}$, where each $\D_n$ is over $\bool^n$, there exist a polynomial $\rho$ and a probabilistic polynomial-time algorithm $A$ such that for all $n, s,\ell,k\in \mathbb{N}$, and all $t\geq \rho(n)$,
		\[
		\Prob_{x\sim \D_n, A} \mleft[A(x,1^t, 1^{\ell}, 1^k) \textnormal{ outputs an $(1/\ell)$-$\rK_{\lambda}^t$-witness of $x$}\mright] \geq 1-\frac{1}{k}.
		\]
		
		\item\label{i:owf-rKt-3}  \emph{\textbf{($\mathsf{MINrKT}$ is easy on average over polynomial-time samplable distributions)}} For every $\lambda\in[0,1)$, every polynomial-time samplable distribution family $\{\D_n\}_{n\in\mathbb{N}}$, where each $\D_n$ is over $\bool^n$, there exist a polynomial $\rho$ and a probabilistic polynomial-time algorithm $A$ such that for all $n, s,\ell,k\in \mathbb{N}$, and all $t\geq \rho(n)$,
		\[
		\Prob_{x\sim \D_n, A} \mleft[A(-,1^k)  \textnormal{ decides } \lambda\textnormal{-}\mathsf{MINrKT} \text{ on input } (x,1^s,1^t,1^\ell)\mright] \geq 1-\frac{1}{k}.
		\]

  		\item\label{i:owf-rKt-4}  \emph{\textbf{($\mathsf{MINrKT}$ is easy on average over the uniform distribution)}} There exist a polynomial $\rho$ and a probabilistic polynomial-time algorithm $A$ such that for all $n, s, \ell, k \in \Nat$,
		\[
		\Prob_{x\sim \bool^n, A} \mleft[A(-,1^k) \textnormal{ decides } (2/3)\textnormal{-}\mathsf{MINrKT} \textnormal{ on input }(x,1^s,1^{\rho(n)},1^{\ell})\mright] \geq 1-\frac{1}{k}.
		\]
	\end{enumerate}
\end{restatable}

This extends the $\pK^\poly$ characterization of one-way functions by \citet{LiuP23_crypto_conf} 
to the $\rK^\poly$ counterparts.

\subsubsection{Optimal Coding Theorems for Next-Bits Predictable Distributions}

The key lemma behind all the results above is an \emph{unconditional} efficient coding theorem for $\rK^\poly$ with respect to \emph{next-bits predictable distributions}.

\begin{definition}[See also \Cref{def: nextbit approx}]
For a family $\mathcal{D}=\{\mathcal{D}_n\}_{n \in \Nat}$ of distributions and a function $\varepsilon \colon \Nat \to (0, 1)$,
a \emph{next-bits predictor for $\mathcal{D}$ with accuracy $\varepsilon$} is a randomized polynomial-time algorithm 
such that 
for every $n\in\N$, every $x\in \supp(\mathcal{D}_n)$, every $i\in \{1, \cdots, |x|\}$, and every $b\in\bin$,
	\[
 \Prob_P\mleft[\mathcal{D}^*_n(b \mid x_{[i-1]})-\varepsilon(n) \le P(x_{[i-1]},b,1^n) \le \mathcal{D}^*_n(b \mid x_{[i-1]})+\varepsilon(n) \mright] \ge 1-\varepsilon(n),
 \]
 where $\mathcal{D}^*_n(b \mid x_{[i-1]})$ denotes the probability, over $X \sim \mathcal{D}_n$, that the $i$-th bit of $X$ is $b$ conditioned that the first $(i-1)$-bits prefix of $X$ is equal to that of $x$.%
 \footnote{Throughout this paper, we only consider a family $\mathcal{D}$ of distributions such that every $x \in \supp(\mathcal{D}_n)$ has length at most $p(n)$ for some polynomial $p$. }
 We say that $\mathcal{D}$ is \emph{next-bits predictable} if for all polynomials $q$, there exists a next-bits predictor for $\mathcal{D}$ with accuracy $1/q$.
\end{definition}

This definition should be compared with Yao's next-bit predictor \cite{Yao82a_focs_conf,Vadhan12_fttcs_journals}. 
There are three differences between Yao's \emph{next-bit} predictor and our \emph{next-bits} predictor.
\begin{enumerate}
    \item  For a distribution $\mathcal{D}$, Yao's next-bit predictor only predicts the $i$-th bit given the first $i-1$ bits of a random string $x$ sampled from $\mathcal{D}$ for \emph{some} index $i \in \{1, \cdots, |x|\}$.
    In contrast, the definition of a \emph{next-bits} predictor requires that for \emph{all} indices $i \in \{1, \cdots, |x|\}$, the $i$-th bit is predictable given the first $i-1$ bits of $x$ .
    \item  We require that the accuracy of the prediction can be made arbitrarily small, whereas Yao's next-bit predictor is accurate with a non-negligible probability.
    \item The output of Yao's next-bit predictor is considered to be correct if the bit produced by Yao's next-bit predictor is equal to the next bit, whereas a next-bits predictor aims to estimate the probability density function of the next bits.
    For example, the uniform distribution is next-bits predictable, but does not have Yao's next-bit predictor.
    In this sense, our notion of next-bits predictor is close in spirit to the notion of $\KL$ predictor of \citet{VadhanZ12_stoc_conf}.
\end{enumerate}

For next-bits predictable distributions, we present an optimal and efficient coding theorem for $\rK^\poly$ up to an additive $O(\log n)$ term.
This extends the previous work of \citet{Levin86_siamcomp_journals} because any polynomial-time computable distribution is next-bits predictable.

\begin{theorem}[See also \Cref{t:opt-avg-coding-approx}]
\label{thm:opt-avg-coding-in-introduction}
For any next-bits predictable family $\mathcal{D}=\{\mathcal{D}_n\}_{n \in \Nat}$ of distributions, and  for every polynomial $q$,
there exists an \emph{efficient encoding and decoding scheme} whose expected encoding length is $\Shannon(\mathcal{D}_n) + \log q(n)$; that is, 
there exists a pair $(\Enc, \Dec)$ of randomized polynomial-time algorithms such that 
for every $n \in \Nat$,
\[
 \Exp_{x \sim \mathcal{D}_n, \Enc} [ |\Enc(1^n, x)| ] \le \Shannon(\mathcal{D}_n) + \log p(n)
\]
and for every $x \in \supp(\mathcal{D}_n)$, \[\Prob_{\Enc, \Dec} \mleft[ \Dec(1^n, \Enc(1^n, x)) = x \mright] \ge \frac{2}{3}.
\]
\end{theorem}

Moreover, we obtain a \emph{worst-case} coding theorem for $\rK^\poly$ optimal up to a $(1+\epsilon)$-factor for every constant $\epsilon > 0$, which is instrumental in \Cref{sec:uniform-HMS}.
\begin{theorem}[See also \Cref{t:almost-opt-wst-coding}]
For any next-bits predictable family $\mathcal{D}=\{\mathcal{D}_n\}_{n \in \Nat}$ of distributions,  for every $\epsilon>0$, there exists a polynomial $p$ such that for every $n\in\N$ and every $x\in \supp(\mathcal{D}_n)$, 
    \[\rK^{p(n)}(x)\le (1+\epsilon)\log\frac{1}{\mathcal{D}_n(x)} + \log p(n).\]
\end{theorem}

These results could be interesting to practical data compression.  Data compressors in practice (see \cite{0023825_daglib_books}) work by predicting next symbols by a \emph{deterministic} algorithm.  Our results show that compression is possible even if a predictor is \emph{randomized} and makes a small additive error.

Although it remains open whether an existential coding theorem for $\rK^\poly$ with respect to $\PSamp$ can be obtained unconditionally,
we remark that the non-existence of a one-way function is necessary for an \emph{efficient} coding theorem for $\rK^\poly$.
Thus, \Cref{thm:opt-avg-coding-in-introduction} is unlikely to be extended to any samplable distribution.
\begin{theorem}\label{t:char-owf}
    The following are equivalent.
    \begin{enumerate}
        \item Infinitely-often one-way functions do not exist.\label{i:char-owf:no-owf}
        \item For every polynomial-time samplable distribution $\mathcal{D} = \{\mathcal{D}_n\}_{n \in \Nat}$, there exists an efficient encoding and decoding scheme whose expected encoding length is $\Shannon(\mathcal{D}_n) + O(\log n)$.\label{i:char-owf:optimal}
        \item For any constant $\varepsilon > 0$, for every polynomial-time samplable distribution over $\binset^n$ with entropy $n^\varepsilon$, there exists a polynomial-time encoding and decoding with expected length $n - 3$.\label{i:char-owf:nontrivial}
    \end{enumerate}
\end{theorem}

\section{Proof Overview}
At a high level, our proof of symmetry of information for $\rK^\poly$ under the non-existence of one-way functions proceeds as follows.
\begin{enumerate}
    \item To prove an efficient coding theorem for $\rK^\poly$ with respect to next-bits predictable distributions (\Cref{thm:opt-avg-coding-in-introduction}), we apply arithmetic encoding \cite[cf.][Sections~5.9 and~13.3]{0016881_daglib_books} to a \emph{randomized} next-bits predictor and then ``pseudo-derandomize'' the encoding by using the techniques of \citet{GoldbergS91_siamcomp_journals,TrevisanVZ05_cc_journals}.
    \item Using the distributional inverter of \citet{ImpagliazzoL89_focs_conf} (as in \cite{ImpagliazzoL90_focs_conf,HiraharaN23_focs_conf}), it can be shown that under the non-existence of one-way functions, any samplable distribution has a next-bits predictor (on average).  This enables us to deduce the average-case efficient conditional coding theorem for $\rK^\poly$
    under the non-existence of one-way functions (\Cref{i:OWF-io-noOWF} $\implies$ \Cref{i:OWF-io-efficient} in \Cref{t:OWF-io}) from \Cref{thm:opt-avg-coding-in-introduction}.
    \item Average-case symmetry of information follows from the average-case conditional coding theorem, as in    \cite{HiraharaILNO23_stoc_conf}.
\end{enumerate}
We emphasize the simplicity over the previous work \cite{HiraharaILNO23_stoc_conf}, which we view as the strength of this work.
Below, we explain the proof ideas of the coding theorem for $\rK^\poly$ in \Cref{subsection:optcode}
and then the proof of the search-to-decision reduction in \Cref{subsection:search-to-decision}.

\subsection{Optimal Coding Theorems for \texorpdfstring{$\rK^\poly$}{rK poly}}
\label{subsection:optcode}

The starting point of this work is the insightful work of \citet{HaitnerMS23_innovations_conf}, who showed that any distribution incompressible to $k$ bits has $k - 2$ bits of next-bit pseudoentropy for non-uniform algorithms.
Although we defer the definition of next-bit pseudoentropy to \Cref{sec:uniform-HMS},
their proof proceeds as follows.
\begin{enumerate}
    \item If a distribution $\mathcal{D}$ does not have $k-2$ bits of next-bit pseudoentropy, then by the work of \citet{VadhanZ12_stoc_conf}, there exists a ``$\KL$ predictor'' $P$ that predicts any next bit with a reasonably high accuracy.
    \item The predictor $P$ induces a \emph{polynomial-size-computable distribution} $\mathcal{D}^P$ (i.e., a non-uniform analogue of $\PComp$) such that $\mathcal{D}^P$ and $\mathcal{D}$ are close in terms of $\KL$ divergence.
    \item Applying the arithmetic encoding to $\mathcal{D}^P$ as in the work of \citet{Levin86_siamcomp_journals}, they obtain an efficient coding theorem for $\mathcal{D}^P$ and thus for $\mathcal{D}$; i.e., $\mathcal{D}$ admits an encoding and decoding scheme that can be computed by \emph{polynomial-size circuits}.
\end{enumerate}

The main idea of the proof of (\Cref{i:OWF-io-noOWF} $\implies$ \Cref{i:OWF-io-efficient}) in \Cref{t:OWF-io} is to replace the $\KL$ predictor in the proof of \cite{HaitnerMS23_innovations_conf} with the next-bits predictor that can be constructed from \cite{ImpagliazzoL89_focs_conf,ImpagliazzoL90_focs_conf,HiraharaN23_focs_conf}.
However, this poses a technical challenge:  We consider a one-way function secure against \emph{uniform algorithms}, in which case the next-bits predictor $P$ is a randomized algorithm, and arithmetic encoding may not be applicable.
This issue was also noted in \cite{HaitnerMS23_innovations_conf} and prevented \citeauthor{HaitnerMS23_innovations_conf} from obtaining the uniform version of their results when shared random bits are not available.
To explain the issue briefly, let us explain how the arithmetic encoding works.
Assuming that the cumulative distribution function $F_{\mathcal{D}}(x):=\sum_{y<x}\mathcal{D}(x)$ is efficiently computable by a deterministic algorithm,
a string $x$ can be encoded into the first $\lceil-\log \mathcal{D}(x)\rceil+1$ bits of the value $F_{\mathcal{D}}(x) + \mathcal{D}(x)/2$. 
If $F_\mathcal{D}(x)$ is efficiently approximated by a \emph{randomized} algorithm $P$, 
the arithmetic encoding of $x$ may largely depend on the internal randomness of $P$.
To address this issue, we would need a \emph{pseudo-deterministic} algorithm that approximates $F_\mathcal{D}(x)$, i.e., an algorithm that produces a fixed approximate value for $F_\mathcal{D}(x)$ with high probability.

We pseudo-derandomize arithmetic encoding by using the techniques of ``adding noise and rounding'' developed by \citet{GoldbergS91_siamcomp_journals,TrevisanVZ05_cc_journals}, where they addressed similar issues to obtain randomized compression algorithms for flat sources over a language in $\P$~\cite{GoldbergS91_siamcomp_journals} and a witness set for $\NP$~\cite{TrevisanVZ05_cc_journals}.  The rough ideas are as follows.  When we execute the next-bits predictor for approximating $F_{\mathcal{D}}(x)$ and $\mathcal{D}(x)$, we add a random noise and round the noised value to the nearest value in the integer multiples of a certain real value. As was shown in the previous works, the outcome of the next-bits predictor is fixed with high probability over the choice of the random noise, where the random noise is required to be shared between the encoder and decoder, but the description length of the random noise is logarithmically small.  This idea enables us to make the next-bits predictor {pseudodeterministic} only with a logarithmic amount of shared randomness.

    The techniques of ``adding noise and rounding'' can cause another issue with the accuracy of the next-bits predictor. More specifically, adding noise and rounding yield an additional accuracy error, and when the error is much larger than the next-bit probability, the accuracy of the approximation of the next-bit probability can become insufficient for decoding in arithmetic encoding. To address this issue, we avoid using the approximate value produced by the next-bits predictor when the next-bit probability is small, and in this case, we embed the next bit into the encoding with the position. We call a next bit with a small next-bit probability a \emph{light next bit}. Namely, our encoding algorithm first determines whether the next bit is light {by using the next-bits predictor} and if so, it embeds the next bit into the encoding; otherwise, it uses the next-bit prediction with the techniques of addition noise and rounding.  Using these ideas, we can show that, for each string $x\in\bin^n$ in the support of the distribution $\mathcal{D}$, the length of the encoding is roughly at most $-\log \mathcal{D}(x)+O((\num(x) + 1) \cdot \log n)$, where $\num(x)$ is the number of the light next bits of $x$. We can easily observe that the number of the light next-bits of $x$ is $0$ with high probability over $x\sim\mathcal{D}$, which implies the optimal coding property for $\rK^{\poly}$ with an efficient encoder. The same idea enables us to prove the \emph{conditional} coding because the next-bits predictor can approximate the next-bit probability starting from any conditional string.

\subsection{Error-Prone Average-Case Search-to-Decision Reductions for \texorpdfstring{$\rK^\poly$}{rK poly}}
\label{subsection:search-to-decision}

We describe the proof ideas behind \Cref{t:owf-rKt-formal}. First of all, it was implicitly shown in \cite{LiuP20_focs_conf} that average-case tractability of (decisional) $\mathsf{MINrKT}$ over the uniform distribution implies the non-existence of one-way functions. Then, to show \Cref{t:owf-rKt-formal}, it suffices to show that if one-way functions do not exist, then $\mathsf{Search}$-$\mathsf{MINrKT}$ can be solved on average over polynomial-time samplable distributions.

At a high level, our proof follows the approach in \cite{LiuP20_focs_conf}, which shows that the non-existence of one-way functions implies the average-case tractability of $\mathsf{Search}$-$\mathsf{MINKT}$ over the \emph{uniform} distribution. Here $\mathsf{Search}$-$\mathsf{MINKT}$ is the problem of finding a minimum $t$-time program (or, a $\K^t$-witness) of a given string. In fact, their result can be generalized to any \emph{polynomial-time computable} distribution. Next, we describe this approach in more detail.

Roughly put, the approach consists of the following steps:
\begin{enumerate}
	\item Construct a function $f$ such that if $f$ can be inverted on average over uniformly random inputs of $f$, then one can obtain an average-case algorithm for finding $\K^t$-witnesses, over the distribution where each $x$ has probability mass $2^{-\K^t(x)}$.
	
	\item Show that such an average-case algorithm also works for any fixed-polynomial-time computable distribution.
\end{enumerate}

The authors of \cite{LiuP20_focs_conf} construct a function $f$ as follows: $f$ takes an integer $i\in[n+O(1)]$, representing the length of a program, and a program $\Pi\in\bool^{i}$. It then obtains the output string $x$ of $\Pi$ after running for $t$ steps. Finally, it outputs $(i,x)$.

Let us first suppose that we can invert $f$ in the worst case. Then, given $x$, one can find the smallest $i^*$ such that $(i^*,x)$ is inverted successfully, in which case we obtain a program of length $i^*$ that outputs $x$ within $t$ steps. Such a program will be a $\K^t$-witness of $x$.

In the case that we can invert $f$ on average over uniformly random inputs, such a search algorithm will succeed on average over $x$ sampled according to $\D_f$, which is the distribution induced by $f$ (over uniformly random inputs). It is easy to observe that for each $x$, $\D_f$ will output $x$ with probability at least about $2^{-\K^t(x)}$. In this case, we have that $\D_f$ \emph{dominates}\footnote{Recall that a distribution $\D$ dominates another distribution $\D'$ if $\D(x)\geq \D'(x)/\poly(n)$ for every $x$.} the distribution $Q^t$, which is defined as the (semi-)distribution that assigns each $x$ with probability mass $2^{-\K^t(x)}$. As a result, the average-case search algorithm that works for $\D_f$ also works for $Q^t$. This completes the description of the first step.

For the second step, we want to show that the same average-case search algorithm also works for any fixed-polynomial-time computable distribution $\D$. Again, it suffices to show that for a sufficiently large polynomial $t$, $Q^t$ dominates $\D$. In other words, for every $x$, $2^{-\K^t(x)}\gtrsim \D(x)$. Note that this essentially follows from the known coding theorem for polynomial-time computable distributions with respect to the measure $\K^\poly$ \cite{Levin86_siamcomp_journals}.

Next, we describe how to apply the above approach to obtain an average-case algorithm for finding an $(1/\ell)$-$\rK^t$-witness of $x$ while $x$ is sampled over a \emph{polynomial-time samplable} distribution. More specifically, we will do the following.
\begin{enumerate}
	\item Construct a function $f$ such that if $f$ can be inverted on average over uniformly random inputs, then one can obtain an average-case algorithm for finding an $(1/\ell)$-$\rK^t$-witnesses, over the distribution where each $x$ has probability mass $2^{-\rK^t(x)}$. 
	\item Show that such an average-case algorithm also works for any fixed-polynomial-time samplable distribution.
\end{enumerate}
We will need new ideas in both steps described above.

Our construction of the function $f$ is as follows:  $f$ takes an integer $i\in [n+O(1)]$, a \emph{randomized} program $\Pi\in\bool^{i}$, as well as a string $r\in\bool^t$, which will be used as the internal randomness for running $\Pi$. We then obtain $x$, which is the output of $\Pi$ running with randomness $r$ after $t$ steps. Now the key step here is that we will ensure that $\Pi$ is a random program that outputs $x$ with probability at least $2/3-1/\ell$. This is done using a randomized polynomial algorithm $V$ with the following property: With high probability, for every $(\Pi,x)$,
\begin{itemize}
		\item if within $t$ steps, $\Pi$ outputs $x$ with probability at least $2/3$, then $V$ accepts, and  
\item if within $t$ steps, $\Pi$ outputs $x$ with probability less than $2/3-1/\ell$, then $V$ rejects.
\end{itemize}
(For the sake of simplicity in this proof overview, think of $V$ as a deterministic algorithm.)
Finally, if $(\Pi,x)$ passes the test of $V$, we output $(i,x)$. Otherwise, we output $\bot$.

The idea is that for such a function $f$, if we invert an image $(i,x)$ successfully, then we will obtain a randomized program $\Pi$ along with some randomness $r$ such that $\Pi$, running on $r$ for $t$ steps, outputs $x$. Moreover, it holds that $(\Pi,x)$ passes the test of $V$, which means $\Pi$ outputs $x$ with probability at least $2/3-1/\ell$. Therefore, by finding the smallest $i^*$ such that $(i^*,x)$ is inverted successfully, we can obtain an $(1/\ell)$-$\rK^t$-witness of $x$.

Also, when considering the distribution $\D_f$ induced by $f$ (over uniformly random inputs), note that for every $x$, if a $\rK^t$-witness of $x$ is picked, which happens with probability at least $\frac{1}{O(n)}\cdot 2^{-\rK^t(x)}$, then after running $\Pi$ for $t$ steps, we will obtain $x$ with probability at least $2/3$. Moreover, $(\Pi,x)$ will pass the test of $V$. Therefore, for every $x$, $\D_f$ will output $x$ with probability at least about $2^{-\rK^t(x)}$. 

As discussed previously, this allows us to obtain an average-case search algorithm for $\mathsf{Search}$-$\mathsf{MINrKT}$ over the distribution $Q^t$, which assigns each $x$ with probability mass $2^{-\rK^t(x)}$. Now, to show that the same (average-case) search algorithm also works for any fixed-polynomial-time samplable distribution $\D$, it suffices to show that $Q^t$ dominates $\D$. Another key observation here is that to show the former, we do not necessarily need $Q^t$ to dominate $\D$ \emph{in the worst case}; it suffices for $Q^t$ to dominate $\D$ \emph{on average}. In other words, for almost all $x\sim \D$, $2^{-\rK^t(x)}\gtrsim \D(x)$. Then this follows from our average-case coding theorem for polynomial-time samplable distributions with respect to $\rK^{\poly}$ under the non-existence of one-way functions (\Cref{i:OWF-ae-noOWF} $\implies$ \Cref{i:OWF-ae-Coding} in \Cref{t:OWF-ae}).

\paragraph{Acknowledgements.} Zhenjian Lu received support from the UKRI Frontier Research Guarantee Grant EP/Y007999/1.

 \section{Preliminaries}
	\subsection{Notation}
	We use the notation $\varepsilon$ to represent an empty string.

    For a distribution $\D$ supported over $\bool^n\times\bool^n$ and $y\in\bool^n$, we let $\D(\cdot \mid y)$ denote the conditional distribution of $\D$ on the first half given that the second half is $y$.
	
	For $n,n'\in \N$ with $n\le n'$, let $[n:n'] = \{n,n+1,\ldots, n'\}$. Let $[n]:=[1:n] = \{1,\ldots,n\}$ for each $n\in\N$. 
	
	For a string $x\in\bin^*$ and $S\subseteq [|x|]$, let $x_{S}$ denote a substring of $x$ indicated by $S$, i.e., $x_S=x_{i_1}\ccc x_{i_k}$ for $S=\{i_1,\ldots,i_k\}$, where $i_1<\cdots <i_k$. Particularly, $x_{[i]} = x_1\ccc x_i$ and $x_{[i:j]} = x_i\circ x_{i+1}\ccc x_j$. For simplicity, let $x_{[0]}= \varepsilon$ for every $x\in\bin^*$.
	
	For a distribution $\mathcal{D}$ over $\bin^*$ and strings $x,y\in\bin^*$, we define $\mathcal{D}^*(x\mid y)\in [0,1]$ as \[\mathcal{D}^*(x\mid y):= \Prob_{z\sim \mathcal{D}}\mleft[z_{[|y|+1:|y|+|x|]}=x\middle| z_{[|y|]} = y\mright].\] 
	When $y=\varepsilon$, we drop ``$|y$'' from the notation above, i.e., 
	\[\mathcal{D}^*(x):= \Prob_{z\sim \mathcal{D}}\mleft[z_{[|x|]}=x\mright].\]
	
	For every distribution $\mathcal{D}$ over $\bin^*$, every $x\in \bin^*$, and $k\in\N$, we use the notation $\Next(\mathcal{D};x)$ to refer to the conditional distribution of the next bit of a subsequent string of $x$ selected according to $\mathcal{D}$. Namely, for each $b\in\bin$,
	\[\Prob_{b' \sim \Next(\mathcal{D};x)}[b'=b] = \mathcal{D}^*(b\mid x).\]

    When we consider a \emph{polynomial-time} decoding algorithm, the time bound is regarded as a function in the length of the original string before being encoded rather than the given encoding or input.  
	
	\subsection{Useful Tools}

 \begin{theorem}[Coding Theorem \cite{levin74laws}]\label{t:coding} Let $\mathcal{E}$ be a distribution whose cumulative distribution function can be computed by some program $p$. Then for every $x\in\supp(\mathcal{E})$,
	\[
	\K(x\mid p) \leq \log \frac{1}{\mathcal{E}(x)}+O(1).
	\]
\end{theorem}
	
	\begin{lemma}[See, e.g., {\cite[Lemma 9]{HiraharaILNO23_stoc_conf}}]\label{l:incompressible}
		There exists a universal constant $b>0$ such that for every distribution family $\{\mathcal{D}_n\}_{n \in \Nat}$ supported over $\bool^n\times \bool^n$, every $n\in\mathbb{N}$, and every $y\in\bool^n$,
		\[
		\Prob_{x\sim \mathcal{D}_n(\cdot\mid y)}\mleft[\mathsf{K}(x\mid y)< \log\frac{1}{\mathcal{D}_n(x\mid y)}-\alpha\mright]< \frac{n^b}{2^{\alpha}}.
		\]
	\end{lemma}

    	\begin{fact}\label{f:K_small_than_rK}
		For every $x\in\{0,1\}^*$ and $t \in \mathbb{N}$,
		\[
		\mathsf{K}(x)\leq \mathsf{rK}^t(x).
		\]
	\end{fact}
 
  	\begin{lemma}[Success Amplification for $\rK^t$]\label{l:success_amp}
		For any string $x\in\bool^*$, time bound $t\in\mathbb{N}$, and $q\in\mathbb{N}$, we have
		\[
		\rK_{1-1/q}^{t'}(x)\leq \rK^{t}(x) + O(\log \log q),
		\]
		where $t'\vcentcolon=t\cdot O(\log q)$.
	\end{lemma}

	    \begin{lemma}
		[{See, e.g., \cite[Lemma 6.14]{HiraharaN23_focs_conf}}]
		\label{lemma:random string is computationally shallow}
		For every polynomial-time samplable distribution family $\{\mathcal{D}_n\}_{n \in \Nat}$,
		there exists a polynomial $\rho$ such that 
		for every $n \in \Nat$, every $t \ge \rho(n)$, and every $\alpha \in \Nat$,
		\[
		\Prob_{x \sim \mathcal{D}_n} \mleft[ \mathsf{cd}^t(x) > \alpha \mright] \le 2^{-\alpha + O(\log n)}.
		\]
	\end{lemma}

    \begin{lemma}[Implicit in \cite{ImpagliazzoL90_focs_conf,ImpagliazzoL89_focs_conf}; see also \cite{HiraharaN23_focs_conf}]\label{t:next-bit-extrapolation}
        If almost everywhere \emph{(}resp. infinitely-often\emph{)} secure one-way functions do not exist, then for every samplable distribution family $\{\mathcal{D}_n\}_{n\in\N}$, where each $\mathcal{D}_n$ is over $\bin^{\ell(n)}$ for an efficiently computable $\ell(n)\le \poly(n)$, and for every polynomial $p$, there exists a polynomial-time randomized algorithm $\Ext$ such that for infinitely many \emph{(}resp. for all\emph{)} $n\in\N$,
        \[\Prob_{x\sim \mathcal{D}_n}\mleft[\forall i\in [n],\> \SD\mleft(\Ext\mleft(x_{[i-1]},1^{n}\mright), \Next(\mathcal{D}_n;x_{[i-1]})\mright)\le \frac{1}{p(n)}\mright]\ge 1-\frac{1}{p(n)},\]
        where $\SD(,)$ represents the total variation distance between two distributions.
    \end{lemma}

    Particularly, we obtain the following.
	
	\begin{theorem}\label{t:extrapolation}
		If almost everywhere \emph{(}resp. infinitely-often\emph{)} secure one-way functions do not exist, then for every samplable distribution family $\{\mathcal{D}_n\}_{n\in\N}$, where each $\mathcal{D}_n$ is over $\bin^n\times\bin^n$, there exists a probabilistic polynomial-time algorithm $\mathsf{Ext}$  such that for all $\varepsilon^{-1},\delta^{-1}\in\N$ and for infinitely many \emph{(}resp. for all\emph{)} $n\in\Nat$, 
		\[\Prob_{y\sim \mathcal{D}^{(2)}_n}\mleft[\SD\mleft(\mathsf{Ext}(y;1^{\varepsilon^{-1}},1^{\delta^{-1}}),\mathcal{D}_n(\cdot\mid y)\mright)\le \varepsilon\mright]\ge 1-\delta,\]
		where $\mathcal{D}_n^{(2)}$ denotes the marginal distribution of the second element of $\mathcal{D}_n$.
	\end{theorem}
\section{Coding for \texorpdfstring{$\rK^{\poly}$}{} Based on Next-bits Prediction}

In this section, we first present a meta-theorem showing that the approximation of next-bit probability yields a coding theorem for $\rK^{\poly}$ with an efficient encoder, where the encoding length can be worse than optimal depending on the number of \emph{light next-bits} defined below. Then, we prove \Cref{optimal coding-io,optimal coding-ae} as corollaries.

We introduce some notions to state the meta-theorem. First, we present the definition of next-bits prediction on a \emph{subset} of the support. For a distribution family $\mathcal{D}=\{\mathcal{D}_n\}_n$ and a subset $S\subseteq \supp(\mathcal{D})$, we use the notation $S_n$ to represent $S\cap\supp(\mathcal{D}_n)$ for each $n\in\N$ throughout the section.
\begin{definition}\label{def: nextbit approx}
	Let $\mathcal{D}=\{\mathcal{D}_n\}_n$ be a distribution family over $\bin^*$. For $S\subseteq \supp(\mathcal{D})$ and a polynomial $q$, the distribution $\mathcal{D}$ is said to be next-bits-predictable on $S$ with error parameter $q$ if there exists a randomized polynomial-time algorithm $P$ such that for every $n\in\N$, every $x\in S_n$, every $i\in [|x|]$, and every $b\in\bin$,
	\[\Prob_P\mleft[P(x_{[i-1]},b,1^n)\in [\mathcal{D}^*_n(b \mid x_{[i-1]})-1/q(n),\mathcal{D}^*_n(b \mid x_{[i-1]})+1/q(n)]\mright]\ge 1-1/q(n).\]
\end{definition}

Next, we introduce the key notion of light next-bits, which affects the bound on the length of encoding in the meta-theorem.
\begin{definition}[Light next-bit]\label{def: light next-bit}
For a distribution $\mathcal{D}$ over $\bin^*$, $\delta\in[0,1]$, $x\in\supp(\mathcal{D})$, and $i\in[|x|]$, we say that $b\in\bin$ is a $\delta$-light next-bit of $x_{[i-1]}$ (with respect to $\mathcal{D}$) if $\mathcal{D}^*(b\mid x_{[i-1]})\le \delta$. Moreover, we say that $x$ has a $\delta$-light next-bit if there exists $i\in[|x|]$ such that $x_i$ is a $\delta$-light next-bit of $x_{[i-1]}$.
\end{definition}

\begin{definition}[$\num_{\mathcal{D},\delta}$]
For a distribution $\mathcal{D}$ over $\bin^*$, $\delta\in[0,1]$, and $x\in\supp(\mathcal{D})$, we define $\num_{\mathcal{D},\delta}(x)$ as the number of $\delta$-light next-bits in $x$, i.e., 
\[\num_{\mathcal{D},\delta}(x) = |\{i:x_i\text{ is a $\delta$-light next-bit of $x_{[i-1]}$}\}|.\] 
For $j,j'\in[|x|]$ with $j<j'$, we also define $\num^{j,j'}_{\mathcal{D},\delta}(x)$ as
\[\num^{j,j'}_{\mathcal{D},\delta}(x) = |\{i\in[j,j']:x_i\text{ is a $\delta$-light next-bit of $x_{[i-1]}$}\}|.\] 
\end{definition}
We often omit the subscript ``$\mathcal{D}$'' from $\num_{\mathcal{D},\delta}$ and $\num^{j,j'}_{\mathcal{D},\delta}$ when $\mathcal{D}$ is trivially identified in context. 

The following property immediately follows from the definition.
\begin{proposition}
Let $\mathcal{D}$ be a distribution over $\bin^*$. For any $x\in\supp(\mathcal{D})$, the following hold:
\begin{itemize}
    \item For every $i,j\in[|x|]$ with $i<j$ and every $\delta,\delta'\in [0,1]$ with $\delta\le \delta'$, $\num^{i,j}_\delta(x)\le \num^{i,j}_{\delta'}(x)$.
    \item For every $\delta\in[0,1]$ and every $i,j,i'j'\in [|x|]$ such that $[i:j]\subseteq [i':j']$,  $\num^{i,j}_\delta(x)\le \num^{i',j'}_{\delta}(x)$.
\end{itemize}
\end{proposition}

A sample $x$ has no $\delta$-light next-bit if and only if $\num_{\delta}(x) = 0$. It is easily observed that $x$ has no light next-bit with respect to a distribution $\mathcal{D}$ with high probability over the choice of $x\sim \mathcal{D}$.

\begin{proposition}\label{prop: large conditional}
For every $n\in\N$, $\delta\in[0,1]$ and every distribution $\mathcal{D}$ over $\bin^n$,
\[\Prob_{x\sim \mathcal{D}}\mleft[x \textnormal{ has a $\delta$-light next-bit with respect to }\mathcal{D}\mright]\le n\delta.\]
\end{proposition}

\begin{proof}
Sampling according to $\mathcal{D}$	can be performed sequentially as $x_i\sim \Next(\mathcal{D};x_{[i-1]})$ for $i=1,\ldots,n$ (in this order). For each $i$, the probability that a $\delta$-light next-bit is sampled is at most $\delta$ by the definition. Therefore, by the union bound, the probability that the sample has a $\delta$-light next-bit is at most $n\cdot \delta$.
\end{proof}

Now, we formally state the meta-theorem.

\begin{theorem}\label{thm: nbp to coding}
	For any distribution family $\mathcal{D}=\{\mathcal{D}_n\}$, where $\mathcal{D}_n$ is over $\bin^{\ell(n)}$ for an efficiently computable function $\ell(n)\le\poly(n)$, and any polynomial $q$, there exists a polynomial $p$ such that if $\mathcal{D}$ is next-bits-predictable on $S\subseteq\supp(\mathcal{D})$ with error parameter $p$, then for every $n\in\N$, every $x\in S_n$, and every $i\in[|x|]$ 
\[\rK^{p(n)}(x_{[i:\ell(n)]}\mid x_{[i-1]})\le -\log \mathcal{D}^*_n(x_{[i:\ell(n)]}\mid x_{[i-1]}) + \num^{i,\ell(n)}_{1/q(n)}(x)\cdot O(\log \ell(n)) + O(\log n\ell(n)q(n)),\]
where the hidden constants in $O(\cdot)$ are independent of $\ell$ and $q$. In particular,
\[\rK^{p(n)}(x)\le -\log \mathcal{D}_n(x) + \num_{1/q(n)}(x)\cdot O(\log \ell(n)) + O(\log n\ell(n)q(n)).\]
Moreover, there exists a polynomial-time randomized algorithm $\Enc$ that outputs, for given input $(x,1^n,i)$, a description of a polynomial-time \emph{(}randomized\emph{)} program that outputs $x_{[i:\ell(n)]}$ when $x_{[i-1]}$ is given, and the description length satisfies the upper bound on $\rK^{p(n)}(x_{[i:\ell(n)]}\mid x_{[i-1]})$, where the success probability of $\Enc$ is at least $1-1/q(n)$.
\end{theorem}

In \Cref{sec: proof of wst coding}, we present the proof of \Cref{thm: nbp to coding}. In \Cref{sec:cor-to-meta-theorem}, we derive \Cref{optimal coding-io,optimal coding-ae} and the \emph{almost} optimal worst-case coding theorem and \emph{optimal} (average-case) coding theorem for next-bits-predictable source from \Cref{thm: nbp to coding}.

\subsection{Proof of \texorpdfstring{\Cref{thm: nbp to coding}}{}}\label{sec: proof of wst coding}

First, we make a next-bits predictor pseudo-deterministic with the help of \emph{short advice} that depends on the input. Note that the auxiliary advice can be selected with high probability from uniformly random sampling but is required to be selected for each encoded string.    

\begin{lemma}\label{next-bit predictor}
For every distribution family $\mathcal{D}=\{\mathcal{D}_n\}$, where $\mathcal{D}_n$ is over $\bin^{\ell(n)}$ for an efficiently computable function $\ell(n)\le\poly(n)$, and for every polynomial $q$, if $\mathcal{D}$ is next-bits-predictable on $S\subseteq\supp(\mathcal{D})$ with error parameter $32\ell(n)q(n)^3$, then there exists a polynomial-time randomized algorithm $\tilde{P}$ such that for every $n\in\N$ and every $x\in S_n$, with probability at least $1-1/q(n)$ over the choices of advice $\alpha_{x}\sim \{0,1,\ldots, 2\ell(n)q(n)-1\}\subseteq \N$ and $2\ell(n)$ independent random seeds $r_1,r'_1,\ldots,r_{\ell(n)},r'_{\ell(n)}$ for $\tilde{P}$, the following properties hold for all $i\in [\ell(n)]$ and all $b\in\bin$,
\begin{enumerate}
	\item\emph{(}$\tilde{P}$ is pseudo-deterministic.\emph{)} \begin{equation}\label{property 1}
		\tilde{P}(x_{[i-1]},b,\alpha_{x},1^{n};r_i) = \tilde{P}(x_{[i-1]},b,\alpha_{x},1^{n};r'_i);
	\end{equation}
	\item\emph{(}$\tilde{P}$ determines a distribution.\emph{)} \begin{equation}\label{property 2}\tilde{P}(x_{[i-1]},0,\alpha_{x},1^{n};r_i) + \tilde{P}(x_{[i-1]},1,\alpha_{x},1^{n};r'_i) = 1;\end{equation}
	\item\emph{(}$\tilde{P}$ is accurate.\emph{)} 
    \begin{equation}\label{property 3}\mathcal{D}^*_n(b\mid x_{[i-1]})-1/q(n)^2\le \tilde{P}(x_{[i-1]},b,\alpha_{x},1^{n};r_i)\le \mathcal{D}^*_n(b\mid x_{[i-1]})+1/q(n)^2.\end{equation}
    Particularly, if $b$ is not an $(1/q(n))$-light next-bit of $x_{[i-1]}$, then
	\begin{equation*}(1-1/q(n))\mathcal{D}^*_n(b\mid x_{[i-1]})\le \tilde{P}(x_{[i-1]},b,\alpha_{x},1^{n};r_i)\le (1+1/q(n))\mathcal{D}^*_n(b\mid x_{[i-1]}).\end{equation*}
\end{enumerate}
\end{lemma}
For convenience, we call the polynomial $q$ above a \emph{mordified error parameter}.

\begin{proof}
	Since $\mathcal{D}$ is next-bits-predictable on $S$, there exists a polynomial-time next-bits predictor $P$ satisfying the properties of \Cref{def: nextbit approx} with error parameter $q'(n):=32\ell(n)q(n)^3$.
	
	For given input $(x_{[i-1]},b,\alpha_{x},1^{n})$, where $\alpha_{x}\sim\{0,1,\ldots, 2\ell(n)q(n)-1\}\subseteq \N$, the algorithm $\tilde{P}$ first executes $v_0\gets P\mleft(x_{[i-1]},0,1^n\mright)$. Then, $\tilde{P}$ applies the technique of adding noise and rounding to be pseudo-deterministic. Here, the amount of noise is $\alpha_{x,\delta}\cdot 1/(8\ell(n)q(n)^3)$. Let $\tilde{v_0}$ be the nearest value to $v_0 + \alpha_{x}\cdot 1/(8\ell(n)q(n)^3)$ in multiples of $1/(4q(n)^2)$ in $[0,1]$ (i.e., $\tilde{v_0} = N\cdot 1/(4q(n)^2)$) for some $N\in\{0,1,\ldots,4q(n)^{2}\}$), where ties are broken by choosing the smaller one. 
	
	If the input $b$ is $0$, the algorithm $\tilde{P}$ outputs $\tilde{v_0}$; otherwise (i.e., if $b=1$), $\tilde{P}$ outputs $\tilde{v_1}:=1-\tilde{v_0}$. Note that $\tilde{P}$ uses its internal randomness only for executing $P$. 
	
	We show the three properties in the lemma. Recall that $\alpha_{x}\sim \{0,\ldots, 2\ell(n)q(n)-1\}$. For each $i\in[\ell(n)]$, we consider the execution of $\tilde{P}(x_{[i-1]},b,\alpha_{x},1^{n})$ with the global advice $\alpha_{x}$. Let $v_0$ be the value produced by $P\mleft(x_{[i-1]},0,1^n\mright)$ during the execution. 
	
	Suppose that $P$ does not fail in the sense that 
	\[\mleft|v_0-\mathcal{D}^*_n(0\mid x_{[i-1]})\mright|\le 1/q'(n) = 1/(32\ell(n)q(n)^3).\]
	Recall that the property of $P$ shows that the event above occurs with probability at least $1-1/q'(n)$ as long as $x\in S_n$.
	
	We define $I_{x,i}\subseteq[0,1]$ as 
	\[I_{x,i}:= [\mathcal{D}^*_n(0\mid x_{[i-1]})-1/(32\ell(n)q(n)^3),\mathcal{D}^*_n(0\mid x_{[i-1]})+1/(32\ell(n)q(n)^3)] \cap [0,1].\] 
	Then, $v_0\in I_{x,i}$ as long as $P$ is performed successfully. Note that $|I_{x,i}|\le 1/(16\ell(n)q(n)^3)$, and $I_{x,i}$ is independent of $\alpha_{x}$.
	
	Notice that (i) the noise $\alpha_{x}\cdot 1/(8\ell(n)q(n)^3)$ varies in $1/(8\ell(n)q(n)^3) \,(> |I_{x,i}|)$ increments, and (ii) the rounded value $\tilde{v_0}$ varies in $1/(4q(n)^2) \,(> (2\ell(n)q(n)-1)\cdot 1/(8\ell(n)q(n)^3))$ increments. Thus, the number of $\alpha_{x}$ for which there exist two values $v_0,v_0' \in I_{x,i}$ that are rounded to two distinct values is at most $1$.
	
	Since the same argument holds for all $i\in[\ell(n)]$, we have that with probability at least $1-\ell(n)/(2\ell(n)q(n))=1-1/(2q(n))$ over the choice of $\alpha_{x}$, for all $i\in[\ell(n)]$, the rounded value $\tilde{v_0}$ for given $x_{[i-1]}$ always takes the same value as long as $P$ is successfully executed. Below, we observe the three properties in the lemma under the events that (i) such a good value of $\alpha_{x}$ is selected, and (ii) all of the $2\ell(n)$ executions of $P$ are successfully performed. This completes the proof of the lemma because, by the union bound, the two events occur simultaneously with probability at least $1-1/(2q(n))-2\ell(n)/q'(n)\ge 1-1/q(n)$.
	
	The first property (\Cref{property 1}) has already verified because the rounded value $\tilde{v_0}$ always takes the same value as long as $P$ is successfully executed, and $\tilde{P}$ outputs either of $\tilde{v_0}$ and $\tilde{v_1}=1-\tilde{v_0}$ depending on $b$.  
	
	Next, we observe the second property. Let $\tilde{v_0}$ and $\tilde{v_1}$ be the values produced by $\tilde{P}$ given $b=0$ and $b=1$, respectively, for fixed randomness $r$. By the construction of $\tilde{P}$, they always satisfy $\tilde{v_0}+\tilde{v_1} = 1$. Let $\tilde{v_0}'$ and $\tilde{v_1}'$ be the values produced by $\tilde{P}$ given $b=0$ and $b=1$, respectively, for fixed randomness $r'$ different from $r$. The first property implies that $\tilde{v_0}=\tilde{v_0}'$, and it further implies that $\tilde{v_0}'+\tilde{v_1} = \tilde{v_0}+\tilde{v_1} = 1$. Therefore, \Cref{property 2} holds under the same events.
	
	Finally, we observe the third property and complete the proof. Recall that $\mleft|v_0-\mathcal{D}^*_n(0\mid x_{[i-1]})\mright|\le 1/(32\ell(n)q(n)^3)$ under the condition. Adding the noise $(\alpha_{x,\delta}\cdot 1/(8\ell(n)q(n)^3))$ and rounding to the multiples of $1/(4q(n)^2)$ only changes the value at most 
	\[2\ell(n)q(n)\cdot \frac{1}{8\ell(n)q(n)^3}+\frac{1}{4q(n)^2} = \frac{1}{2q(n)^2}.\]
	Thus, we have that $\mleft|\tilde{v_0}-\mathcal{D}^*_n(0\mid x_{[i-1]})\mright|\le 1/(32\ell(n)q(n)^3)+ 1/(2q(n)^2) \le 1/q(n)^2$. Notice that \[\mleft|\tilde{v_1}-\mathcal{D}^*_n(1\mid x_{[i-1]})\mright| = \mleft|\tilde{v_0}-\mathcal{D}^*_n(0\mid x_{[i-1]})\mright|\le 1/q(n)^2.\] 
	
	Therefore, in any case,
	\[\mleft|\tilde{v_b}-\mathcal{D}^*_n(b\mid x_{[i-1]})\mright|\le  \frac{1}{q(n)^2},\]
 and \Cref{property 3} holds because $P$ outputs $\tilde{v_b}$ for the given $b$.
 In addition, if $b$ is not an $(1/q(n))$-light next-bit of $x_{[i-1]}$, then
 \[\mleft|\tilde{v_b}-\mathcal{D}^*_n(b\mid x_{[i-1]})\mright|\le  \frac{1}{q(n)^2}\le \frac{1}{q(n)}\mathcal{D}^*_n(b\mid x_{[i-1]}).\]
 By rearranging the above,
	\[(1-1/q(n))\mathcal{D}^*_n(b\mid x_{[i-1]})\le \tilde{v_b} \le (1+1/q(n))\mathcal{D}^*_n(b\mid x_{[i-1]}),\]
 as desired.
\end{proof}

Next, we use the \emph{modified} next-bits predictor $\tilde{P}$ for the arithmetic encoding (i.e., Shannon--Fano--Elias coding) to obtain the coding theorem for $\rK^{\poly}$. 

Let $\mathcal{D}=\{\mathcal{D}_n\}$ be a distribution family that is next-bits-predictable on $S\subseteq \supp(\mathcal{D})$, where $\mathcal{D}_n$ is over $\bin^{\ell(n)}$. The encoding and decoding algorithms are the following, where $q(n)$ represents an arbitrary polynomial and $\tilde{P}$ represents the algorithm in \Cref{next-bit predictor} with \emph{modified} error parameter $q'(n):=\ell(n)q(n)+1$. We only consider the encoding of $x_{[k:\ell(n)]}$ given $x_{[k-1]}$ for $x\in S_n$ and $k\in[\ell(n)]$. 

Note that, at the end of each round $i$, the values of the variables $p_<$ and $p_=$ in the encoding and decoding algorithms are excepted to be the approximations of $\sum_{y< x_{[k:i]}}\mathcal{D}^*(y\mid x_{[k-1]})$ and $\mathcal{D}^*(x_{[k:i]}\mid x_{[k-1]})$, respectively. However, the algorithms ignore the round $i$ when $x_i$ is a light next-bit of $x_{i-1}$ (i.e., the next-bit probability of $x_i$ is regarded to be $1$) and embed $x_i$ to the encoding with the position $i$. The way of update is based on the following expressions:
\[\sum_{y< x_{[k:i]}}\mathcal{D}^*(y\mid x_{[k-1]}) = \begin{cases}
    \sum_{y< x_{[k:i-1]}}\mathcal{D}^*(y\mid x_{[k-1]}) &\text{if $x_i=0$}\\
    \sum_{y< x_{[k:i-1]}}\mathcal{D}^*(y\mid x_{[k-1]}) + \mathcal{D}^*(x_{[k:i-1]}\mid x_{[k-1]})\cdot\mathcal{D}^*(0\mid x_{[i-1]}) &\text{if $x_i=1$,}
\end{cases}\]
and
\[\mathcal{D}^*(x_{[k:i]}\mid x_{[k-1]})= \mathcal{D}^*(x_{[k:i-1]}\mid x_{[k-1]})\cdot \mathcal{D}^*(x_i\mid x_{[i-1]}).\]

\begin{algorithm}
	\caption{Enc$_q(x_{[k:\ell(n)]},n;x_{[k-1]})$}
	\SetKwInOut{Input}{Input}
	\Input{$x_{[k:\ell(n)]}\in\bin^{\ell(n)-k}$, $n\in\N$, and a conditional string $x_{[k-1]}\in\bin^{k-1}$, where $x\in S_n$.}
	Let $p_{<} := 0$ and $p_= := 1$\;
	Select $\alpha\sim \{0,\ldots, 2\ell(n)q'(n)-1\}$ uniformly at random\;
	Let $L$ be an empty list (where the element is expected to be in $[\ell(n)]\times\bin$)\;
	\For{$i:= k$ to $\ell(n)$}{
        Execute $q_i\gets \tilde{P}(x_{[i-1]},x_i,\alpha,1^{n})$\;
		\If{$q_i\le 2/q'(n)$}{\label{line:hev-comp} Add $(i,x_i)\in [\ell(n)]\times \bin $ to $L$ and go to the next loop\;}
		\lIf{$x_i=1$}{$p_< := p_< + p_= \cdot \tilde{P}(x_{[i-1]},0,\alpha,1^{n})$}
		$p_= := p_= \cdot q_i$\;
	}
	Let $v$ be the first $\lceil-\log p_=\rceil +1$ bits of $p_< + (p_=/2)$\;
	\textbf{return } $(v,L,\alpha,n,k)$\;
\end{algorithm}

\begin{algorithm}
	\caption{Dec$_q(v,L,\alpha,n;x_{[k-1]})$}
	\SetKwInOut{Input}{Input}
	\Input{an encoding $(v,\alpha_{x,\delta},n,\delta)$ and a conditional string $x_{[k-1]}\in\bin^{k-1}$.}
	Let $p_{<} := 0$ and $p_= := 1$\;
	Let $\tilde{x}:=x_{[k-1]}$\;
	\For{$i:= k$ to $\ell(n)$}{
		\If{$(i,b)\in L$ for some $b\in\bin$}
		{Update $\tilde{x}:=\tilde{x}\circ b$ and go to the next loop\;}
		Compute $q_0 = \tilde{P}(\tilde{x}_{[i-1]},0,\alpha,1^{n})$ and $q_1=1-q_0$\;
		\lIf{$v\ge p_< + p_=\cdot q_0$}{let $\tilde{x}:=\tilde{x}\circ 1$, $p_<:=p_<+p_=\cdot q_0$, and $p_= = p_= \cdot q_1$}
		\lElse{$\tilde{x}:=\tilde{x}\circ 0$ and $p_=:=p_=\cdot q_0$}
	}
	\textbf{return } $\tilde{x}_{[k:\ell(n)]}$\;
\end{algorithm}

It is easily observed that $\Enc_q$ and $\Dec_q$ halt in polynomial time in $n$ since $\tilde{P}$ is polynomial time. We show that the encoding and decoding algorithms above work with high probability over the choice of \emph{independent} random seeds and estimate the length of the encoding.

\begin{lemma}\label{lemma: arithmetic encoding}
For every distribution family $\mathcal{D}=\{\mathcal{D}_n\}$, where $\mathcal{D}_n$ is over $\bin^{\ell(n)}$ for an efficiently computable function $\ell(n)\le\poly(n)$, and every polynomial $q$, if $\mathcal{D}$ is next-bits-predictable on $S\subseteq\supp(\mathcal{D})$ with error parameter $32\ell(n)q'(n)^3 \,(= 32\ell(n)^4q(n)^3)$, then for every $n\in\N$, every $x\in S_n$, and every $k\in[|x|]$, it holds that
\[\Prob_{\Enc_q,\Dec_q}\mleft[\Dec_q(\Enc_q(x_{[k:\ell(n)]},n;x_{[k-1]});x_{[k-1]})=x_{[k:\ell(n)]}\mright]\ge 1-\frac{1}{q(n)\ell(n)};\]
and
\begin{align*}
&\quad\Prob_{\Enc_q}\mleft[|\Enc_q(x_{[k:\ell(n)]},n;x_{[k-1]})|\le -\log \mathcal{D}_n(x)+\num^{k,\ell(n)}_{1/q(n)}(x)\cdot O(\log \ell(n)) + O(\log n\ell(n)q(n))\mright]\\
&\ge 1-\frac{1}{q(n)\ell(n)},
\end{align*}
where the hidden constants in $O(\cdot)$ are universal independent of $\ell$ and $q$.

In particular, for a sufficiently large polynomial $p$,
\[\rK^{p(n)}(x_{[k:\ell(n)]}\mid x_{[k-1]})\le -\log \mathcal{D}_n(x)+\num^{k,\ell(n)}_{1/q(n)}(x)\cdot O(\log \ell(n)) + O(\log n\ell(n)q(n)),\]
and there exists a polynomial-time randomized algorithm $\Enc$ that outputs, for given input $(x,1^n,k)$, a description of a polynomial-time \emph{(}randomized\emph{)} program that outputs $x_{[k:\ell(n)]}$ when $x_{[k-1]}$ is given, and the description length satisfies the upper bound on $\rK^{p(n)}(x_{[k:\ell(n)]}\mid x_{[k-1]})$ above, where the success probability of $\Enc$ is at least $1-6/(q(n)\ell(n))$.
\end{lemma}

Note that \Cref{thm: nbp to coding} immediately follows from \Cref{lemma: arithmetic encoding}.

\begin{proof}
	\Cref{lemma: arithmetic encoding} follows from the correctness of the arithmetic encoding and \Cref{next-bit predictor}. 
	
	First, we observe that, in the execution of $\Dec_q(\Enc_q(x_{[k:\ell(n)]},n;x_{[k-1]});x_{[k-1]})$, the next-bits predictor $\tilde{P}$ is executed only on $(x_{[i-1]},b,\alpha,1^{n})$ for $i\in[\ell(n)]$ and $b\in\bin$ as long as the sequential decoding is performed along $x$ (i.e., $\tilde{x} = x_{[i]}$ for each stage $i$). Note that $\Enc_q$ and $\Dec_q$ do not share the randomness for executing $\tilde{P}$; thus, they may execute $\tilde{P}$ on the same input but using independent randomness. However, \Cref{next-bit predictor} shows that, with probability at least $1-1/(q(n)\ell(n))$ over the choice of $\alpha$ and randomness for $\tilde{P}$ (i.e., over the randomness for $\Enc_q$ and $\Dec_q$), all the executions of $\tilde{P}$ yield consistent values and determine the conditional distribution of each next bit. In this case, $\Dec_q(\Enc_q(x_{[k:\ell(n)]},n;x_{[k-1]});x_{[k-1]})$ performs the arithmetic encoding~\cite[cf.][Sections~5.9 and~13.3]{0016881_daglib_books} for the distribution induced by executing $\tilde{P}$ on each prefix of $x$ except for the positions $i$ on which the value of $q_i$ is less than $2/q'(n)$. Namely, under the good choices of $\alpha$ and randomness as indicated in \Cref{next-bit predictor}, the value $v$ produced by $\Enc_{q}$ satisfies that
    \[p_<^{\Enc} \le p_<^{\Enc}+p_=^{\Enc}/2 - 2^{-\lceil-\log p_=^{\Enc}\rceil -1} < v <p_<^{\Enc}+p_=^{\Enc},\]
    where $p_<^{\Enc}$ and $p_=^{\Enc}$ represent the values of variables $p_<$ and $p_=$ at the end of the execution of $\Enc_q$, respectively. In addition, for every round $i$ in the execution of $\Dec_q$,
 \[\begin{cases}
     v < p_<^{\Enc}+p_=^{\Enc} \le p_< +p_=\cdot q_0 &\text{if $x_i=0$}\\
     v \ge p_<^{\Enc} \ge p_< +p_=\cdot q_0 &\text{if $x_i=1$},
 \end{cases}\]
 where $p_<,p_=$, and $q_0$ represent the variables computed in $\Dec_q$. Thus, whenever $q_i>2/q'(n)$, $\Dec_q$ successfully decodes the $i$-th next bit. In the other case where $q_i\le 2/q'(n)$, the pair $(i,x_i)$ is contained in the list $L$, and $\Dec_q$ also successfully decodes the next bit. Therefore, we obtain that
	\[\Prob_{\Enc_q,\Dec_q}\mleft[\Dec_q(\Enc_q(x_{[k:\ell(n)]},n;x_{[k-1]});x_{[k-1]})=x_{[k:\ell(n)]}\mright]\ge 1-\frac{1}{q(n)\ell(n)}.\]
	
	Next, we evaluate the length of the encoding. Below we let $p_=$ represent the value of the variable $p_=$ at the end of the execution of $\Enc_q$. We also assume that the randomness for $\Enc_q$ (i.e., $\alpha$ and randomness for executing $\tilde{P}$) satisfies the condition of \Cref{next-bit predictor}. This event occurs with probability at least $1-1/(\ell(n)q(n))$.

  Let $L$ be the set of indices $i$ such that $q_i\le 2/q'(n)$ in the execution of $\Enc_q$, and let $H=[\ell(n)]\setminus L$. Recall that the variable $p_=$ is updated only when $i\in H$. We observe that for every $i\in L$, the $i$-th bit $x_i$ is a $3/q'(n)$-light next bit of $x_{[i-1]}$ because
  \[\mathcal{D}^*(x_i\mid x_{[i-1]})\le q_i + \frac{1}{q'(n)^2}\le \frac{2}{q'(n)} + \frac{1}{q'(n)^2} \le \frac{3}{q'(n)}.\]
  In addition, for every $i\in H$, the $i$-th bit $x_i$ is not an $(1/q'(n))$-light next bit of $x_{[i-1]}$ because
  \[\mathcal{D}^*(x_i\mid x_{[i-1]})\ge q_i - \frac{1}{q'(n)^2}> \frac{2}{q'(n)} - \frac{1}{q'(n)^2} \ge \frac{1}{q'(n)}.\]

 Without loss of generality, we assume that  $(\lceil-\log p_=\rceil +1) = O(\ell(n))$; otherwise, we can replace the encoding for $x$ with the canonical encoding of length $\ell(n)+O(1)$ by embedding $x$. In addition, we assume that $\ell(n)\ge 3$. Then, by the standard prefix-free encoding, the output $(v,L,\alpha,n,k)$ of $\Enc_q$ is represented in 
\begin{multline*}
	\lceil-\log p_=\rceil +1 +|H|\cdot O(\log \ell(n)) +O(\log(\lceil-\log p_=\rceil +1))+O(\log nq(n)\ell(n))\\
	=-\log p_=  + \num^{k,\ell(n)}_{3/q'(n)}(x) \cdot O(\log \ell(n))+ O(\log nq(n)\ell(n))\\
	\le  -\log p_= + \num^{k,\ell(n)}_{1/q(n)}(x)\cdot O(\log \ell(n)) + O(\log n\ell(n)q(n))\qquad \text{bits,} 
\end{multline*}
where all constants in $O(\cdot)$ notations are independent of $q$ and $\ell$, and we used 
\[
3/q'(n) \le 3/(\ell(n)q(n))\le 1/q(n) \text{\quad and\quad}\num^{k,\ell(n)}_{3/q'(n)}(x) \le \num^{k,\ell(n)}_{1/q(n)}(x).
\]

 \Cref{next-bit predictor} further shows that for every $i\in H$,
	\[(1-1/q'(n))\mathcal{D}^*_n(x_i\mid x_{[i-1]})\le \tilde{P}(x_{[i-1]},b,\alpha,1^{n},1^{\delta^{-1}})\le (1+1/q'(n))\mathcal{D}^*_n(x_i\mid x_{[i-1]})\]
 since $x_i$ is not an $(1/q'(n))$-light next bit of $x_{[i-1]}$ whenever $i\in H$.
	
Therefore, at the end of the execution of $\Enc_q$, the variable $p_=$ takes the value 
	\begin{align*}
		p_= &= \prod_{i\in H} \tilde{P}(x_{[i-1]},x_i,\alpha,1^{n}) \\
		&\ge  (1-1/(\ell(n)q(n)+1))^{|H|}\cdot \prod_{i\in H}\mathcal{D}_n^*(x_i\mid x_{[i-1]})\\
        &\ge  (1-1/(\ell(n)q(n)+1))^{\ell(n)}\cdot \prod_{i=1}^{\ell(n)}\mathcal{D}_n^*(x_i\mid x_{[i-1]})\\
		&= (1-1/(\ell(n)q(n)+1))^{\ell(n)}\cdot \mathcal{D}(x)\\
        &\ge e^{-1/q(n)}\cdot \mathcal{D}(x).
	\end{align*}
	Thus, we have 
	\begin{align*}
		-\log p_{=} \le -\log\mathcal{D}(x) +1/q(n) + \log e\le -\log\mathcal{D}(x) +3.
	\end{align*}
	Therefore, the length of the encoding is at most
	 \begin{multline*}
	     -\log p_= + \num^{k,\ell(n)}_{1/q(n)}(x)\cdot C\log \ell(n) + C\cdot\log n\ell(n)q(n) \\\le -\log\mathcal{D}(x) + \num^{k,\ell(n)}_{1/q(n)}(x)\cdot C \log\ell(n) + C\log n\ell(n)q(n) +3,
	 \end{multline*}
  where $C$ is a universal constant independent of $\ell$ and $q$.
	 
	The final statement on $\rK$ is based on the following probabilistic argument: By the union bound, with probability at least $1-2/(\ell(n)q(n))$ over the choice of randomness for $\Enc_q$ and $\Dec_q$, it holds that (i) the length of the output of $\Enc_q$ is at most $-\log\mathcal{D}(x) + \num^{k,\ell(n)}_{1/q(n)}(x)\cdot O(\log\ell(n)) + O(\log n\ell(n)q(n))$, and (ii) $\Dec_q(\Enc_q(x_{[k:\ell(n)]},n;x_{[k-1]});x_{[k-1]})=x_{[k:\ell(n)]}$. By Markov's inequality, with probability at least $1-6/(\ell(n)q(n))$ over the randomness for $\Enc_q$, (i) the length of the encoding satisfies the same bound, and (ii) $\Dec_q(\Enc_q(x_{[k:\ell(n)]},n;x_{[k-1]});x_{[k-1]})$ produces $x_{[k:\ell(n)]}$ with probability at least $2/3$ over the randomness for $\Dec_q$. Since $\Enc_q$ and $\Dec_q$ are polynomial-time algorithms and uniform, the statement on $\rK$ holds.
\end{proof}

\subsection{Corollaries}\label{sec:cor-to-meta-theorem}
We derive \Cref{optimal coding-io,optimal coding-ae} from \Cref{thm: nbp to coding} in \Cref{sec:condit-coding} and show \Cref{t:char-owf} in \Cref{sec:condit-coding}. In \Cref{section: coding for rand nbps}, we show the almost optimal \emph{worst-case} coding theorem and optimal (average-case) coding theorem for next-bits-predictable distributions. In \Cref{sec:uniform-HMS}, we show the uniform variant of the result of \cite{HaitnerMS23_innovations_conf} (stated in the nonuniform model or the uniform model with shared randomness) at the expense of arbitrarily small multiplicative error.

\subsubsection{Proofs of Conditional Coding Theorems}\label{sec:condit-coding}
First, we observe that, under the non-existence of one-way functions, every samplable distribution becomes next-bits predictable on a subset of large weight.  

\begin{lemma}\label{lemma:ae-ext-to-nba}
 If there is no almost everywhere one-way function, then for every samplable distribution $\mathcal{D}=\{\mathcal{D}_n\}$ and every polynomials $q(n)$ and $s(n)$, there exists a subset $S\subseteq \supp(\mathcal{D})$ such that \emph{(}i\emph{)} $\mathcal{D}$ is next-bits-predictable on $S$ with error parameter $q$ and \emph{(}ii\emph{)} for infinitely many $n\in\N$, $\Prob_{x\sim \mathcal{D}_n}[x\in S_n]\ge 1-1/s(n)$. 
\end{lemma}
\begin{proof}
    The lemma follows from \Cref{t:next-bit-extrapolation} for almost everywhere one-way functions and the standard empirical estimation of the probability that the algorithm $\Ext$ outputs each bit.
\end{proof}

We also obtain the lemma for the infinitely-often security in the same way. 

\begin{lemma}\label{lemma:io-ext-to-nba}
 If there is no infinitely-often one-way function, then for every samplable distribution $\mathcal{D}=\{\mathcal{D}_n\}$ and every polynomials $q(n)$ and $s(n)$, there exists a subset $S\subseteq \supp(\mathcal{D})$ such that \emph{(}i\emph{)} $\mathcal{D}$ is next-bits-predictable on $S$ with error parameter $q$ and \emph{(}ii\emph{)} for all $n\in\N$, $\Prob_{x\sim \mathcal{D}_n}[x\in S_n]\ge 1-1/s(n)$. 
\end{lemma}
\begin{proof}
    The proof is the same as \Cref{lemma:ae-ext-to-nba} except we use \Cref{t:next-bit-extrapolation} for infinitely-often one-way functions.
\end{proof}

Now, we prove (\Cref{i:OWF-io-noOWF} $\implies$ \Cref{i:OWF-io-efficient}) in \Cref{t:OWF-io}, which is restated as follows.
\begin{theorem}\label{optimal coding-io}
	If there is no one-way function, then for every samplable distribution $\mathcal{D}=\{\mathcal{D}_n\}$ supported over $\bool^n\times\bool^n$ and every polynomial $q$, there exists a polynomial $p$ such that for infinitely many $n\in\N$, 
	\[\Prob_{(x,y)\sim\mathcal{D}_n}\mleft[\rK^{p(n)}(x\mid y)\le \log \frac{1}{\mathcal{D}_n(x\mid y)}  + \log p(n) \mright]\ge1-\frac{1}{q(n)}.\]
    Moreover, there exists an efficient algorithm $\Enc$ that outputs, for given $(x,y)\sim \mathcal{D}_n$, a description of a $p(n)$-time program $\Pi$ of length at most $-\log\D_n(x\mid y)+\log p(n)$ with probability at least $1-1/q(n)$ over the choice of $(x,y)\sim \mathcal{D}_n$ and randomness for $\Enc$, such that $\Pi$ outputs $x$ for given $y$ and randomness $r\sim\bin^{p(n)}$ with probability at least $2/3$ over the choice of $r$.
\end{theorem}

\begin{proof}
    Let $\mathcal{D}=\{\mathcal{D}_n\}$ be an arbitrary samplable distribution, where each $\mathcal{D}_n$ is supported over $\bin^n\times\bin^n$. Let $\bar{\mathcal{D}}=\{\bar{\mathcal{D}}_n\}$ be another samplable distribution defined as $\bar{\mathcal{D}}_n\equiv\mathcal{D}^{(2)}_n\circ \mathcal{D}^{(1)}_n$. Namely, $\bar{\mathcal{D}}_n$ is distributed over $\bin^{2n}$.

    Let $q$ be an arbitrary polynomial. We apply \Cref{thm: nbp to coding} for $\bar{\mathcal{D}}$ and a polynomial $4q(n)n$. Then, there exists a polynomial $p$ such that if $\bar{\mathcal{D}}$ is next-bits-predictable on $S\subseteq\supp(\bar{\mathcal{D}})$ with error parameter $p$, then for every $n\in\N$, every $y\circ x\in S_n$ with $|y|=|x|=n$, 
    \begin{align*}
        \rK^{p(n)}(x \mid y) &\le -\log \bar{\mathcal{D}}^*_n(x \mid y) + (\num^{n,2n}_{\bar{\mathcal{D}}_n,1/4q(n)n}(y\circ x)+1)\cdot O(\log nq(n))\\
        &\le -\log \mathcal{D}_n(x \mid y) + (\num_{\bar{\mathcal{D}}_n,1/4q(n)n}(y\circ x)+1)\cdot O(\log nq(n)).
    \end{align*}
    
    Suppose that there is no almost everywhere one-way function. By \Cref{lemma:ae-ext-to-nba}, there exists a subset $S\in \supp(\bar{\mathcal{D}})$ such that (i) $\bar{\mathcal{D}}$ is next-bits-predictable on $S$ with error parameter $p$ and (ii) for infinitely many $n\in\N$, $\Prob_{y\circ x\sim \bar{\mathcal{D}}_n}[y\circ x\in S_n]\ge 1-1/(2q(n))$. Below, we fix such an $n$ arbitrarily. 

    By \Cref{prop: large conditional}, $\Prob_{y\circ x\sim \bar{\mathcal{D}}_n}[\num_{1/4q(n)n}(y\circ x) = 0]\ge 1-1/(2q(n))$. Thus, by the union bound,
    \[\Prob_{(x,y)\sim\mathcal{D}_n}\mleft[\num_{\bar{\mathcal{D}}_n,1/4q(n)n}(y\circ x) = 0 \text{ and } y\circ x\in S_n\mright]\ge 1-\frac{1}{q(n)}.\]
    For any $(x,y)$ satisfying the event above, we obtain that 
    \begin{align*}
        \rK^{p(n)}(x \mid y)
        &\le -\log \mathcal{D}_n(x \mid y) + (\num_{\bar{\mathcal{D}}_n,1/4q(n)n}(y\circ x)+1)\cdot O(\log nq(n))\\
        &\le -\log \mathcal{D}_n(x \mid y) + O(\log nq(n)).
    \end{align*}
    By selecting a large enough polynomial $p'$, we conclude that for infinitely many $n\in\N$, 
    \[\Prob_{(x,y)\sim\mathcal{D}_n}\mleft[\rK^{p'(n)}(x \mid y)\le -\log \mathcal{D}_n(x \mid y) + \log p'(n)\mright]\ge 1-\frac{1}{q(n)},\]
    where the existence of the efficient encoder follows from that of \Cref{thm: nbp to coding}.
\end{proof}

Next, we show  (\Cref{i:OWF-ae-noOWF} $\implies$ \Cref{i:OWF-ae-Coding}) in \Cref{t:OWF-ae} in almost the same way. 
\begin{theorem}\label{optimal coding-ae}
	If there is no infinitely-often one-way function, then for every samplable distribution $\mathcal{D}=\{\mathcal{D}_n\}$ supported over $\bool^n\times\bool^n$, there exists a polynomial $p$ such that for all $n,k\in\N$, 
	\[\Prob_{(x,y)\sim\mathcal{D}_n}\mleft[\rK^{p(n,k)}(x\mid y)\le \log \frac{1}{\mathcal{D}_n(x\mid y)}  + \log p(n,k) \mright]\ge1-\frac{1}{k}.\]
\end{theorem}

\begin{proof}
Let $\mathcal{D}=\{\mathcal{D}_n\}$ be an arbitrary samplable distribution. We define another samplable distribution $\bar{\mathcal{D}}=\{\bar{\mathcal{D}}_{\lrang{n,k}}\}_{n,k\in\N}$ as  $\bar{\mathcal{D}}_{\lrang{n,k}}\equiv \mathcal{D}^{(2)}_n\circ \mathcal{D}^{(1)}_n$ for each $n,k\in\N$, where $\lrang{,}$ is the (standard) efficiently computable and efficiently invertible one-to-one pairing function satisfying that $\max\{n,k\}\le \lrang{n,k}\le \poly(n,k)$.

We apply \Cref{thm: nbp to coding} for $\bar{\mathcal{D}}$ and a polynomial $4n\cdot n$. We use the same argument as the proof of \Cref{optimal coding-io} except we use \Cref{lemma:io-ext-to-nba} instead of \Cref{lemma:ae-ext-to-nba}. Then, we can show that there exists a polynomial $p$ such that for \emph{all} $n,k\in\N$ (because of \Cref{lemma:io-ext-to-nba}),
\[\Prob_{(x,y)\sim\mathcal{D}_n}\mleft[\rK^{p(\lrang{n,k})}(x \mid y)\le -\log \mathcal{D}_n(x \mid y) + \log p(\lrang{n,k})\mright]\ge 1-\frac{1}{\lrang{n,k}}\ge 1-\frac{1}{k},\]
which implies \Cref{optimal coding-ae} because $\lrang{n,k}\le \poly(n,k)$. Again, the existence of the efficient encoder follows from that of \Cref{thm: nbp to coding}.
\end{proof}

\begin{proof}[Proof of \Cref{t:char-owf}]
(\Cref{i:char-owf:optimal} $\implies$ \Cref{i:char-owf:nontrivial}) trivially holds. (\Cref{i:char-owf:nontrivial} $\implies$ \Cref{i:char-owf:no-owf}) has been proved in \cite[][Proposition~3.2]{TrevisanVZ05_cc_journals} based on Levin's observation. Thus, it suffices to show (\Cref{i:char-owf:no-owf} $\implies$ \Cref{i:char-owf:optimal}).

\Cref{optimal coding-ae} (in the formulation using $(\Enc,\Dec)$ of \Cref{lemma: arithmetic encoding}) shows that, under the non-existence of infinitely-often one-way functions, for every samplable distribution $\mathcal{D}=\{\mathcal{D}_n\}$ (without loss of generality, we assume that each $\mathcal{D}_n$ is over $\bin^{\ell(n)}$ for an efficiently computable function $\ell(n)\le \poly(n)$ by padding), there exists a pair of polynomial-time computable functions $\Enc$ and $\Dec$ such that for every $n\in\N$,
\[\Prob_{x\sim\mathcal{D}_n,\Dec,\Enc}[\Dec(\Enc(x,1^n))=x \text{ and } |\Enc(x,1^n)|\le -\log\mathcal{D}(x)+O(\log n)]\ge 1-\frac{1}{4\ell(n)^2}.\]
By Markov's inequality, we have
\begin{align*}\label{eq:good-x}
    &\quad\Prob_{x}\mleft[\Prob_{\Enc,\Dec}[\Dec(\Enc(x,1^n))=x\text{ and } |\Enc(x,1^n)|\le -\log\mathcal{D}(x)+O(\log n)]\ge 1-\frac{1}{2\ell(n)}\mright]\\
    &\ge 1-\frac{1}{2\ell(n)}.
    \numberthis
\end{align*}
We consider a modified efficient encoder $\Enc'$ that, for a given $x\sim\mathcal{D}_n$ and $1^n$, performs the empirical estimation of the probability that $\Dec(\Enc(x,1^n))=x$ and $|\Enc(x,1^n)|\le -\log\mathcal{D}(x)+O(\log n)$ hold with additive accuracy error $1/(8\ell(n))$ and with failure probability at most $2^{-n}$. If the estimated probability $\tilde{p}$ is at least $1-\frac{3}{4\ell(n)}$, the encoder $\Enc'$ sends $0\Enc(x,1^n)$ (executed with fresh random seeds); otherwise, $\Enc'$ sends $1x$. We also define a modified efficient decoder $\Dec'$ that outputs for given encoding $e'$, if $e'$ takes the form of $0e$, it outputs $\Dec(e)$; otherwise if $e'$ takes the form of $1x$, it outputs $x$.

We verify the \emph{worst-case} correctness of $(\Enc',\Dec')$. For every $n\in\N$ and every $x\in\supp(\mathcal{D}_n)$, if $x$ passes the empirical test in $\Enc'$ under the condition that the empirical estimation is performed successfully, the probability that $\Dec(\Enc(x,1^n))=x$ holds is at least $1-3/(4\ell(n))-1/(8\ell(n)) = 1-7/(8\ell(n))$. Thus, the failure probability that $\Dec'(\Enc'(x,1^m))\neq x$ given this condition is at most $7/(8\ell(n))$. By contrast, if $x$ does not pass the test, $\Dec'(\Enc'(x,1^m))= x$ holds with probability $1$ given this event. Thus, for every $x\in\supp(\mathcal{D}_n)$, it holds that
\[\Prob_{\Enc',\Dec'}[\Dec'(\Enc'(x,1^n))=x]\ge 1-\frac{7}{8\ell(n)}-2^{-n}\ge \frac{2}{3}.\]

We also evaluate the expected length of the encoding. For every $x\in\supp(\mathcal{D}_n)$, if $x$ satisfies the event in \Cref{eq:good-x}, the probability that $x$ passes the empirical test is at least $1-2^{-n}$ (in this case, $\Enc'$ outputs $\Enc(x,1^n)$). Thus, the expected length of the encoding under this condition that $x$ satisfies the event in \Cref{eq:good-x} is at most 
\[2^{-n}\cdot(\ell(n)+1) + 1\cdot |\Enc(x,1^n)|\le -\log \mathcal{D}(x)+ O(\log n).\]
By contrast, if $x$ does not satisfy the event in \Cref{eq:good-x}, the length of the outcome of $\Enc'(x,1^n)$ is always at most $\ell(n)+1$.
Thus, we conclude that
\[\Exp_{x\sim \mathcal{D}_n}[|\Enc'(x,1^n)|]\le \Exp_{x\sim \mathcal{D}_n}[-\log \mathcal{D}_n(x)] + O(\log n) + \frac{\ell(n)+1}{2\ell(n)}= \Shannon(\mathcal{D}_n)+O(\log n).\]
This completes the proof of \Cref{t:char-owf}.
\end{proof}

\subsubsection{Coding Theorems for Next-Bits-Predictable Source}\label{section: coding for rand nbps}
We show that the next-bits prediction on the whole support yields the \emph{optimal} (average-case) coding theorem.
\begin{theorem}\label{t:opt-avg-coding-approx}
    If a distribution family $\mathcal{D}=\{\mathcal{D}_n\}$, where each $\mathcal{D}_n$ is over $\bin^n$, is next-bits-predictable on $\supp(\mathcal{D})$ with arbitrary polynomial error parameter, then for every polynomial $q$, there exists a polynomial $p$ such that for every $n\in\N$, 
    \[\Prob_{x\sim\mathcal{D}_n}\mleft[\rK^{p(n)}(x)\le -\log\mathcal{D}_n(x) + \log p(n)\mright]\ge 1-\frac{1}{q(n)}.\]

    In particular, there exists a pair $(\Enc, \Dec)$ of randomized polynomial-time algorithms such that for every $n \in \Nat$,
\[
 \Exp_{x \sim \mathcal{D}_n, \Enc} [ |\Enc(x,1^n)| ] \le \Shannon(\mathcal{D}_n) + \log p(n)
\]
and for every $x \in \supp(\mathcal{D}_n)$, \[\Prob_{\Enc, \Dec} \mleft[ \Dec(\Enc(x,1^n)) = x \mright] \ge \frac{2}{3}.
\]
\end{theorem}

\begin{proof}
    We apply \Cref{thm: nbp to coding} for $\mathcal{D}$ and a polynomial $nq(n)$. Then, there exists a polynomial $p$ such that for every $n\in\N$ and every $x\in\supp(\mathcal{D}_n)$,
    \[\rK^{p(n)}(x)\le -\log \mathcal{D}_n(x) + \num_{1/(nq(n))}(x)\cdot O(\log n) + \log p(n).\]
    By \Cref{prop: large conditional}, it holds that for every $n\in\N$,
    \[\Prob_{x\sim\mathcal{D}_n}\mleft[\num_{1/(nq(n))}(x)>0\mright]\le \frac{n}{nq(n)}= \frac{1}{q(n)}.\]
    From the two expressions above, we obtain 
    \[\Prob_{x\sim\mathcal{D}_n}\mleft[\rK^{p(n)}(x)\le -\log\mathcal{D}_n(x) + \log p(n)\mright]\ge 1-\frac{1}{q(n)}.\]
    In particular, \Cref{lemma: arithmetic encoding} shows that there exists a pair $(\Enc, \Dec)$ of randomized polynomial-time algorithms such that for every $n \in \Nat$,
    \[\Prob_{x\sim\mathcal{D}_n,\Enc}\mleft[|\Enc(x,1^n)|\le -\log\mathcal{D}_n(x) + \log p(n)\mright]\ge 1-\frac{1}{2n}.\]
    and for every $x \in \supp(\mathcal{D}_n)$, \[\Prob_{\Enc, \Dec} \mleft[ \Dec(\Enc(x,1^n)) = x \mright] \ge \frac{2}{3}.\]
    Without loss of generality, we assume that $\Enc$ always outputs the encoding of length at most $2n$ for given $x\in\bin^n$ and $1^n$; otherwise, we can replace the encoding to the canonical one into which embedded $x$. Thus, the expected length of the encoding is bounded as follows:
    \[\Exp_{x\sim\mathcal{D}_n,\Enc}[|\Enc(x,1^n)|]\le \Exp_{x\sim \mathcal{D}_n}[-\log \mathcal{D}_n(x)] + O(\log n) + \frac{2n}{2n} = \Shannon(\mathcal{D}_n) + O(\log n),\]
    as desired.
\end{proof}

Furthermore, the same class of next-bits-predictable distributions admits the almost optimal worst-case coding theorem in the following sense. 
\begin{theorem}\label{t:almost-opt-wst-coding}
    If a distribution family $\mathcal{D}=\{\mathcal{D}_n\}$, where each $\mathcal{D}_n$ is over $\bin^n$, is next-bits-predictable on $\supp(\mathcal{D})$ with arbitrary polynomial error parameter, then for every $\epsilon>0$, there exists a polynomial $p$ such that for every $n\in\N$ and every $x\in \supp(\mathcal{D}_n)$, 
    \[\rK^{p(n)}(x)\le -(1+\epsilon)\log\mathcal{D}_n(x) + \log p(n).\]
\end{theorem}
Note that the term ``almost'' optimal is due to the arbitrarily small constant $\epsilon>0$ above.

\begin{proof}
    Let $k\in\N$ be a sufficiently large constant (with respect to $\epsilon^{-1}$) specified later.  We apply \Cref{thm: nbp to coding} for $\mathcal{D}$ and a polynomial $q(n)=n^{k}$. Then, there exists a polynomial $p$ such that for every $n\in\N$ and every $x\in\supp(\mathcal{D}_n)$,
    \[\rK^{p(n)}(x)\le -\log \mathcal{D}_n(x) + C\cdot \num_{n^{-k}}(x)\log n + O(\log n),\]
    where $C>0$ is a universal constant independent of $k$.

    For every $x\in\supp(\mathcal{D}_n)$, we can observe that 
    \[\mathcal{D}_n(x) =\prod_{i=1}^n \mathcal{D}^*(x_i\mid x_{[i-1]}) \le (n^{-k})^{\num_{n^{-k}}(x)} = n^{-k\num_{n^{-k}}(x)}.\]
    By rearranging the above,
    \[\num_{n^{-k}}(x)\log n\le -\frac{1}{k}\log\mathcal{D}_n(x).\]

    Therefore, by selecting $k$ to be sufficiently large so that $C/k \le \epsilon$, we have
    \begin{align*}
        \rK^{p(n)}(x) &\le -\log \mathcal{D}_n(x) + C\cdot \num_{n^{-k}}(x)\log n + O(\log n)\\
        &\le -\mleft(1+\frac{C}{k}\mright)\log \mathcal{D}_n(x) + O(\log n)\\
        &\le -(1+\epsilon)\log \mathcal{D}_n(x) + O(\log n),
    \end{align*}
    as desired.
\end{proof}

\subsubsection{Towards the Uniform Version of the Haitner--Mazor--Silbak Theorem}\label{sec:uniform-HMS}
Recently, \citet{HaitnerMS23_innovations_conf} presented the clear relationship between incompressibility and next-bit pseudoentropy in the \emph{nonuniform} computational model. They further mentioned that the result holds in the uniform computational model when \emph{shared randomness is available} between the encoder and decoder \cite[see][Remark~6]{HaitnerMS23_innovations_conf}. Note that the shared randomness is used for executing a distinguisher, and polynomially many shared random bits are required in general.  We extend the relationship to the \emph{uniform} computational model without the usage of the shared randomness only at the expense of \emph{a small multiplicative loss}. 

First, we review the main result of \cite{HaitnerMS23_innovations_conf}. For this, we recall the notions of incompressibility and next-bit pseudoentropy.

\begin{definition}[Imcompressibility]
    For $k\colon\N\to\N$, a distribution family $\mathcal{D}=\{\mathcal{D}_n\}$ is said to be nonuniformly $k$-incompressible if for every nonuniform polynoimal-time algorithms $\Enc$ and $\Dec$ such that $\Dec(\Enc(x,1^n))=x$ for all $x\in\supp(\mathcal{D}_n)$, it holds that for every large enough $n\in\N$, 
    \[\Exp_{x\sim \mathcal{D}_n}[|\Enc(x,1^n)|]\ge k(n).\]
\end{definition}

\begin{definition}[Pseudoentropy]
    Let $(X,B)=\{(X_n,B_n)\}_n$ be a family of joint distributions over strings. We say that $B$ has nonuniform-conditional-pseudoentropy (resp.\ uniform-conditional-pseudoentropy) $k\colon\N\to\N$ given $X$ if for every polynomial $p$, there exists a distribution family $C=\{C_n\}$ that jointly distributed with $\{X_n\}$ as satisfies the following:
    \begin{itemize}
        \item $\Shannon(C_n\mid x_n)\ge k(n)-1/p(n)$ for each $n\in\N$;
        \item $(X,B)$ and $(X,C)$ are computationally indistinguishable by nonuniform (resp.\ uniform) randomized polynomial-time algorithms.
    \end{itemize}
\end{definition}

\begin{definition}[Next-bit pseudoentropy]
    A distribution family $\mathcal{D}=\{\mathcal{D}_n\}$, where each $\mathcal{D}_n$ is over $\bin^{\ell(n)}$, is said to have nonuniform-next-bit-pseudoentropy (resp.\ uniform-next-bit-pseudoentropy) $k\colon\N\to\N$ if a distribution family $\{(\mathcal{D}_n)_{I_n}\}_n$, where each $I_n$ is the uniform distribution over $[\ell(n)]$ and $(\mathcal{D}_n)_{I_n}$ represents the $I_n$-th bit of $\mathcal{D}_n$, has nonuniform-conditional-pseudoentropy (resp.\ uniform-conditional-pseudoentropy) $k(n)/\ell(n)$ given $\{(\mathcal{D}_n)_{[I_n-1]}\}_n$.
\end{definition}

One of the main theorems of \cite{HaitnerMS23_innovations_conf} is stated as follows:
\begin{theorem}[{\cite[][Lemma~1]{HaitnerMS23_innovations_conf}}]
For a distribution family $\mathcal{D}=\{\mathcal{D}_n\}$, if $\mathcal{D}$ is nonuniformly $k(n)$-incompressible, then $\mathcal{D}$ has nonuniform-next-bit-pseudoentropy $k(n) - 2$.
\end{theorem}

In this section, we show the uniform variant with a multiplicative loss in the parameter $k(n)$.
First, we present the notion of randomized compression in the uniform computational model, which was formally studied in~\cite{TrevisanVZ05_cc_journals}.

\begin{definition}[Randomized compression]
    We say that a pair $(\Enc,\Dec)$ of randomized algorithms compresses a distribution family $\mathcal{D}=\{\mathcal{D}_n\}_{n\in\N}$ to length $k\colon\N\to\N$ with decoding error $\delta\colon\N\to[0,1]$ if it satisfies that for infinitely many $n\in\N$,%
    \footnote{In this work, we consider a randomized compression only on infinitely many $n\in\N$ to discuss the uniform version of \emph{almost everywhere} imcompressibility}
    \begin{itemize}
        \item for any $x\in\supp(\mathcal{D}_n)$, $\Prob_{\Enc,\Dec}[\Dec(\Enc(x,1^n))=x]\ge 1-\delta(n)$;
        \item $\Exp_{x\sim\mathcal{D}_n,\Enc}[|\Enc(x)|]\le k(n)$.
    \end{itemize}
    We say that $\mathcal{D}$ is randomly compressible to length $m$ with decoding error $\delta$ if there exists a pair of randomized polynomial-time algorithms that compresses $\mathcal{D}$ to length $m$ with decoding error $\delta$. Moreover, we say that $\mathcal{D}$ is randomly incompressible to length $m$ with decoding error $\delta$ if not randomly compressible to length $m$ with decoding error $\delta$.
\end{definition}

It is known that the decoding error can be exponentially reduced.
\begin{lemma}[{\cite[][Lemma 2.11]{TrevisanVZ05_cc_journals}}]\label{l:reducing-error}
    For a distribution family $\mathcal{D}=\{\mathcal{D}_n\}$, where each $\mathcal{D}_n$ is over $\bin^{\ell(n)}$, if $\mathcal{D}$ is randomly compressible to length $m$ with decoding error $\delta$, then it is also randomly compressible to the length $m+3\delta\cdot\ell(n)+2$ with decoding error $2^{-n}$.
\end{lemma}

Now, we state the main result in this section. 
\begin{theorem}\label{t:HMS-uniform}
For a distribution family $\mathcal{D}=\{\mathcal{D}_n\}$ and every constant $\epsilon>0$, if $\mathcal{D}$ is randomly $k(n)$-incompressible with decoding error $2^{-n}$, then $\mathcal{D}$ has uniform-next-bit-pseudoentropy $(1-\epsilon)k(n) - O(\log n)$.
\end{theorem}
A large part of the proof follows from that of \cite{HaitnerMS23_innovations_conf}, so we strongly recommend the reader to refer to the prior work first. Below, we extract the relevant points to our work. 

We introduce some notions following from \cite{VadhanZ12_stoc_conf, HaitnerMS23_innovations_conf}. For a function $p\colon \bin^*\times \bin\to\mathbb{R}_{>0}$, we define a conditional probability $C_p(\cdot|\cdot)$ as for each $b\in\bin$ and $x\in\bin^*$,
\[C_p(b\mid x) = \frac{p(x,b)}{p(x,0)+p(x,1)}.\]
For a randomized algorithm $P$ that maps $(x,b)\in \bin^*\times\bin$ to a real positive value, we extend the notion above as
\[C_P(b\mid x) = \Exp_{r}\mleft[\frac{P(x,b;r)}{P(x,0;r)+P(x,1;r)}\mright].\]
Notice that
\[C_P(0\mid x) + C_P(1\mid x) = \Exp_{r}\mleft[\frac{P(x,0;r)+P(x,1;r)}{P(x,0;r)+P(x,1;r)}\mright]= 1.\]
Furthermore, for each $m\in\N$, we define a distribution $\mathcal{D}_m^P$ over $\bin^m$ as for each $x\in\bin^m$,
\[\mathcal{D}_m^P(x) = \prod_{i=1}^m C_P(x_i\mid x_{[i-1]}).\]

The prior work~\cite{HaitnerMS23_innovations_conf} showed the following technical lemma.
\begin{lemma}[{\cite[][Section~3.3]{HaitnerMS23_innovations_conf}} building upon \cite{VadhanZ12_stoc_conf}]\label{l:HMS-tech}
    If a distribution family $\mathcal{D}=\{\mathcal{D}_n\}$, where $\mathcal{D}_n$ is over $\bin^{\ell(n)}$, does not have uniform-next-bit pseudoentropy $k(n)$, then there exists a randomized polynomial-time algorithm $P$ that maps $(x,b)\in \bin^*\times\bin$ to a real positive value so that for infinitely many $n\in\N$,
    \[\KL(\mathcal{D}_n\| \mathcal{D}^P_{\ell(n)})\le k(n)-\Shannon(\mathcal{D}_n).\]
\end{lemma}

We also use the following well-known fact.
\begin{lemma}[{\cite[cf.][Theorem~5.4.3]{0016881_daglib_books}}]\label{l:wrong-code}
For every distributions $\mathcal{D}$ and $\mathcal{E}$ with $\KL(\mathcal{D}\|\mathcal{E})<\infty$,
\[\Exp_{x\sim \mathcal{D}}[-\log \mathcal{E}(x)]= \Shannon(\mathcal{D})+\KL(\mathcal{D}\|\mathcal{E}).\]
\end{lemma}

Now, we prove \Cref{t:HMS-uniform} based on \Cref{t:almost-opt-wst-coding} and \Cref{l:HMS-tech,l:wrong-code}.

\begin{proof}[Proof of \Cref{t:HMS-uniform}]
    Let $k(n)$ be an arbitrary polynomial. Let $\mathcal{D}=\{\mathcal{D}_n\}$ be a distribution family, where each $\mathcal{D}_n$ is over $\bin^{\ell(n)}$. By \Cref{l:HMS-tech}, if $\mathcal{D}$ does not have uniform-next-bit pseudoentropy $k(n)$, then there exists a randomized polynomial-time algorithm $P$ that maps $(x,b)\in \bin^*\times\bin$ to a real positive value so that for infinitely many $n\in\N$,
    \begin{equation}\label{eq:KL-bound}
        \KL(\mathcal{D}_n\| \mathcal{D}^P_{\ell(n)})\le k(n)-\Shannon(\mathcal{D}_n).
    \end{equation}

    We observe that for each $n\in\N$, $x\in \supp(\mathcal{D}_{\ell(n))}^P)$, $b\in\bin$, and $i\in[\ell(n)]$, the conditional probability $\mathcal{D}_{\ell(n))}^{P*}(b\mid x_{[i-1]})= C_P(b\mid x_{[i-1]})$ is predictable with additive error $1/p(n)$, where $p$ is an arbitrarily large polynomial, by the empirical estimation of the quantity 
    \[\Exp_{r}\mleft[\frac{P(x,b;r)}{P(x,0;r)+P(x,1;r)}\mright].\]
    Note that the approximation halts in polynomial time in $n$ and $p(n)$.

    Therefore, by \Cref{t:almost-opt-wst-coding} (in the form of \Cref{lemma: arithmetic encoding}), there exists a pair $(\Enc,\Dec)$ of randomized polynomial-time algorithms such that for every $n\in\N$ and every $x\in \supp(\mathcal{D}_{\ell(n)}^P)$, it holds that 
    \begin{equation}\label{eq:decoding-error}
        \Prob_{\Enc,\Dec}[\Dec(\Enc(x,1^n))=x]\ge 1-\frac{1}{3\ell(n)}
    \end{equation}
    and
    \[|\Enc(x,1^n)|\le -(1+\epsilon)\log \mathcal{D}_{\ell(n)}^P(x) + O(\log n).\]

     Thus, we have that for every $n$ satisfying \Cref{eq:KL-bound},
    \begin{align}
        \Exp_{\Enc,x\sim \mathcal{D}_n}\mleft[|\Enc(x,1^n))|\mright] &\le  (1+\epsilon)\Exp_{x\sim \mathcal{D}_n}[-\log \mathcal{D}_{\ell(n)}^P(x)] + O(\log n)\notag\\
        &\le (1+\epsilon)\mleft(\Shannon(\mathcal{D}_n) + \KL(\mathcal{D}_n\| \mathcal{D}_{\ell(n)}^P)\mright) + O(\log n)\notag\\
        &\le (1+\epsilon)\cdot k(n) + O(\log n),\label{eq:exp-length-bound}
    \end{align}
    where the second inequality follows from \Cref{l:wrong-code}, and the last inequality follows from \Cref{eq:KL-bound}.

    Thus, from \Cref{eq:decoding-error,eq:exp-length-bound}, $\mathcal{D}$ is randomly compressible to length $(1+\epsilon)\cdot k(n) + O(\log n)$ with decoding error $1/(3\ell(n))$. By \Cref{l:reducing-error}, $\mathcal{D}$ is also randomly compressible with decoding error $2^{-n}$ to the length \[(1+\epsilon)\cdot k(n) + O(\log n)+\frac{3\ell(n)}{3\ell(n)}+2 = (1+\epsilon)\cdot k(n) + O(\log n).\]
    By retaking $k(n)$ to be $(1/1+\epsilon)(k(n)-O(\log n)) = (1-\epsilon/(1+\epsilon))(k(n)-O(\log n))$, we obtain the theorem.
\end{proof}
\section{One-Way Functions, Conditional Coding and Symmetry of Information for \texorpdfstring{$\rK^\poly$}{PDFstring}}

In this section, we prove \Cref{t:OWF-io} and \Cref{t:OWF-ae}.
We first show \Cref{t:OWF-io}, which is restated below.

\OWFIO*
\begin{proof}
	(\Cref{i:OWF-io-noOWF} $\implies$ \Cref{i:OWF-io-Coding}) follows directly from \Cref{optimal coding-io}.
	
	We then show the following implications in subsequent sections.
	\begin{itemize}
		\item \Cref{i:OWF-io-Coding} $\implies$ \Cref{i:OWF-io-SoI} (\Cref{l:avg-coding-soi-io} in \Cref{sec:avg-coding-soi-io}).
		\item \Cref{i:OWF-io-SoI} $\implies$ \Cref{i:OWF-io-noOWF} (\Cref{l:avg-soi-owf-io} in \Cref{sec:avg-soi-owf-io}).
	\end{itemize}
	This will complete the proof of \Cref{t:OWF-io}.
\end{proof}

\subsection{Average-Case Symmetry of Information from Conditional Coding}\label{sec:avg-coding-soi-io}

\begin{lemma}[\Cref{i:OWF-io-Coding} $\implies$ \Cref{i:OWF-io-SoI} in \Cref{t:OWF-io}]\label{l:avg-coding-soi-io}
	If average-case conditional coding holds for $\rK^t$, then average-case symmetry of information also holds.
\end{lemma} 

We first need the following lemma.

\begin{lemma}\label{optimal coding-io K-version}
	If one-way functions do not exist, then for every samplable distribution family $\mathcal{D}=\{\mathcal{D}_n\}$ supported over $\bool^n\times\bool^n$ and every polynomial $q$, there exists a polynomial $p$ such that for infinitely many $n\in\N$, 
	\[\Prob_{(x,y)\sim\mathcal{D}_n}\mleft[\rK^{p(n)}(x\mid y)\le \K(x\mid y)  + \log p(n) \mright]\ge1-\frac{1}{q(n)}.\]
\end{lemma}
\begin{proof}
	Let $\mathcal{D}=\{\mathcal{D}_n\}$ be a polynomial-time samplable distribution family and $q$ be a polynomial.
	
	Assuming one-way functions do not exist, by \Cref{optimal coding-io}, we have that there exists some $p'$ such that for infinitely many $n\in\N$, 
	\[
		\Prob_{(x,y)\sim\mathcal{D}_n}\mleft[\rK^{p'(n)}(x\mid y)\le \log \frac{1}{\D_n(x\mid y)}  + \log p'(n) \mright]\ge1-\frac{1}{2q(n)}.
	\]
	
	Also, by \Cref{l:incompressible}, we get that for every $n$, with probability at least $1-(2q(n))$ over $(x,y)\sim\D_n$, we have
	\[
	\mathsf{K}(x\mid y) \geq  \log\frac{1}{\mathcal{D}_n(x\mid y)}-O(\log q(n)).
	\]
	By a union bound, we get that for infinitely many $n$, with probability at least $1-q(n)$ over $(x,y)\sim\D_n$,
	\[
		\rK^{p'(n)}(x\mid y)\le \K(x\mid y)  + \log p'(n) +O(\log q(n)).
	\]
	The lemma follows by letting $p$ be a sufficient large polynomial.
\end{proof}

We are now ready to show \Cref{l:avg-coding-soi-io}.

\begin{proof}[Proof of \Cref{l:avg-coding-soi-io}]
	Let $\mathcal{D}=\{\mathcal{D}_n\}$ be a polynomial-time samplable distribution family and $q$ be a polynomial.
	
	We show the following claim.
	
	\begin{claim}\label{c:good-n}
		There exists a polynomial $p'$ such that for infinitely many $n$, both the following hold with probability at least $1-1/q(n)$ over $(x,y)\sim\D_n$.
		\begin{enumerate}
			\item\label{i:good-n-1} $\rK^{p'(n)} (x\mid y)\leq \log \K(x\mid y)+\log p'(n)$.
			\item\label{i:good-n-2} $\rK^{p'(n)} (y)\leq \K(y) +\log p'(n)+O(\log n)$.
		\end{enumerate}
	\end{claim}
	\begin{proof}[Proof of \Cref{c:good-n}]\renewcommand\qedsymbol{$\diamond$}
 	Let $\D'\vcentcolon=\{\mathcal{D}'_n\}$ be the polynomial-time samplable distribution family where each $\mathcal{D}'_n$ is sampled by first sampling $y\sim \D_n^{(2)}$ and outputs $(y,0^n)$, where $\D_n^{(2)}$ is the marginal distribution of $\D_n$ on the second half.
		
		Finally, let $\mathcal{E}$ be the uniform mixture of  $\D$ and $\D'$.
		
		Suppose one-way functions do not exist. Then by \Cref{optimal coding-io K-version}, there exists a polynomial $p'$ such that for infinitely many $n\in\N$, 
		\[\Prob_{(x,y)\sim\mathcal{E}_n}\mleft[\rK^{p'(n)}(x\mid y)> \K(x\mid y)  + \log p'(n) \mright]\leq \frac{1}{4q(n)}.\]		
		Since $\mathcal{E}_n$ samples $\D_n$ with probability $1/2$,
		\begin{equation}\label{eq:D-fails}
			\Prob_{(x,y)\sim\mathcal{D}_n}\mleft[\rK^{p'(n)}(x\mid y)> \K(x\mid y)  + \log p'(n) \mright]\leq \frac{1}{2q(n)}.
		\end{equation}
		
		Similarly, we get
		\begin{equation}\label{eq:d_2}
			\Prob_{(a,b)\sim\mathcal{D}'_n}\mleft[\rK^{p'(n)}(a\mid b)> \K(a\mid b)  + \log p'(n) \mright]\leq \frac{1}{2q(n)}.		
		\end{equation}
		Note that the \Cref{eq:d_2} essentially means
		\[
		\Prob_{y\sim\mathcal{D}^{(2)}_n}\mleft[\rK^{p'(n)}(y\mid 0^n)> \K(y\mid 0^n)  + \log p'(n) \mright]\leq \frac{1}{2q(n)}.		
		\]
		Finally, note that $\rK^{p'(n)}(y)\leq \rK^{p'(n)}(y\mid 0^n)+O(\log n)$ and $\K(y)\geq \K(y\mid 0^n)$. Therefore, the above implies
		\begin{equation}\label{eq:D-2-fails}
		\Prob_{y\sim\mathcal{D}^{(2)}_n}\mleft[\rK^{p'(n)}(y)> \K(y)  + \log p'(n) +O(\log n) \mright]\leq \frac{1}{2q(n)}.		
		\end{equation}
		The claim follows by taking a union bound over Equations (\ref{eq:D-fails}) and (\ref{eq:D-2-fails}).
	\end{proof}
	
	Fix any $n$ and $(x,y)$ such that both the conditions stated in \Cref{c:good-n} hold. Then we have
	\begin{align*}
		\mathsf{rK}^{p'(n)}(x\mid y) &\leq \K(x\mid y) +  \log p'(n,k) \tag{by \Cref{i:good-n-1} of \Cref{c:good-n}}\\
		&\leq \K(x,y) -\K(y) +O(\log n) +  \log p'(n) \tag{by Symmetry of Information for $\K$}\\
		&\leq \K^{t}(x,y) - \rK^{p'(n)}(y) + O(\log n) + 2\log p'(n). \tag{by \Cref{i:good-n-2} of \Cref{c:good-n}}
	\end{align*}
	By letting $p$ be a sufficiently large polynomial, we get that for every $t\geq p(n)$,
	\[
	\rK^{t}(x \mid y) \leq \rK^t(x,y) - \rK^{t}(y) +  \log t,
	\]
	as desired.
\end{proof}

\subsection{Inverting One-Way Functions from Average-Case Symmetry of Information}\label{sec:avg-soi-owf-io}
\begin{lemma}[\Cref{i:OWF-io-SoI} $\implies$ \Cref{i:OWF-io-noOWF} in \Cref{t:OWF-io}]\label{l:avg-soi-owf-io} 	If average-case symmetry of information holds for $\rK^t$, then one-way functions do not exist.
\end{lemma}
\begin{proof}
	The proof uses ideas from \cite{LongpreW95_iandc_journals}.
	
	Let $f\colon\bool^n\to\bool^n$ be any function that is supposed to be infinitely-often secure. Let $q$ be any polynomial, we show that, for infinitely many $n$, we can invert $f$ with high probability $1-1/q(n)$ over $x\sim\bool^n$ in time $\poly(n)$.
	
	We first observe the following.	
	\begin{claim}[{\cite[Lemma 3.5]{LongpreW95_iandc_journals}}]\label{c:LW}
		For every $n$ and $x\in\bool^n$, we have
		\[
		\K(f(x))\geq \K(x) -\log |f^{-1}(f(x))| - O(\log n).
		\]
	\end{claim}
	\begin{proof}[Proof of \Cref{c:LW}]\renewcommand\qedsymbol{$\diamond$}
		First of all, for every $x\in\bool^n$, we have
		\begin{equation}\label{e:eq-1}
			\K(x\mid f(x)) \leq \log |f^{-1}(f(x))| + O(\log n).
		\end{equation}
		To see this, note that given $f(x)$, we can specify $x$ using the index of $x$ in the set $f^{-1}(f(x))$, which takes $\leq \log |f^{-1}(f(x))|$ bits. 
		Then we have
		\begin{align*}
			\K(x) &\leq \K(x\mid f(x)) +\K(f(x))\\
			&\leq \log |f^{-1}(f(x))| +\K(f(x))+ O(\log n) \tag{by \Cref{e:eq-1}},
		\end{align*}
		as desired.
	\end{proof}
	
	Next, assuming that average-case symmetry of information hods for $\rK^\poly$ (\Cref{i:OWF-io-Coding} in \Cref{t:OWF-io}), we show the following. 
	\begin{claim}\label{c:upper}
		There is a polynomial $p$ such that for infinitely many $n$, with probability at least $1-1/q(n)^2$ over $x\sim\bool^n$, we have
		\[
		\rK^{p(n)}(x\mid f(x))\leq \log |f^{ -1}(f(x))|+ \log p(n).
		\]
	\end{claim}
	\begin{proof}[Proof of \Cref{c:upper}]\renewcommand\qedsymbol{$\diamond$}
		Consider the polynomial-time samplable distribution family $\{\D_n\}$ where each $\D_n$ samples $x\sim\bool^n$ and outputs $(x,f(x))$.
		
		By the assumption that average-case symmetry of information hods for $\rK^\poly$ (\Cref{i:OWF-io-Coding} in \Cref{t:OWF-io}), there exists a polynomial $p'$ such that for infinitely many $n$, with probability at least $1-1/(2q(n)^2)$ over $x\sim\bool^n$, we have
		\begin{align*}\label{e:pKt-SoI}
			\mathsf{rK}^{p'(n)}(x\mid f(x))
			&\leq \mathsf{rK}^{p'(n)}(x,f(x))- \mathsf{rK}^{p'(n)}(f(x)) + \log p'(n)\\
			&\leq \mathsf{rK}^{p'(n)/2}(x) - \mathsf{rK}^{p'(n)}(f(x)) + \log p'(n) + O(\log n)\\
			&\leq \mathsf{rK}^{p'(n)/2}(x) - \mathsf{K}(f(x)) + \log p'(n) + O(\log n)\\
			&\leq \mathsf{rK}^{p'(n)/2}(x) - \mleft(\K(x) -\log |f^{-1}(f(x))| - O(\log n)\mright) + \log p'(n) + O(\log n)\\
			&\leq \mathsf{rK}^{p'(n)/2}(x)- \K(x) +\log |f^{-1}(f(x))| + \log p'(n) + O(\log n),\numberthis   			
		\end{align*}
		where the second inequality uses the fact that given $x$ we can compute $f(x)$ efficiently, and the second last inequality follows from \Cref{c:LW}.
		
		Note that by a counting argument, for every $n$, with probability at least $1-1/(2q(n)^2)$ over $x\sim\bool^n$, we have
		\[
		\K(x) \geq n- O(\log q(n)),
		\]
		which yields
		\begin{equation}\label{eq:rK-K}
			\mathsf{rK}^{p'(n)/2}(x)- \K(x)\leq O(\log q(n)).
		\end{equation}
		By Plugging \Cref{eq:rK-K} into \Cref{e:pKt-SoI} and by a union bound, we get that, for infinitely many $n$, with probability at least $1-1/q(n)^2$ over $x\sim\bool^n$
		\[
		\mathsf{rK}^{p'(n)}(x\mid f(x)) \leq \log |f^{-1}(f(x))| +  \log p'(n) + O(\log q(n)).
		\]
		The claim follows by letting $p$ be a sufficiently large polynomial.
	\end{proof}
	
	In what follows, we fix $n$ so that the condition stated in \Cref{c:upper} holds.
	
	Next, we observe the following equivalent way of sampling $(x,f(x))$ while $x$ is uniformly at random: We first sample $y:=f(z)$ for a uniformly random $z$ and then sample $x\sim f^{-1}(y)$. By an averaging argument, \Cref{c:upper} yields that with probability at least $1-1/q(n)$ over $y$ sampled this way, for at least $1-1/q(n)$ fraction of the $x\in f^{-1}(y)$, we have
	\begin{equation}\label{eq:good-small_pK}
		\rK^{p(n)}(x\mid y)\leq \log |f^{ -1}(y)|+ \log p(n).          
	\end{equation}
	
	Consider any \emph{good} $y$ such that \Cref{eq:good-small_pK} holds for at least $1-1/q(n)$ fraction of the $x\in f^{-1}(y)$. Also, let $S_y$ be the set of $x\in f^{-1}(y)$ such that \Cref{eq:good-small_pK} holds. Note that
	\[
	|S_y|\geq (1-1/q(n))\cdot |f^{-1}(y)|.
	\]
	
	Consider the following procedure $\mathsf{A}$ that takes $n$ and $y$ as input and does the following.
	\begin{enumerate}
		\item Pick $s\sim[2n]$,
		\item Pick $\Pi \sim \{0,1\}^s$,
		\item View $\Pi$ as a \emph{randomized} program, run $U(\Pi,y)$ for $p(n)$ steps, and return its output.
	\end{enumerate}
	
	It is easy to see from the definition of $\rK^t$ that for \emph{every} $x\in S_y$ (which satisfies \Cref{eq:good-small_pK}), the above procedure $\mathsf{A}$ outputs $x$ with probability at least
	\[
	\frac{1}{O(n)}\cdot \frac{1}{2^{\log |f^{ -1}(y)|+ \log p(n)}} \cdot \frac{2}{3}\geq \frac{1}{|f^{ -1}(y)|}\cdot \frac{1}{p(n)^2}.
	\]
	Since the above holds for every $x\in S_y$, we get that the probability that $\mathsf{A}(1^n,y)$ outputs \emph{some} $x\in S_y$ is at least
	\[
	|S_y|\cdot \frac{1}{|f^{ -1}(y)|}\cdot \frac{1}{p(n)^2}\geq \frac{1}{\poly(n)}.
	\]
	In other words, with probability at least $1-1/k$ over $x\sim\{0,1\}^n$ (in which case $f(x)$ is good), $\mathsf{A}(1^n,f(x))$ outputs some pre-image of $f(x)$ with probability at least $1/\poly(n)$. This breaks the one-way-ness of $f$.
\end{proof}

\subsection{Characterizing Infinitely-Often One-Way Functions}
In this subsection, we show \Cref{t:OWF-ae}, which is restated below.

\OWFAE*
\begin{proof}
		(\Cref{i:OWF-ae-noOWF} $\implies$ \Cref{i:OWF-ae-Coding}) follows directly from \Cref{optimal coding-ae}.
		The proof of (\Cref{i:OWF-ae-Coding} $\implies$ \Cref{i:OWF-ae-SoI}) can be easily adapted from that of \Cref{l:avg-coding-soi-io}.
        Also, the proof of (\Cref{i:OWF-ae-SoI} $\implies$ \Cref{i:OWF-ae-noOWF}) can be easily adapted from that of \Cref{l:avg-soi-owf-io}.
		This shows the equivalence of \Cref{i:OWF-ae-noOWF}, \Cref{i:OWF-ae-Coding}, and \Cref{i:OWF-ae-SoI}.
		
		We then show the following implications in the rest of this subsection.
		 \begin{itemize}
		 	\item \Cref{i:OWF-ae-Coding} $\iff$ \Cref{i:OWF-ae-Coding-cd} (\Cref{l:coding-avg-to-worst} and \Cref{l:coding-worst-to-avg}).
		 	\item \Cref{i:OWF-ae-SoI} $\iff$ \Cref{i:OWF-ae-SoI-cd} (\Cref{l:soi-avg-to-worst} and \Cref{l:soi-worst-to-avg}).
		 \end{itemize}
		 This will complete the proof of \Cref{t:OWF-ae}.
\end{proof}

\begin{lemma}\label{l:coding-avg-to-worst} 
	We have \emph{(}\Cref{i:OWF-ae-Coding} $\implies$ \Cref{i:OWF-ae-Coding-cd}\emph{)} in \Cref{t:OWF-ae}.
\end{lemma}
\begin{proof}
	Fix $n,t\in\mathbb{N}$ such that $t\geq n$. Let $\{\mathcal{D}_n\}_{n \in\mathbb{N}}$ be any computable distribution family. Also, Let $\alpha$ be any integer and 	let $c>0$ be a constant to be specified later. It suffices to show that for any $(x,y)\in \bool^n \times \bool^n$, if $\mathsf{cd}^t(x,y)\leq \alpha$, then
	\begin{equation}\label{eq:worst-coding}
		\mathsf{rK}^{({2^{\alpha}\cdot t})^c}(x\mid y)\leq  \log \frac{1}{\D_n(x \mid y)} + c\cdot(\log t + \alpha).
	\end{equation}
	We will show the contrapositive. That is, if \Cref{eq:worst-coding} is false, then $\mathsf{cd}^t(x,y)> \alpha$.
	
	Let $d>0$ be a sufficiently large constant, and let
	\[
	k\vcentcolon= 2^{\alpha}\cdot t^d.
	\]
	
	We defined the following polynomial-time samplable distribution family $\{Q_{\langle n,t\rangle}\}_{n,t\in\mathbb{N}}$, where each $Q_{\langle n,t\rangle}$ does the following.
	\begin{enumerate}
		\item Pick $s\sim [2n]$.
		\item Pick $r\sim\bool^{t}$.
		\item Pick $\Pi \sim \bool^{s}$.
		\item Run $U(\Pi,r)$ for $t$ steps. If we obtain a pair $(x,y)\in\bool^{n}\times \bool^{n}$, output $(x,y)$. Otherwise output $(0^n,0^n)$\footnote{Here, we let $Q_{\langle n,t\rangle}$ output pairs of strings in $\bool^n\times \bool^n$. By using padding, we can also ensure that $Q_{\langle n,t\rangle}$ outputs pairs of strings in $\bool^{\langle n,t\rangle}\times \bool^{\langle n,t\rangle}$. This will not affect the correctness of the argument. We omit this technicality for simplicity of presentation.}.
	\end{enumerate}
	It is easy to see from the definition of $\pK^t$ that for every $(x,y)\in\bool^n\times\bool^n$,
	\begin{equation}\label{eq:universal-pK}
		Q_{\langle n,t\rangle}(x,y)\geq \frac{1}{O(n)}\cdot\frac{2}{3} \cdot \frac{1}{2^{\pK^{t}(x,y)}}.
	\end{equation}

	By applying \Cref{i:OWF-ae-Coding} of \Cref{t:OWF-ae} to $\{Q_{\langle n,t\rangle}\}$ and by letting $d$ be sufficiently large, we have
	\begin{equation}\label{eq:avg-coding}
		\Prob_{(x,y)\sim Q_{\langle n,t\rangle}} \mleft[\mathsf{rK}^{(tk)^d}(x\mid y)\leq \log \frac{1}{Q_{\langle n,t\rangle}(x\mid y)}+ d\cdot\log (tk) \mright] \geq 1-\frac{1}{2k}.
	\end{equation}
	Also, by \Cref{l:incompressible}, we have
	\begin{equation}\label{eq:K-large}
		\Prob_{(x,y)\sim Q_{\langle n,t\rangle}} \mleft[\K(x\mid y)> \log \frac{1}{Q_{\langle n,t\rangle}(x\mid y)}- \log k- O(\log n) \mright] \geq 1-\frac{1}{2k}.
	\end{equation}
	Moreover, by the coding theorem for (time-unbounded) Kolmogorov complexity (\Cref{t:coding}), we have that for every $(x,y)\in\supp(\D_n)$
	\begin{equation}\label{eq:K-coding}
		\K(x\mid y)\leq \log \frac{1}{\D_n(x\mid y)}+ O(\log n).
	\end{equation}
	By combining Equations (\ref{eq:avg-coding}), (\ref{eq:K-large}) and (\ref{eq:K-coding}), we get that 
	\[
	\Prob_{(x,y)\sim Q_{\langle n,t\rangle}} \mleft[\mathsf{rK}^{(tk)^d}(x\mid y)\leq \frac{1}{\D_n(x\mid y)}+ 2d\cdot\log (tk) \mright] \geq 1-\frac{1}{k}.
	\]
	
	Now, consider the set $E$ of $(x,y)$ such that
	\[
	\mathsf{rK}^{(tk)^d}(x\mid y)\leq \frac{1}{\D_n(x\mid y)}+ 2d\cdot\log (tk).
	\] 
	Note that by letting $c>d$ be a sufficiently large constant, for any $(x,y)$ such that \Cref{eq:worst-coding} is false, we get that $(x,y)\in E$. Therefore, it suffices to show that for every $(x,y)\in E$, we have $\mathsf{cd}^t(x,y)> \alpha$.
	
	First of all, we have
	\[
	\sum_{(x,y)\in E} Q_{\langle n,t\rangle}(x,y) \leq \frac{1}{k},
	\]
	which implies
	\[
	\sum_{(x,y)\in E} Q_{\langle n,t\rangle}(x,y)\cdot k\leq 1.
	\]
	We can then define a distribution $\mathcal{E}$ whose support is $E$ and $\mathcal{E}(x,y)=Q_{\langle n,t\rangle}(x,y)\cdot k$. Note that $\mathcal{E}$ is computable since $E$ is decidable.
	
	Applying the coding theorem (\Cref{t:coding}) on $\mathcal{E}$, we get that for every $(x,y)\in E$,
	\begin{align*}\label{eq:universal-depth}
		\K(x,y) &\leq \log \frac{1}{\mathcal{E}(x,y)} + O(\log t)\\
		&
		= \log \frac{1}{Q_{\langle n,t\rangle}(x,y)\cdot k} + O(\log t).\numberthis
	\end{align*}
	
	Finally, we get that for every $(x,y)\in E$,
	\begin{align*}
		\pK^{t}(x,y)-\K(x,y)&\geq 	\log \frac{1}{Q_{\langle n,t\rangle}(x,y)} -\K(x,y) -O(\log n)\tag{by \Cref{eq:universal-pK}}\\
		&\geq \log k -O(\log t)\tag{by \Cref{eq:universal-depth}}\\ 
		&> \alpha,
	\end{align*}
	as desired.
\end{proof}

\begin{lemma}\label{l:coding-worst-to-avg}
	We have \emph{(}\Cref{i:OWF-ae-Coding-cd} $\implies$ \Cref{i:OWF-ae-Coding}\emph{)} in \Cref{t:OWF-ae}.
\end{lemma}
\begin{proof}
	Fix $n\in\mathbb{N}$ and any polynomial-time samplable distribution samplable distribution family $\{\mathcal{D}_n\}$.

	By \Cref{lemma:random string is computationally shallow}, there exists a polynomial $\rho$ such that
	\[
	\Prob_{(x,y) \sim \mathcal{D}_n} \mleft[ \mathsf{cd}^{\rho(n)}(x,y) \leq O(\log nk) \mright] \geq 1-\frac{1}{k}.
	\]
	Also, by the assumption that \Cref{i:OWF-ae-Coding-cd} in \Cref{t:OWF-ae} is true, there exists a constant $c>0$ such that for $t\vcentcolon=\rho(n)$ and all $(x,y)\in\supp(\D_n)$
	\[
	\mathsf{rK}^{({2^{\mathsf{cd}^t(x,y)}\cdot t})^c}(x\mid y)\leq \log \frac{1}{\D_n(x \mid y)} + c\cdot(\log t + \mathsf{cd}^t(x,y)).
	\]

	It follows that with probability at least $1-1/k$ over $(x,y)\sim\D_n$, 
	\[
	\mathsf{rK}^{(nk\rho(n))^{O(c)}}(x\mid y)\leq \log \frac{1}{\D_n(x \mid y)} + c\cdot(\log \rho(n) + O(\log nk)),
	\]
	which implies 
	\[
	\mathsf{rK}^{p(n,k)}(x\mid y)\leq \log \frac{1}{\D_n(x\mid y)}+ \log p(n,k),
	\]
	where $p$ is a polynomial.
\end{proof}

\begin{lemma}\label{l:soi-avg-to-worst}
	We have \emph{(}\Cref{i:OWF-ae-SoI} $\implies$ \Cref{i:OWF-ae-SoI-cd}\emph{)} in \Cref{t:OWF-ae}.
\end{lemma}
\begin{proof}[Proof Sketch]
	The proof can be easily adapted from that of \Cref{l:coding-avg-to-worst}. The idea is to apply average-case symmetry of information (\Cref{i:OWF-ae-SoI} in in \Cref{t:OWF-ae}) to the distribution family $\{Q_{\langle n,t\rangle}\}_{n,t\in\mathbb{N}}$ as defined in the proof of \Cref{l:coding-avg-to-worst}. This allows us to say that the set of pairs of strings where symmetry of information fails must have large computational depth. We omit the details here.
\end{proof}

\begin{lemma}\label{l:soi-worst-to-avg}
	We have \emph{(}\Cref{i:OWF-ae-SoI-cd} $\implies$ \Cref{i:OWF-ae-SoI}\emph{)} in \Cref{t:OWF-ae}.
\end{lemma}
\begin{proof}[Proof Sketch]
	The proof is essentially the same as that of \Cref{l:coding-worst-to-avg}, using the fact that a string drawn from any polynomial-time samplable distribution has small computational depth (\Cref{lemma:random string is computationally shallow}). We omit the details here.
\end{proof}
\section{One-Way Functions and Average-Case Hardness of \texorpdfstring{$\rK^\poly$}{PDFstring}}
In this section, we prove \Cref{t:owf-rKt-formal}, which is restated below.

\OWFrK*

\subsection{Technical Lemmas}

\begin{lemma}[\cite{LiuP20_focs_conf}]\label{l:liu-pass}
	If $\mathsf{MINrKT}$ is easy on average over the uniform distribution \emph{(}i.e., \Cref{i:owf-rKt-4} in \Cref{t:owf-rKt-formal} holds\emph{)}, then infinitely-often one-way functions do not exist.
\end{lemma}

The following is a key lemma for proving \Cref{t:owf-rKt-formal}.
\begin{lemma}\label{l:LP-OWF-cond}
	If infinitely-often one-way functions do not exist, then for every $\lambda\in[0,1)$, there exists a probabilistic polynomial-time algorithm $A$ such that for all $n,t,\ell,k\in\mathbb{N}$ with $t\geq n^{1.01}$, with probability at least $1-1/k$ over the internal randomness of $A$, 
	\begin{equation}\label{eq:uni-dis-cond}
		\sum_{x\in\bool^n} 2^{-\rK^t_{\lambda}(x)} \cdot 1_{\mleft[\text{$A(x,1^t,1^\ell,1^k; r_{_A})\not\in \lambda\text{-}\mathsf{Search}$-$\mathsf{MINrKT}(x,1^t,1^{\ell})$}\mright]} \leq \frac{\poly(n)}{k}.
	\end{equation}
\end{lemma}

To show \Cref{l:LP-OWF-cond}, we will need the following simple proposition.

\begin{proposition}\label{l:valid}
	For every $\lambda \in [0,1)$, there is a probabilistic polynomial-time algorithm $V$ such that for all $n,t,\ell,k\in\mathbb{N}$, the following holds with probability at least $1-2^{k}$ over the interval randomness of $V$. For every randomized program $\Pi\in\bool^{\leq 2n}$ and $x\in\bool^n$,
	\begin{enumerate}
		\item if within $t$ steps, $\Pi$ outputs $x$ with probability at least $\lambda$, then $V(x,\Pi,1^t,1^{\ell},1^k)$ accepts, and  
		\item if within $t$ steps, $\Pi$ outputs $x$ with probability less than $\lambda-1/\ell$, then $V(x,\Pi,1^t,1^{\ell},1^k)$ rejects.
	\end{enumerate}
\end{proposition}
\begin{proof}[Proof Sketch]
	Given a randomized program $\Pi\in\bool^{\leq 2n}$ and $x\in\bool^n$, we repeatedly runs the randomized program $\Pi$ for $t$ steps, for $\poly(n,k,\ell)$ times and counts the fraction $\mu$ of times that $x$ is obtained. Using standard concentration bounds, it is easy to show that the following holds with probability at least $1-2^{-n^2}\cdot 2^{-k}$. If the condition stated in the first item of the proposition is true, then $\mu\geq \lambda -1/(2\ell)$, and if the condition stated in the first item, then $\mu< \lambda -1/(2\ell)$. Finally, we apply a union bound over all $\Pi\in\bool^{\leq 2n}$ and $x\in\bool^n$.
\end{proof}

We now show \Cref{l:LP-OWF-cond}.

\begin{proof}[Proof of \Cref{l:LP-OWF-cond}]
	We will show the lemma for $\lambda \vcentcolon=2/3$. It is not hard to adapt the proof to any $\lambda \in [0,1)$.
	
	Let $c>0$ be a constant so that $\rK^t(x)\leq n +c$ for every $x\in\bool^n$ and $t\geq n^{1.01}$. Let $V$ be the algorithm in \Cref{l:valid} instantiated with $\lambda$.
	
	Let $f$ be a polynomial-time computable function defined as follows. 
	\begin{quote}
		On input $(r_{_V},i,\Pi, r_{t}, r_\ell, r_{k})$, where $r_{_V}\in\bool^{\poly(t,\ell,k)}$, $i\in [n+c]$, $\Pi\in\bool^{n+c}$, $r_t\in\bool^t$, $r_{\ell}\in\bool^{\ell}$ $r_k\in\bool^k$ and $r_2\in\bool^k$. We run $U(\Pi_{[i]}, r_t)$ for $t$ steps and obtain a string $x$. If $|x|=n$ and $V(x,\Pi_{[i]},1^t,1^{\ell},1^k;r_{_V})=1$, we output $(r_{_V}, i, x, 1^t,1^{\ell},1^k)$; otherwise we output $\bot$.
	\end{quote}	
 
	Since we assume infinitely-often one-way functions do not exist (which implies infinitely-often weak one-way functions do not exist), there is a polynomial-time algorithm $\mathsf{Invert}$ such that for all $n,t,\ell,k\in\mathbb{N}$, it holds that
	\[
	\Prob \mleft[\mathsf{Invert}(r_{_V},i, x, 1^t,1^{\ell},1^k; r_{_\mathsf{Invert}}) \textnormal{ succeeds}\mright]\geq 1-\frac{1}{2k^2},
	\]
	where $(r_{_V}, i, x, 1^t,1^{\ell},1^k)$ is sampled according to $f$, $r_{_\mathsf{Invert}}\in\bool^{\poly(t,\ell,k)}$ is the internal randomness of $\mathsf{Invert}$, and  ``$\mathsf{Invert}(r_{_V}, i, x, 1^t,1^{\ell},1^k)$ succeeds'' means $\mathsf{Invert}(i, x, 1^t,1^{\ell},1^k)$ returns a pre-image of $(r_{_V}; i, x, 1^t,1^{\ell},1^k)$.
	
	By an averaging argument, we get that with probability at least $1-1/(2k)$ over $r_{_V}$ (the internal randomness of $V$ used in the definition of $f$) and $r_{_\mathsf{Invert}}$ (the internal randomness of $\mathsf{Invert}$), it holds that
	\begin{equation}\label{eq:failure-unnder-LP-cond}
		\Prob \mleft[\mathsf{Invert}(r_{_V},i, x, 1^t,1^{\ell},1^k; r_{_\mathsf{Invert}}) \textnormal{ succeeds}\mright]\geq 1-\frac{1}{k},
	\end{equation}
	where the above probability is only over $i$ and $x$.
    Also, with probability at least $1-2^{-k}$ over a uniformly random $r_{_V}$, the condition stated in \Cref{l:valid} will hold. By a union bound, with probability at least $1-1/k$, both \Cref{eq:failure-unnder-LP-3-cond} and the condition stated in \Cref{l:valid} hold for a uniform random $(r_{_V}, r_{_\mathsf{Invert}})$. Let us consider any such \emph{good} pair $(r_{_V}, r_{_\mathsf{Invert}})$.
	
	By a union bound, \Cref{eq:failure-unnder-LP-cond} yields that for all $i\in[n+c]$, 
	\begin{equation}\label{eq:failure-unnder-LP-2-cond}
		\Prob \mleft[\mathsf{Invert}(r_{_V},i, x, 1^t,1^{\ell},1^k; r_{_\mathsf{Invert}}) \textnormal{ succeeds}\mright]\geq 1-\frac{n+c}{k},
	\end{equation}
	where now the probability is only over $x$.
	
	Next, for every $i\in[n+c]$, let $\D_i$ be the following distribution (note that we have also fixed a good $r_{_V}$):
	\begin{enumerate}
		\item Pick $\Pi\sim\bool^{n+c}$.
		\item Pick $r_t\sim\bool^t$.
		\item Run $U(\Pi_{[i]}, r_t)$ for $t$ steps and obtain a string $x$. If $|x|=n$ and $V(x,\Pi_{[i]},1^t,1^{\ell},1^k;r_{_V})=1$, we output $x$; otherwise output $\bot$.
	\end{enumerate}
	Then \Cref{eq:failure-unnder-LP-2-cond} implies that for all $i\in[n+c]$,
	\begin{equation}\label{eq:failure-unnder-LP-3-cond}
		\Prob_{x\sim \D_i} \mleft[\mathsf{Invert}(r_{_V},i, x, 1^t,1^{\ell},1^k; r_{_\mathsf{Invert}}) \textnormal{ fails}\mright]\leq \frac{n+c}{k}.
	\end{equation}
	
	Now consider the following algorithm $A$ for solving $\mathsf{Search}$-$\mathsf{MINrKT}$:
	\begin{quote}
		On input $(x,1^t,1^{\ell},1^k)$, we pick $r_{_V}\sim\bool^{\poly(t,\ell,k)}$ and $r_{_\mathsf{Invert}}\sim \bool^{\poly(t,\ell,k)}$. We then try to find the smallest $i\in[n+c]$ such that the following holds: $\mathsf{Invert}(r_{_V},i, x, 1^t,1^{\ell},1^k; r_{_\mathsf{Invert}})$ returns some $(r_{_V},i,\Pi, r_{t}, r_\ell, r_{k})$ such that after running $U(\Pi_{[i]}, r_t)$ for $t$ steps, $x$ is obtained, and that  $V(x,\Pi_{[i]},1^t,1^{\ell},1^k;r_{_V})=1$. In this case, we output $\Pi_{[i]}$.
	\end{quote} 
	
	We claim that the algorithm $A$ satisfies the condition stated in the lemma. 
	
	First note that in the description of the above algorithm, with probability at least $1-1/k$, the internal randomness of $A$, $r_{_A}\vcentcolon=(r_{_V}, r_{_\mathsf{Invert}})$, will be good, and hence satisfy \Cref{eq:failure-unnder-LP-3-cond} and the condition stated in \Cref{l:valid}. In what follows, we fix such a good $r_{_A}$. 
 
    For the sake of contradiction, suppose
	\begin{equation}\label{eq:contradic-cond}
		\sum_{x\in\bool^n} 2^{-\rK^t(x)} \cdot 1_{\mleft[\text{$A(x,1^t,1^\ell,1^k; r_{_A})\not\in \mathsf{Search}$-$\mathsf{MINrKT}(x,1^t,1^{\ell})$}\mright]} \leq \frac{n^b}{k},
	\end{equation}
	where $b>0$ is a constant to be specified later.
	
	Note that for every $i\in [n+c]$ and every $x\in\bool^n$ such that $\rK^t(x)=i$, we have
	\begin{equation}\label{eq:dominate-Kt-1-cond}
		\D_i(x) \geq 2^{-\rK^t(x)}\cdot \frac{2}{3}.
	\end{equation}
	This is because in the sampling procedure of $\D_i$, with probability $2^{-\rK^t(x)}$, we will pick a $\Pi$ whose $i$-bit prefix corresponds to a $\rK^t$-witness of $x$. In this case, with probability at least $2/3$ over $r_t$, we obtain $x$. Finally, note that by the property of $V$ and the fact that $r_{_V}$ is good, in this case we have $V(x,\Pi_{[i]},1^t,1^{\ell},1^k; r_{_V})=1$ and hence $x$ will be output.
	
	Also, for every $i\in [n+c]$ and every $x\in\bool^n$ such that $\rK^t(x)=i$, if there exists an $i^*\leq i$ such that $\mathsf{Invert}(r_{_V},i^*, x, 1^t,1^{\ell},1^k;r_{_\mathsf{Invert}})$ succeeds, then it means we obtain some $(r_{_V},i^*,\Pi, r_{t}, r_\ell, r_{k})$ such that after running $U(\Pi_{[i^*]}, r_t)$ for $t$ steps, $x$ is obtained and that $V(x,\Pi_{[i^*]},1^t,1^{\ell},1^k;r_{_V})=1$. Again, by the property of $V$ and the fact that $r_{_V}$ is good, $\Pi_{[i^*]}$ is an $(1/\ell)$-$\rK^t$-witness of $x$, which means $A(x,1^t, 1^{\ell}, 1^k;r_{_V})\in \mathsf{Search}$-$\mathsf{MINrKT}(x,1^t,1^{\ell})$.
	
	In other words, for every $i\in [n+c]$ and every $x\in\bool^n$ such that $\rK^t(x)=i$, if $A(x,1^t, 1^{\ell}, 1^k;r_{_A})\not\in \mathsf{Search}$-$\mathsf{MINrKT}(x,1^t,1^{\ell})$, then for all $i^*\leq i$, $\mathsf{Invert}(r_{_V},i^*, x, 1^t,1^{\ell},1^k;r_{_\mathsf{Invert}})$ fails. In particular, the latter holds for $i^*=i$.
	
	Then we have
	\begin{align*}
		\frac{n^b}{k} &\leq \sum_{i\in [n+c]}\,\, \sum_{\substack{x\in\bool^n:\\ \rK^t(x)=i}} 2^{-\rK^t(x)} \cdot 1_{\mleft[\text{$A(x,1^t,1^{\ell},1^k;r_{_A})\not\in \mathsf{Search}$-$\mathsf{MINrKT}(x,1^t,1^{\ell})$}\mright]} \tag{by \Cref{eq:contradic-cond}}\\
		&\leq  \sum_{i} \sum_{\substack{x\in\bool^n:\\ \rK^t(x)=i}} \frac{3}{2}\cdot \D_i(x) \cdot 1_{\mleft[\text{$A(x,1^t,1^{\ell},1^k;r_{_A})\not\in \mathsf{Search}$-$\mathsf{MINrKT}(x,1^t,1^{\ell})$}\mright]} \tag{\text{by \Cref{eq:dominate-Kt-1-cond}}}\\
		&\leq \frac{3}{2}\cdot  \sum_{i} \sum_{\substack{x\in\bool^n:\\ \rK^t(x)=i}} \D_i(x) \cdot 1_{\mleft[\mathsf{Invert}(r_{_V},i^*, x, 1^t,1^{\ell},1^k;r_{_\mathsf{Invert}})\textnormal{ fails}\mright]}.
	\end{align*}
	By averaging, the above implies that there exists some $i$ such that 
	\[
	\sum_{\substack{x\in\bool^n:\\ \rK^t(x)=i}} \D_i(x) \cdot 1_{\mleft[\mathsf{Invert}(r_{_V},i, x, y, 1^t, 1^{k};r_{_\mathsf{Invert}})\textnormal{ fails}\mright]} \geq \frac{2\cdot n^b}{3\cdot (n+c)\cdot k},
	\]
	which contradicts \Cref{eq:failure-unnder-LP-3-cond} by letting $b$ be a sufficiently large constant.
\end{proof}

\subsection{Proof of \texorpdfstring{\Cref{t:owf-rKt-formal}}{}}

We now show \Cref{t:owf-rKt-formal}.

\begin{proof}[Proof of \Cref{t:owf-rKt-formal}]
	(\Cref{i:owf-rKt-2}$\implies$\Cref{i:owf-rKt-3}) and (\Cref{i:owf-rKt-3}$\implies$\Cref{i:owf-rKt-4}) are trivial.     (\Cref{i:owf-rKt-4}$\implies$\Cref{i:owf-rKt-1}) follows from \Cref{l:liu-pass}. Below, we show (\Cref{i:owf-rKt-1}$\implies$\Cref{i:owf-rKt-2}).
	
	Let $\lambda \in [0,1)$, and let $\{\D_{n}\}_{n}$ be a polynomial-time samplable distribution family. Also let $\rho$ to be a polynomial to be specified later.
	
	Fix $n,\ell,k\in\mathbb{N}$ and $t\geq \rho(n,k)$.
	
	Let $A$ be the polynomial-time algorithm in \Cref{l:LP-OWF-cond}, and let $d>0$ be some constant to be specified later. 
	We have that with probability at least $1-1/(2k)$ over the internal randomness $r_{_A}$ of $A$, 
	\begin{equation}\label{eq:failure-under-universal-cond}
		\sum_{x\in\bool^n} 2^{-\rK^t_{\lambda}(x)} \cdot 1_{\mleft[\text{$A(x,1^t,1^\ell,1^{(nk)^d}; r_{_A})\not\in \mathsf{Search}$-$\mathsf{MINrKT}(x,1^t,1^{\ell})$}\mright]} \leq \frac{1}{(nk)^d}.		
	\end{equation}

	Also, by \Cref{optimal coding-ae}, there exists a polynomial $p$ such that,
	\begin{equation}\label{coding-uner-D-cond}
		\Prob_{x\sim\D_{n}}\mleft[\rK^{p(n,k)}(x\mid y)\le\log\frac{1}{\D_{n}(x)}+\log p(n,k)\mright]\ge 1-\frac{1}{4k}.
	\end{equation}
	
	Fix any \emph{good} $r_{_A}$ such that \Cref{eq:failure-under-universal-cond} holds. 
	We claim that
	\begin{equation}\label{eq:x-good-cond}
		\Prob_{x\sim \D_{n}}\mleft[A(x,1^t,1^\ell,1^{(nk)^d}; r_{_A})\in \mathsf{Search}\text{-}\mathsf{MINrKT}(x,1^t,1^{\ell})\mright] \geq 1-\frac{1}{2k}.
	\end{equation}
	Note that this suffices to show the theorem, since $r_{_A}$ is good with probability at least $1-1/(2k)$.
	
	Suppose, for the sake of contradiction, \Cref{eq:x-good-cond} is not true. Then 
	\begin{equation}\label{failure-uner-D-cond}
		\Prob_{x\sim \D_{n}}\mleft[A(x,1^t,1^\ell,1^{(nk)^d}; r_{_A})\not\in \mathsf{Search}\text{-}\mathsf{MINrKT}(x,1^t,1^{\ell})\mright] >\frac{1}{2k}.
	\end{equation}
	Let $\mathcal{E}(x)$ be the event that both the following hold.
	\begin{itemize}
		\item \text{$A(x,1^t,1^\ell,1^{(nk)^d}; r_{_A})\not\in \mathsf{Search}\text{-}\mathsf{MINrKT}(x,1^t,1^{\ell})$}
		\item $\rK^{p(n,k)}(x)\le\log\frac{1}{\D_{n}(x)}+\log p(n,k)$.
	\end{itemize}
	By \Cref{coding-uner-D-cond} and \Cref{failure-uner-D-cond}, we get that
	\begin{equation}\label{eq:large-over-D-cond}
		\sum_{x\in\bool^n} \D_n(x) \cdot 1_{\mathcal{E}(x)} \geq \frac{1}{4k}.		    
	\end{equation}
	Note that whenever $\mathcal{E}(x)$ holds, we have
	\begin{equation}\label{eq:dominate-Kt-2-cond}
		\D_{n}(x) \leq  \frac{p(n,k)}{2^{\rK^{p(n,k)}(x)}}.
	\end{equation}
	Now we have
	\begin{align*}
		\frac{1}{4k} &\leq 	\sum_{x\in\bool^n} \D_n(x) \cdot 1_{\mathcal{E}(x)} \tag{\text{by \Cref{eq:large-over-D-cond}}}\\
		&\leq  \sum_{x\in\bool^n} \frac{p(n,k)}{2^{\rK^{p(n,k)}(x)}} \cdot 1_{\mathcal{E}(x)} \tag{\text{by \Cref{eq:dominate-Kt-2-cond}}}\\
		&\leq  p(n,k) \cdot \sum_{x\in\bool^n} 2^{-\K^{p(n,k)}(x)} \cdot 1_{\mathcal{E}(x)}\\
		&\leq  p(n,k) \cdot \sum_{x\in\bool^n} 2^{-\rK^{p(n,k)}(x)} \cdot 1_{\mleft[\text{$A(x,1^t,1^\ell,1^{(nk)^d}; r_{_A})\not\in \mathsf{Search}\text{-}\mathsf{MINrKT}(x,1^t,1^{\ell})$}\mright]}\label{eq:dominated}
		\numberthis
	\end{align*}
	Note that if $\lambda \leq 2/3$, we have $2^{-\rK^{p(n,k)}(x)}\leq 2^{-\rK^{t}_{\lambda}(x)}$ for all $t\geq p(n,k)$. On the other hand, if $\lambda >2/3$, then by \Cref{l:success_amp}, we have $2^{-\rK^{p(n,k)}(x)}\leq 2^{-\rK^{t}_{\lambda}(x)+O(\log (1/(1-\lambda))}$ for  $t\geq p(n,k)\cdot O(\log (1/(1-\lambda)))$. In both cases, if $t\geq \rho(n,k)$ for some sufficiently large polynomial $\rho$, we have
	\[
	\text{\Cref{eq:dominated}} \leq \mleft(\frac{1}{1-\lambda}\mright)^{O(1)} \cdot p(n,k) \cdot \sum_{x\in\bool^n} 2^{-\rK^{t}_{\lambda}(x)} \cdot 1_{\mleft[\text{$A(x,1^t,1^\ell,1^{(nk)^d}; r_{_A})\not\in \mathsf{Search}\text{-}\mathsf{MINrKT}(x,1^t,1^{\ell})$}\mright]}.
	\]
	By rearranging, we get
	\[
	\sum_{x\in\bool^n} 2^{-\rK^{t}_{\lambda}(x)} \cdot 1_{\mleft[\text{$A(x,1^t,1^\ell,1^{(nk)^d}; r_{_A})\not\in \mathsf{Search}\text{-}\mathsf{MINrKT}(x,1^t,1^{\ell})$}\mright]} \geq \frac{(1-\lambda)^{O(1)}}{4k\cdot p(n,k)}.
	\]
	However, this contradicts \Cref{eq:failure-under-universal-cond} by letting $d$ be a sufficiently large constant.
\end{proof}

\printbibliography

@article{HastadILL99_siamcomp_journals,
  author = {H{\aa}stad, Johan and Impagliazzo, Russell and Levin, Leonid A. and Luby, Michael},
  title = {{A Pseudorandom Generator from any One-way Function}},
  journal = {{SIAM} J. Comput.},
  volume = {28},
  number = {4},
  pages = {1364--1396},
  year = {1999},
  doi = {10.1137/S0097539793244708},
  bibsource = {dblp computer science bibliography, http://dblp.org}
}

@inproceedings{CarmosinoIKK16_coco_conf,
  author = {Carmosino, Marco L. and Impagliazzo, Russell and Kabanets, Valentine and Kolokolova, Antonina},
  title = {{Learning Algorithms from Natural Proofs}},
  booktitle = {Proceedings of the Conference on Computational Complexity (CCC)},
  pages = {10:1--10:24},
  year = {2016},
  doi = {10.4230/LIPIcs.CCC.2016.10},
  bibsource = {dblp computer science bibliography, https://dblp.org}
}

@article{Vadhan12_fttcs_journals,
  author = {Vadhan, Salil P.},
  title = {{Pseudorandomness}},
  journal = {Foundations and Trends in Theoretical Computer Science},
  volume = {7},
  number = {1-3},
  pages = {1--336},
  year = {2012},
  doi = {10.1561/0400000010},
  bibsource = {dblp computer science bibliography, https://dblp.org}
}

@article{Trevisan01_jacm_journals,
  author = {Trevisan, Luca},
  title = {{Extractors and pseudorandom generators}},
  journal = {J. {ACM}},
  volume = {48},
  number = {4},
  pages = {860--879},
  year = {2001},
  doi = {10.1145/502090.502099},
  bibsource = {dblp computer science bibliography, https://dblp.org}
}

@article{RazRV02_jcss_journals,
  author = {Raz, Ran and Reingold, Omer and Vadhan, Salil P.},
  title = {{Extracting all the Randomness and Reducing the Error in {Trevisan's} Extractors}},
  journal = {J. Comput. Syst. Sci.},
  volume = {65},
  number = {1},
  pages = {97--128},
  year = {2002},
  doi = {10.1006/jcss.2002.1824},
  bibsource = {dblp computer science bibliography, https://dblp.org}
}

@inproceedings{Yao82a_focs_conf,
  author = {Yao, Andrew Chi{-}Chih},
  title = {{Theory and Applications of Trapdoor Functions (Extended Abstract)}},
  booktitle = {Proceedings of the Symposium on Foundations of Computer Science (FOCS)},
  pages = {80--91},
  year = {1982},
  doi = {10.1109/SFCS.1982.45},
  bibsource = {dblp computer science bibliography, https://dblp.org}
}

@article{BogdanovT06_fttcs_journals,
  author = {Bogdanov, Andrej and Trevisan, Luca},
  title = {{Average-Case Complexity}},
  journal = {Foundations and Trends in Theoretical Computer Science},
  volume = {2},
  number = {1},
  year = {2006},
  doi = {10.1561/0400000004},
  bibsource = {dblp computer science bibliography, https://dblp.org}
}

@inproceedings{Impagliazzo95_coco_conf,
  author = {Impagliazzo, Russell},
  title = {{A Personal View of Average-Case Complexity}},
  booktitle = {Proceedings of the Structure in Complexity Theory Conference},
  pages = {134--147},
  year = {1995},
  doi = {10.1109/SCT.1995.514853},
  bibsource = {dblp computer science bibliography, https://dblp.org}
}

@article{Levin86_siamcomp_journals,
  author = {Levin, Leonid A.},
  title = {{Average Case Complete Problems}},
  journal = {{SIAM} J. Comput.},
  volume = {15},
  number = {1},
  pages = {285--286},
  year = {1986},
  doi = {10.1137/0215020},
  bibsource = {dblp computer science bibliography, https://dblp.org}
}

@article{Trakhtenbrot84_annals_journals,
  author = {Trakhtenbrot, Boris A.},
  title = {{A Survey of {R}ussian Approaches to Perebor (Brute-Force Searches) Algorithms}},
  journal = {{IEEE} Annals of the History of Computing},
  volume = {6},
  number = {4},
  pages = {384--400},
  year = {1984},
  doi = {10.1109/MAHC.1984.10036},
  bibsource = {dblp computer science bibliography, https://dblp.org}
}

@inproceedings{ImpagliazzoL90_focs_conf,
  author = {Impagliazzo, Russell and Levin, Leonid A.},
  title = {{No Better Ways to Generate Hard {NP} Instances than Picking Uniformly at Random}},
  booktitle = {Proceedings of the Symposium on Foundations of Computer Science (FOCS)},
  pages = {812--821},
  year = {1990},
  doi = {10.1109/FSCS.1990.89604},
  bibsource = {dblp computer science bibliography, https://dblp.org}
}

@article{ZvonkinL1970,
  title = {{The complexity of finite objects and the development of the concepts of information and randomness by means of the theory of algorithms}},
  author = {Zvonkin, Alexander K and Levin, Leonid A},
  journal = {Russian Mathematical Surveys},
  volume = {25},
  number = {6},
  pages = {83--124},
  year = {1970},
  publisher = {Turpion Ltd}
}

@article{LongpreW95_iandc_journals,
  author = {Longpr{\'{e}}, Luc and Watanabe, Osamu},
  title = {{On Symmetry of Information and Polynomial Time Invertibility}},
  journal = {Inf. Comput.},
  volume = {121},
  number = {1},
  pages = {14--22},
  year = {1995},
  doi = {10.1006/inco.1995.1120},
  bibsource = {dblp computer science bibliography, https://dblp.org}
}

@article{LeeR05_tcs_journals,
  author = {Lee, Troy and Romashchenko, Andrei E.},
  title = {{Resource bounded symmetry of information revisited}},
  journal = {Theor. Comput. Sci.},
  volume = {345},
  number = {2-3},
  pages = {386--405},
  year = {2005},
  doi = {10.1016/j.tcs.2005.07.017},
  bibsource = {dblp computer science bibliography, https://dblp.org}
}

@inproceedings{ImpagliazzoL89_focs_conf,
  author = {Impagliazzo, Russell and Luby, Michael},
  title = {{One-way Functions are Essential for Complexity Based Cryptography (Extended Abstract)}},
  booktitle = {Proceedings of the Symposium on Foundations of Computer Science (FOCS)},
  pages = {230--235},
  year = {1989},
  doi = {10.1109/SFCS.1989.63483},
  bibsource = {dblp computer science bibliography, https://dblp.org}
}

@article{TrevisanVZ05_cc_journals,
  author = {Trevisan, Luca and Vadhan, Salil P. and Zuckerman, David},
  title = {{Compression of Samplable Sources}},
  journal = {Computational Complexity},
  volume = {14},
  number = {3},
  pages = {186--227},
  year = {2005},
  doi = {10.1007/s00037-005-0198-6},
  bibsource = {dblp computer science bibliography, https://dblp.org}
}

@inproceedings{VadhanZ12_stoc_conf,
  author = {Vadhan, Salil P. and Zheng, Colin Jia},
  title = {{Characterizing pseudoentropy and simplifying pseudorandom generator constructions}},
  booktitle = {Proceedings of the Symposium on Theory of Computing (STOC)},
  pages = {817--836},
  year = {2012},
  doi = {10.1145/2213977.2214051},
  bibsource = {dblp computer science bibliography, https://dblp.org}
}

@inproceedings{Wee04_coco_conf,
  author = {Wee, Hoeteck},
  title = {{On Pseudoentropy versus Compressibility}},
  booktitle = {Proceedings of the Conference on Computational Complexity (CCC)},
  pages = {29--41},
  year = {2004},
  doi = {10.1109/CCC.2004.1313782},
  bibsource = {dblp computer science bibliography, https://dblp.org}
}

@article{AntunesFMV06_tcs_journals,
  author = {Antunes, Luis and Fortnow, Lance and van Melkebeek, Dieter and Vinodchandran, N. V.},
  title = {{Computational depth: Concept and applications}},
  journal = {Theor. Comput. Sci.},
  volume = {354},
  number = {3},
  pages = {391--404},
  year = {2006},
  doi = {10.1016/j.tcs.2005.11.033},
  bibsource = {dblp computer science bibliography, https://dblp.org}
}

@inproceedings{AntunesF09_coco_conf,
  author = {Antunes, Luis Filipe Coelho and Fortnow, Lance},
  title = {{Worst-Case Running Times for Average-Case Algorithms}},
  booktitle = {Proceedings of the Conference on Computational Complexity (CCC)},
  pages = {298--303},
  year = {2009},
  doi = {10.1109/CCC.2009.12},
  bibsource = {dblp computer science bibliography, https://dblp.org}
}

@inproceedings{BarakSW03_random_conf,
  author = {Barak, Boaz and Shaltiel, Ronen and Wigderson, Avi},
  title = {{Computational Analogues of Entropy}},
  booktitle = {Proceedings of the Randomization and Approximation Techniques in Computer Science (RANDOM/APPROX)},
  pages = {200--215},
  year = {2003},
  doi = {10.1007/978-3-540-45198-3_18},
  bibsource = {dblp computer science bibliography, https://dblp.org}
}

@article{Ta-ShmaUZ07_combinatorica_journals,
  author = {Ta{-}Shma, Amnon and Umans, Christopher and Zuckerman, David},
  title = {{Lossless Condensers, Unbalanced Expanders, And Extractors}},
  journal = {Combinatorica},
  volume = {27},
  number = {2},
  pages = {213--240},
  year = {2007},
  doi = {10.1007/s00493-007-0053-2},
  bibsource = {dblp computer science bibliography, https://dblp.org}
}

@article{LongpreM93_ipl_journals,
  author = {Longpr{\'{e}}, Luc and Mocas, Sarah},
  title = {{Symmetry of Information and One-Way Functions}},
  journal = {Inf. Process. Lett.},
  volume = {46},
  number = {2},
  pages = {95--100},
  year = {1993},
  doi = {10.1016/0020-0190(93)90204-M},
  bibsource = {dblp computer science bibliography, https://dblp.org}
}

@inproceedings{Oliveira19_icalp_conf,
  author = {Oliveira, Igor Carboni},
  title = {{Randomness and Intractability in Kolmogorov Complexity}},
  booktitle = {Proceedings of the International Colloquium on Automata, Languages, and Programming (ICALP)},
  pages = {32:1--32:14},
  year = {2019},
  doi = {10.4230/LIPIcs.ICALP.2019.32},
  bibsource = {dblp computer science bibliography, https://dblp.org}
}

@article{GuruswamiUV09_jacm_journals,
  author = {Guruswami, Venkatesan and Umans, Christopher and Vadhan, Salil P.},
  title = {{Unbalanced expanders and randomness extractors from Parvaresh-Vardy codes}},
  journal = {J. {ACM}},
  volume = {56},
  number = {4},
  pages = {20:1--20:34},
  year = {2009},
  doi = {10.1145/1538902.1538904},
  bibsource = {dblp computer science bibliography, https://dblp.org}
}

@inproceedings{ImpagliazzoZ89_focs_conf,
  author = {Impagliazzo, Russell and Zuckerman, David},
  title = {{How to Recycle Random Bits}},
  booktitle = {Proceedings of the Symposium on Foundations of Computer Science (FOCS)},
  pages = {248--253},
  year = {1989},
  doi = {10.1109/SFCS.1989.63486},
  bibsource = {dblp computer science bibliography, https://dblp.org}
}

@book{0016881_daglib_books,
  author = {Cover, Thomas M. and Thomas, Joy A.},
  title = {{Elements of information theory {(2.} ed.)}},
  publisher = {Wiley},
  year = {2006},
  isbn = {978-0-471-24195-9},
  bibsource = {dblp computer science bibliography, https://dblp.org}
}

@inproceedings{Ilango20_coco_conf,
  author = {Ilango, Rahul},
  title = {{Connecting Perebor Conjectures: Towards a Search to Decision Reduction for Minimizing Formulas}},
  booktitle = {Proceedings of the Computational Complexity Conference (CCC)},
  pages = {31:1--31:35},
  year = {2020},
  doi = {10.4230/LIPIcs.CCC.2020.31},
  bibsource = {dblp computer science bibliography, https://dblp.org}
}

@inproceedings{LiuP20_focs_conf,
  author = {Liu, Yanyi and Pass, Rafael},
  title = {{On One-way Functions and Kolmogorov Complexity}},
  booktitle = {Proceedings of the Symposium on Foundations of Computer Science (FOCS)},
  pages = {1243--1254},
  year = {2020},
  doi = {10.1109/FOCS46700.2020.00118},
  bibsource = {dblp computer science bibliography, https://dblp.org}
}

@inproceedings{Hirahara21_stoc_conf,
  author = {Hirahara, Shuichi},
  title = {{Average-case hardness of {NP} from exponential worst-case hardness assumptions}},
  booktitle = {Proceedings of the Symposium on Theory of Computing (STOC)},
  pages = {292--302},
  year = {2021},
  doi = {10.1145/3406325.3451065},
  bibsource = {dblp computer science bibliography, https://dblp.org}
}

@inproceedings{LuO21_icalp_conf,
  author = {Lu, Zhenjian and Oliveira, Igor Carboni},
  title = {{An Efficient Coding Theorem via Probabilistic Representations and Its Applications}},
  booktitle = {Proceedings of the International Colloquium on Automata, Languages, and Programming (ICALP)},
  pages = {94:1--94:20},
  year = {2021},
  doi = {10.4230/LIPIcs.ICALP.2021.94},
  bibsource = {dblp computer science bibliography, https://dblp.org}
}

@inproceedings{RenS21_coco_conf,
  author = {Ren, Hanlin and Santhanam, Rahul},
  title = {{Hardness of {KT} Characterizes Parallel Cryptography}},
  booktitle = {Proceedings of the Computational Complexity Conference (CCC)},
  pages = {35:1--35:58},
  year = {2021},
  doi = {10.4230/LIPIcs.CCC.2021.35},
  bibsource = {dblp computer science bibliography, https://dblp.org}
}

@inproceedings{LiuP21_crypto_conf,
  author = {Liu, Yanyi and Pass, Rafael},
  title = {{On the Possibility of Basing Cryptography on EXP{\(\not =\)} {BPP}}},
  booktitle = {Proceedings of the International Cryptology Conference (CRYPTO)},
  pages = {11--40},
  year = {2021},
  doi = {10.1007/978-3-030-84242-0_2},
  bibsource = {dblp computer science bibliography, https://dblp.org}
}

@inproceedings{AllenderCMTV21_fsttcs_conf,
  author = {Allender, Eric and Cheraghchi, Mahdi and Myrisiotis, Dimitrios and Tirumala, Harsha and Volkovich, Ilya},
  title = {{One-Way Functions and a Conditional Variant of {MKTP}}},
  booktitle = {Proceedings of the Foundations of Software Technology and Theoretical Computer Science (FSTTCS)},
  pages = {7:1--7:19},
  year = {2021},
  doi = {10.4230/LIPIcs.FSTTCS.2021.7},
  bibsource = {dblp computer science bibliography, https://dblp.org}
}

@article{GoldbergK22_eccc_journals,
  author = {Goldberg, Halley and Kabanets, Valentine},
  title = {{A Simpler Proof of the Worst-Case to Average-Case Reduction for Polynomial Hierarchy via Symmetry of Information}},
  journal = {Electronic Colloquium on Computational Complexity {(ECCC)}},
  volume = {007},
  year = {2022}
}

@article{GoldbergS91_siamcomp_journals,
  author = {Goldberg, Andrew V. and Sipser, Michael},
  title = {{Compression and Ranking}},
  journal = {{SIAM} J. Comput.},
  volume = {20},
  number = {3},
  pages = {524--536},
  year = {1991},
  doi = {10.1137/0220034},
  bibsource = {dblp computer science bibliography, https://dblp.org}
}

@inproceedings{HsiaoLR07_eurocrypt_conf,
  author = {Hsiao, Chun{-}Yuan and Lu, Chi{-}Jen and Reyzin, Leonid},
  title = {{Conditional Computational Entropy, or Toward Separating Pseudoentropy from Compressibility}},
  booktitle = {Proceedings of the International Conference on the Theory and Applications of Cryptographic Techniques (EUROCRYPT)},
  pages = {169--186},
  year = {2007},
  doi = {10.1007/978-3-540-72540-4_10},
  bibsource = {dblp computer science bibliography, https://dblp.org}
}

@inproceedings{Hirahara20_focs_conf,
  author = {Hirahara, Shuichi},
  title = {{Characterizing Average-Case Complexity of {PH} by Worst-Case Meta-Complexity}},
  booktitle = {Proceedings of the Symposium on Foundations of Computer Science (FOCS)},
  pages = {50--60},
  year = {2020},
  doi = {10.1109/FOCS46700.2020.00014},
  bibsource = {dblp computer science bibliography, https://dblp.org}
}

@inproceedings{Hirahara18_focs_conf,
  author = {Hirahara, Shuichi},
  title = {{Non-Black-Box Worst-Case to Average-Case Reductions within {NP}}},
  booktitle = {Proceedings of the Symposium on Foundations of Computer Science (FOCS)},
  pages = {247--258},
  year = {2018},
  doi = {10.1109/FOCS.2018.00032},
  bibsource = {dblp computer science bibliography, https://dblp.org}
}

@inproceedings{Hirahara22_coco_conf,
  author = {Hirahara, Shuichi},
  title = {{Symmetry of Information from Meta-Complexity}},
  booktitle = {Proceedings of the Computational Complexity Conference (CCC)},
  pages = {26:1--26:41},
  year = {2022},
  doi = {10.4230/LIPIcs.CCC.2022.26},
  bibsource = {dblp computer science bibliography, https://dblp.org}
}

@inproceedings{LiuP22a_coco_conf,
  author = {Liu, Yanyi and Pass, Rafael},
  title = {{On One-Way Functions from NP-Complete Problems}},
  booktitle = {Proceedings of the Computational Complexity Conference (CCC)},
  pages = {36:1--36:24},
  year = {2022},
  doi = {10.4230/LIPIcs.CCC.2022.36},
  bibsource = {dblp computer science bibliography, https://dblp.org}
}

@article{Naor91_joc_journals,
  author = {Naor, Moni},
  title = {{Bit Commitment Using Pseudorandomness}},
  journal = {J. Cryptol.},
  volume = {4},
  number = {2},
  pages = {151--158},
  year = {1991},
  doi = {10.1007/BF00196774},
  bibsource = {dblp computer science bibliography, https://dblp.org}
}

@inproceedings{Rompel90_stoc_conf,
  author = {Rompel, John},
  title = {{One-Way Functions are Necessary and Sufficient for Secure Signatures}},
  booktitle = {Proceedings of the Symposium on Theory of Computing (STOC)},
  pages = {387--394},
  year = {1990},
  doi = {10.1145/100216.100269},
  bibsource = {dblp computer science bibliography, https://dblp.org}
}

@inproceedings{HiraharaN22_coco_conf,
  author = {Hirahara, Shuichi and Nanashima, Mikito},
  title = {{Finding Errorless Pessiland in Error-Prone Heuristica}},
  booktitle = {Proceedings of the Computational Complexity Conference (CCC)},
  pages = {25:1--25:28},
  year = {2022},
  doi = {10.4230/LIPIcs.CCC.2022.25},
  bibsource = {dblp computer science bibliography, https://dblp.org}
}

@inproceedings{HiraharaS22_innovations_conf,
  author = {Hirahara, Shuichi and Santhanam, Rahul},
  title = {{Errorless Versus Error-Prone Average-Case Complexity}},
  booktitle = {Proceedings of the Innovations in Theoretical Computer Science Conference (ITCS)},
  pages = {84:1--84:23},
  year = {2022},
  doi = {10.4230/LIPIcs.ITCS.2022.84},
  bibsource = {dblp computer science bibliography, https://dblp.org}
}

@inproceedings{GoldbergKLO22_coco_conf,
  author = {Goldberg, Halley and Kabanets, Valentine and Lu, Zhenjian and Oliveira, Igor Carboni},
  title = {{Probabilistic Kolmogorov Complexity with Applications to Average-Case Complexity}},
  booktitle = {Proceedings of the Computational Complexity Conference (CCC)},
  pages = {16:1--16:60},
  year = {2022},
  doi = {10.4230/LIPIcs.CCC.2022.16},
  bibsource = {dblp computer science bibliography, https://dblp.org}
}

@inproceedings{IlangoRS22_stoc_conf,
  author = {Ilango, Rahul and Ren, Hanlin and Santhanam, Rahul},
  title = {{Robustness of average-case meta-complexity via pseudorandomness}},
  booktitle = {Proceedings of the Symposium on Theory of Computing (STOC)},
  pages = {1575--1583},
  year = {2022},
  doi = {10.1145/3519935.3520051},
  bibsource = {dblp computer science bibliography, https://dblp.org}
}

@inproceedings{HaitnerMS23_innovations_conf,
  author = {Haitner, Iftach and Mazor, Noam and Silbak, Jad},
  title = {{Incompressiblity and Next-Block Pseudoentropy}},
  booktitle = {Proceedings of the Innovations in Theoretical Computer Science Conference (ITCS)},
  pages = {66:1--66:18},
  year = {2023},
  doi = {10.4230/LIPIcs.ITCS.2023.66},
  bibsource = {dblp computer science bibliography, https://dblp.org}
}

@article{GoldwasserM84_jcss_journals,
  author = {Goldwasser, Shafi and Micali, Silvio},
  title = {{Probabilistic Encryption}},
  journal = {J. Comput. Syst. Sci.},
  volume = {28},
  number = {2},
  pages = {270--299},
  year = {1984},
  doi = {10.1016/0022-0000(84)90070-9},
  bibsource = {dblp computer science bibliography, https://dblp.org}
}

@inproceedings{LuOZ22_icalp_conf,
  author = {Lu, Zhenjian and Oliveira, Igor Carboni and Zimand, Marius},
  title = {{Optimal Coding Theorems in Time-Bounded Kolmogorov Complexity}},
  booktitle = {Proceedings of the International Colloquium on Automata, Languages, and Programming (ICALP)},
  pages = {92:1--92:14},
  year = {2022},
  doi = {10.4230/LIPIcs.ICALP.2022.92},
  bibsource = {dblp computer science bibliography, https://dblp.org}
}

@inproceedings{Hirahara23_stoc_conf,
  author = {Hirahara, Shuichi},
  title = {{Capturing One-Way Functions via NP-Hardness of Meta-Complexity}},
  booktitle = {Proceedings of the Symposium on Theory of Computing (STOC)},
  pages = {1027--1038},
  year = {2023},
  doi = {10.1145/3564246.3585130},
  bibsource = {dblp computer science bibliography, https://dblp.org}
}

@inproceedings{HiraharaN23_focs_conf,
  author = {Hirahara, Shuichi and Nanashima, Mikito},
  title = {{Learning in Pessiland via Inductive Inference}},
  booktitle = {Proceedings of the Symposium on Foundations of Computer Science (FOCS)},
  pages = {447--457},
  year = {2023},
  doi = {10.1109/FOCS57990.2023.00033},
  bibsource = {dblp computer science bibliography, https://dblp.org}
}

@book{0023825_daglib_books,
  author = {Salomon, David and Motta, Giovanni},
  title = {{Handbook of Data Compression {(5.} ed.)}},
  publisher = {Springer},
  year = {2010},
  doi = {10.1007/978-1-84882-903-9},
  isbn = {978-1-84882-902-2},
  bibsource = {dblp computer science bibliography, https://dblp.org}
}

@inproceedings{LiuP23_crypto_conf,
  author = {Liu, Yanyi and Pass, Rafael},
  title = {{One-Way Functions and the Hardness of (Probabilistic) Time-Bounded Kolmogorov Complexity w.r.t. Samplable Distributions}},
  booktitle = {Proceedings of the International Cryptology Conference (CRYPTO)},
  pages = {645--673},
  year = {2023},
  doi = {10.1007/978-3-031-38545-2_21},
  bibsource = {dblp computer science bibliography, https://dblp.org}
}

@inproceedings{HiraharaILNO23_stoc_conf,
  author = {Hirahara, Shuichi and Ilango, Rahul and Lu, Zhenjian and Nanashima, Mikito and Oliveira, Igor C.},
  title = {{A Duality between One-Way Functions and Average-Case Symmetry of Information}},
  booktitle = {Proceedings of the Symposium on Theory of Computing (STOC)},
  pages = {1039--1050},
  year = {2023},
  doi = {10.1145/3564246.3585138},
  bibsource = {dblp computer science bibliography, https://dblp.org}
}

@inproceedings{LiuP23_tcc_conf,
  author = {Liu, Yanyi and Pass, Rafael},
  title = {{On One-Way Functions and Sparse Languages}},
  booktitle = {Proceedings of the Theory of Cryptography Conference (TCC)},
  pages = {219--237},
  year = {2023},
  doi = {10.1007/978-3-031-48615-9_8},
  bibsource = {dblp computer science bibliography, https://dblp.org}
}

@phdthesis{Lee06_PhD_Thesis,
  title = {{Kolmogorov Complexity and Formula Size Lower Bounds}},
  author = {Lee, Troy},
  year = {2006},
  school = {University of Amsterdam}
}

@inproceedings{MazorP24_coco_conf,
	author       = {Noam Mazor and
	Rafael Pass},
	title        = {Search-To-Decision Reductions for Kolmogorov Complexity},
	booktitle = {Proceedings of the Computational Complexity Conference (CCC)},
	pages        = {34:1--34:20},
	year         = {2024},
	doi          = {10.4230/LIPICS.CCC.2024.34},
	bibsource    = {dblp computer science bibliography, https://dblp.org}
}

@inproceedings{HiraharaKLO24_coco_conf,
	author       = {Shuichi Hirahara and
	Valentine Kabanets and
	Zhenjian Lu and
	Igor C. Oliveira},
	title        = {Exact Search-To-Decision Reductions for Time-Bounded Kolmogorov Complexity},
	booktitle = {Proceedings of the Computational Complexity Conference (CCC)},
	pages        = {29:1--29:56},
	year         = {2024},
	doi          = {10.4230/LIPICS.CCC.2024.29},
	bibsource    = {dblp computer science bibliography, https://dblp.org}
}

@article{levin74laws,
	title = {Laws of information conservation (nongrowth) and aspects of the foundation of probability theory},
	author = {Levin, Leonid A.},
	journal = {Problemy Peredachi Informatsii},
	volume = {10},
	number = {3},
	pages = {30--35},
	year = {1974},
	publisher = {Russian Academy of Sciences},
	_bib2doi_finished = {true},
}

\end{document}